\documentclass[11pt]{article} 
\usepackage{stmaryrd}
\usepackage{amsmath}
\usepackage{graphicx}
\usepackage{color}
\usepackage[title]{appendix}
\usepackage[numbers,square]{natbib}
\usepackage{amsfonts,amscd,amssymb, mathtools,mathrsfs, dsfont} 

\usepackage{pgfplots}
\usepackage{float}

\usepackage{url}
\usepackage{hyperref}

\usepackage{array}
\usepackage{setspace}

\usepackage{subcaption}

\newif\ifblog
\newif\iftex
\blogfalse
\textrue
\newcommand{\citeauth}[1]{\citeauthor{#1} \cite{#1}}

\newcommand{\thmref}[1]{Theorem~{\rm \ref{#1}}}
\newcommand{\lemref}[1]{Lemma~{\rm \ref{#1}}}
\newcommand{\corref}[1]{Corollary~{\rm \ref{#1}}}
\newcommand{\propref}[1]{Proposition~{\rm \ref{#1}}}
\newcommand{\defref}[1]{Definition~{\rm \ref{#1}}}

\newcommand{\secref}[1]{Section~{\rm \ref{#1}}}
\newcommand{\appref}[1]{Appendix~{\rm \ref{#1}}}

\newcommand{\be}{\begin{eqnarray}}
\newcommand{\ee}{\end{eqnarray}}
\newcommand{\bee}{\begin{eqnarray*}}
\newcommand{\eee}{\end{eqnarray*}}

\newcommand{\cc}{\overline{c}}

\def\d{\:{\rm d}}
\def\dt{\:{\rm d}t}
\def\ds{\:{\rm d}s}
\def\dx{\:{\rm d}x}

\def\dw{\:{\rm d} W}
\def\P{{\mathbb P}}
\def\E{{\mathbb E}}

\def\R{{\mathbb R}}
\def\N{{\mathbb N}}

\def\R{{\mathbb R}}

\def\Q{{\cal Q}}
\def\bQ{{\cal Q}_{0}}

\def\p{{\partial}}

\newcommand{\cF}{{\cal F}}
\newcommand{\ep}{\varepsilon}
\newcommand{\al}{\alpha}

\newcommand{\nd}{\noindent}

\newcommand{\la}{\lambda}

\newcommand{\ga}{\gamma}
\newcommand{\si}{\sigma}

\newcommand{\argmax}{\mathop{\rm argmax}\limits}

\newtheorem{theorem}{Theorem}[section]
\newtheorem{lemma}[theorem]{Lemma}
\newtheorem{definition}{Definition}[section]
\newtheorem{corollary}[theorem]{Corollary}
\newtheorem{proposition}[theorem]{Proposition}

\newenvironment{proof}{\noindent{\sc Proof:}}{\strut\hfill $\Box$\\}

\newcommand{\LL}{{\cal L}}
\newcommand{\TT}{{\cal T}}

\def\wx{\widehat{x}}

\def\ov{\overline{v}}

\def\os{\overline{s}}
\def\us{\underline{s}}

\def\Ds{\Delta s}
\def\Dc{\Delta c}
\def\A{{\cal A}}

\def\Dmax{{\cal D}_{\max}}
\def\Dmin{{\cal D}_{\min}}

\def\sw{\mathscr{X}}
\def\sd{\mathscr{Y}}
\def\swb{\mathfrak{M}}
\def\minswb{\mathfrak{m}}

\def\SS{{\cal S}}
\def\NS{{\cal NS}}

\def\DD{\mathcal{C}}
\def\RM{\mathcal{M}}

\renewcommand{\geq}{\geqslant}
\renewcommand{\leq}{\leqslant}

\newcommand{\nn}{\nonumber}

\newcounter{spb}
\graphicspath{{new_figure/}}

\title{Optimal dividend payout with path-dependent drawdown constraint}

\author{
Chonghu Guan\thanks{School of Mathematics, Jiaying University, Meizhou 514015, Guangdong,
China. 
Email: \url{gchonghu@163.com}}
\and
Jiacheng Fan\thanks{Department of Applied Mathematics, The Hong Kong Polytechnic University, Kowloon, Hong Kong SAR, China.  
Email: \url{jiacheng.fan@polyu.edu.hk}}
\and Zuo Quan Xu\thanks{Department of Applied Mathematics, The Hong Kong Polytechnic University, Kowloon, Hong Kong SAR, China.  
Email: \url{maxu@polyu.edu.hk}}
}
\date{}

\begin{document}
\maketitle

\begin{abstract} 
This paper studies an optimal dividend problem with a drawdown constraint in a Brownian motion model, requiring the dividend payout rate to remain above a fixed proportion of its historical maximum. This leads to a path-dependent stochastic control problem, as the admissible control depends on its own past values. The associated Hamilton-Jacobi-Bellman (HJB) equation is a novel two-dimensional variational inequality with a gradient constraint, a type of problem previously only analyzed in the literature using viscosity solution techniques. In contrast, this paper employs delicate PDE methods to establish the existence of a strong solution. This stronger regularity allows us to explicitly characterize an optimal feedback control strategy, expressed in terms of two free boundaries and the running maximum surplus process. Furthermore, we derive key properties of the value function and the free boundaries, including boundedness and continuity. Numerical examples are provided to verify the theoretical results and to offer new financial insights.

\bigskip
\nd {\bf Keywords.} Optimal dividend payout; drawdown constraint; path-dependent constraint; free boundary problem; gradient constraint;

\bigskip
\nd {\bf 2010 Mathematics Subject Classification.} 35R35; 35Q93; 91G10; 91G30; 93E20.

\end{abstract}

\section{Introduction}
There is no denying the fact that dividend payout policy plays a fundamental role in every company's risk management operations. This policy not only affects a company's appeal to its shareholders but also critically impacts its financial stability in preparing for potential adverse events. Dating back to \citeauth{de1957impostazione} and \citeauth{gerber1969entscheidungskriterien}, the optimal distribution of dividends from a dynamic surplus process has attracted considerable attention and continues to be a prominent topic in actuarial science and financial mathematics. At its core, the optimal dividend payout problem views the surplus process as available capital for investment in an insurance portfolio, where the insurer seeks the best strategy to maximize the expected total of discounted future dividends before ruin—that is, when the surplus reaches zero or negative. This framework highlights the essential trade-off between providing attractive returns to shareholders and maintaining sufficient reserves to prevent or postpone bankruptcy.

The literature features numerous variants of this problem, exploring different models of the underlying surplus process (e.g., compound Poisson models in \citeauth{gerber2006optimal} and \citeauth{albrecher2020optimal}; Brownian motion models in \citeauth{asmussen1997controlled}, \citeauth{gerber2004optimal}, and \citeauth{azcue2005optimal}; jump-diffusion models in \citeauth{belhaj2010optimal}; and stochastic profitability models in \citeauth{reppen2020optimal}) as well as various constraints on dividend policies (e.g., \citeauth{KRS21} examined capital injections, while \citeauth{avanzi2009strategies} provides a comprehensive overview).

In this paper, we concentrate on an optimal dividend problem incorporating a \textit{drawdown} constraint, where the manager aims to find the best payout policy while ensuring the dividend rate stays at least a fixed proportion of its historical maximum. In essence, a drawdown constraint represents a form of habit formation in policy decisions, using the running maximum of past payouts as the reference point. This concept has been widely explored in mathematical finance. For instance, \citeauth{dybvig1995dusenberry} introduced an infinite-horizon portfolio selection model with consumption ratcheting, preventing any decline in the agent's living standard. Similar issues involving ratcheting constraints on wealth appear in \citeauth{roche2006optimal} and \citeauth{elie2008optimal}. Within expected utility frameworks for consumption optimization, \citeauth{arun2012merton}, \citeauth{angoshtari2019optimal}, and \citeauth{deng2022optimal} extended \citeauth{dybvig1995dusenberry}'s model by incorporating drawdown constraints that allow decreases but not below a fixed fraction of the maximum. \citeauth{jeon2018portfolio} and \citeauth{albrecher2022optimal} analyzed finite-horizon effects in comparable consumption problems under ratcheting or drawdown constraints. More recently, \citeauth{angoshtari2022optimal} adapted the model to use an exponentially weighted average of past consumption for habit formation. In the context of dividend payouts, \citeauth{albrecher2020optimal} and \citeauth{albrecher2022optimal} investigated ratcheting constraints for Brownian and compound Poisson surplus processes, defining admissible policies as non-decreasing and right-continuous processes; however, they obtained only viscosity solutions without explicit optimal strategies. \citeauth{GX24} addressed ratcheting constraints using PDE methods to derive an optimal feedback strategy, while \citeauth{reppen2020optimal} applied viscosity solutions to stochastic profitability models. Most recently, \citeauth{albrecher2023optimal} incorporated drawdown constraints similar to ours and established optimality conditions for ``two-curve" strategies via viscosity solutions.

From a technical standpoint, introducing drawdown constraints transforms the problem into a two-dimensional optimal control challenge. Unlike traditional variational inequalities with functional constraints (such as in American option pricing), the resulting Hamilton-Jacobi-Bellman (HJB) equation introduces a novel variational inequality featuring a gradient constraint on the value function relative to the running maximum payout rate. This differs from HJB equations with gradient constraints in portfolio selection problems, like the transaction cost models in \citeauth{DY09} and \citeauth{DXZ10}, as our operator lacks derivatives with respect to the running maximum, preventing reduction to a one-dimensional inequality. Consequently, standard approaches like penalty methods prove ineffective. To our knowledge, no universal method exists for such stochastic control problems. For example, \citeauth{albrecher2022optimal} tackled ratcheting constraints by deriving similar variational inequalities and used discretization of dividend policies to prove viscosity solutions and convergence. In this work, we retain the Brownian motion assumption from \citeauth{albrecher2022optimal} but incorporate drawdown constraints that permit controlled declines in payouts. Building on \citeauth{GX24}, we develop a novel PDE approach that is more general than viscosity methods. We discretize the running maximum parameter equidistantly to create a regime-switching system of n-state variational inequalities, solvable as ordinary differential equations (ODEs). By letting n approach infinity, we approximate the continuous case and establish, through classical PDE theory, the existence, uniqueness, and continuity of $W_p^2$ strong solutions via uniform norm estimates. Unlike the linear operator in \citeauth{GX24}, our nonlinear operator requires new arguments. Compared to viscosity solutions, our method's key advantage lies in enabling the derivation of an explicit optimal feedback strategy, expressed through two free boundaries tied to the HJB equation and running maximum process. We also establish properties like boundedness and continuity for the value function and boundaries. 

A detailed comparison with \cite{albrecher2023optimal} and \cite{GX24} is given in \secref{sec:comparison}. Numerical results validate our theoretical findings. They demonstrate excellent agreement with established benchmarks in limiting cases and existing literature across various parameter settings. Moreover, the analysis reveals key insights into parameter effects. Stricter drawdown constraints reduce shareholder value by limiting dividend flexibility. In contrast, higher profitability and interest rates enable more aggressive dividend policies through lower surplus thresholds. Increased volatility promotes conservative strategies by shifting boundaries rightward. Larger maximum payout ratios enhance value functions, with the strongest effects at higher surplus levels.

The remainder of the paper is organized as follows. In \secref{sec:pf}, we formulate an optimal dividend payout problem subject to a path-dependent drawdown constraint. 
\secref{sec:ceiling} focuses on the ceiling model, in which the payout rate is upper bounded by a constant. Section \ref{sec:boundary} provides an explicit solution to the boundary case. \secref{sec:HJB} is devoted to the study of the non-boundary case. We first use delicate PDE techniques to solve the HJB equation in a proper space and then provide a complete answer to the problem by a verification argument. The properties of two free boundaries are thoroughly investigated. \secref{sec:no ceiling} investigates the no-ceiling model, where the payout rate can be arbitrarily large. Notably, some results and proofs established for the ceiling model do not directly apply to the no-ceiling case, necessitating novel analytical approaches. Nevertheless, we find that the solutions for both models exhibit similar qualitative properties. Numerical examples are presented in \secref{sec:numerical} to verify our theoretical results and give some new financial insights that are not proved. Finally, 
we compare our model and approach with \cite{albrecher2023optimal} and \cite{GX24} in \secref{sec:comparison}.


\section{Model formulation}\label{sec:pf} 

Throughout this paper, we fix a probability space $(\Omega,\mathcal{F},\P)$ satisfying the usual assumptions, along with a standard one-dimensional Brownian motion $\{W_t\}_{t\geq 0}$. Let $\{\cF_t\}_{t\geq 0}$ denote the filtration generated by the Brownian motion, argumented with all $\P$-null sets.
Let $\R^+=(0,\infty)$ denote the set of positive real numbers and $\N=\{1,2,\cdots\}$ the set of natural numbers.
We use $\vert Z\vert$ to denote $\sqrt{\mathrm{trace}(ZZ')}$ for any (column) vector or matrix $Z$.

We consider a representative insurer (He) and model its surplus process by a diffusion model. Then the surplus process $\{X_t\}_{t\geq 0}$ after paying dividends follows
\begin{equation}\label{X_eq}
{\rm d}X_t=(\mu-\DD_t) \dt+\sigma \dw_t,\;\; t\geq 0,
\end{equation}
where $\mu$ and $\sigma>0$ are constants, standing respectively for the return rate and volatility rate of the surplus process, $\DD_t$ is the dividend payout rate of the insurer at time $t$, which of course is an $\{\cF_t\}_{t\geq 0}$ adapted process.
As usual, we define the ruin time of the insurer as
$$\tau:=\inf\big\{t\geq 0 \;\Big|\; X_t\leq 0\big\},$$
with the convention that $\inf\emptyset=+\infty$.
We emphasis that the ruin time $\tau$ always depends on the insurer's dividend payout strategy $\{\DD_t\}_{t\geq 0}$ throughout the paper.

Given a (dividend payout) strategy $\{\DD_t\}_{t\geq 0}$,
we define the corresponding running/historical maximum payout process $\{\RM_t\}_{t\geq 0}$ as
$$\RM_t=\sup\limits_{0\leq s\leq t}\DD_s.$$
We fix a maximum allowed drawdown proportion $b\in[0,1]$. 
We first focus on the ceiling model in 
\secref{sec:ceiling}, in which the payout rate is upper bounded by a constant $\cc<\infty$. This foundational case presents the core challenge of our work. 
A dividend payout strategy $\{\DD_t\}_{t\geq 0}$ is called admissible if it is right-continuous with left-limit and satisfies
\begin{align}\label{admissiblestrategy}
b \RM_t \leq \DD_t\leq \cc,~~0\leq t\leq \tau.
\end{align}
Given an initial value $\RM_0=c\in[0, \cc]$,
let $ \Pi_{[c,\cc]}$ denote the set of all the corresponding admissible dividend payout strategies.

In the ceiling model, one's objective is to find an admissible dividend payout strategy $\DD^*=\{\DD^*_t\}_{t\geq 0}\in \Pi_{[c,\cc]}$ to maximize the discounted cumulated dividend until the ruin time. Mathematically, his aim is to determine
\begin{align}\label{value}
V(x,c)=\sup\limits_{\DD\in \Pi_{[c,\cc]}} \E\bigg[\int_0^\tau e^{-r t} \DD_t \dt \;\bigg|\; X_0=x,~\RM_0=c \bigg],~~ (x,c)\in \Q:=[0,+\infty)\times [0,\cc],
\end{align}
where $r>0$ is a constant discount factor. Since $0\leq \DD_t\leq \cc$ and $0\leq \tau\leq \infty$, the value function $V$ is clearly nonnegative and upper bounded by $\cc r^{-1}$.
On the other hand, if the initial surplus is $X_0=0$, then the ruin time $\tau$ is 0, so we have $V(0,c)=0$ for all $c\in[0,\cc]$.

When $b=0$, the problem \eqref{value} can be easily solved by dynamic programming. Whereas when $b=1$, the admissible dividend payout strategies must be non-decreasing over time, and the problem \eqref{value} becomes the so-called dividend ratcheting problem, which has been studied in Albrecher et al. \cite{albrecher2020optimal} and \cite{albrecher2022optimal} for different surplus models by viscosity solution technique and eventually completely solved by Guan and Xu \cite{GX24} by PDE method.

Clearly, the major difficulty in solving the dividend payout problem \eqref{value} comes from the dividend payout constraint \eqref{admissiblestrategy}. Except for $b=0$, the constraint is path-dependent, that is, the current and future admissible strategies depend on the whole past values. There seems no uniform way to deal with this type of stochastic control problems in the literature. In this paper, following \cite{GX24}, we adopt a PDE argument plus stochastic analysis method to tackle the problem \eqref{value}.

We will investigate the no-ceiling model in \secref{sec:no ceiling}, where the payout rate may be arbitrarily large, i.e., $\cc = \infty$. To extend our results to this case, we employ a limiting argument. However, a lot of key estimates from the ceiling model depend implicitly on $\cc$. Consequently, some proofs and conclusions do not directly generalize, necessitating novel arguments to address these challenges.

\section{The ceiling model: $\cc<\infty$ }\label{sec:ceiling}
\subsection{The boundary case: $V(x,\cc)$} \label{sec:boundary}
Before solving the general case, we first need to determine the boundary value $V(x,\cc)$ since it is needed to write out the HJB equation for the problem \eqref{value}.
We emphasis that, for a usual optimal control or stopping problem, the boundary condition is usually trivial to obtain from the problem formulation. For instance, the value function on the boundary is usually nothing but the payoff function of an American option in an option pricing problem.
But for our problem this is not the case; we cannot determine $V(x,\cc)$ from the problem formulation immediately. Indeed, the problem itself is a stochastic control problem. This is very different from Guan and Xu \cite{GX24} where the boundary value is almost trivial to obtain.

Although it is not immediately to know the optimal value $V(x,\cc)$, the boundary case is still relatively easy to solve, because the path-dependent constraint \eqref{admissiblestrategy} in this case becomes path-independent.
We are not only be able to give the optimal value in this case but also able to provide an explicit optimal dividend payout strategy for the problem \eqref{value}.

Indeed, in the boundary case, it is easy to see that a strategy $\{\DD_t\}_{t\geq 0}$ is an admissible strategy in $\Pi_{[\cc,\cc]}$ if and only if it satisfies $b \cc \leq \DD_t\leq \cc$ for all $t\geq 0$. This is because all admissible strategies are upper bounded by $\cc$, and consequently the running maximum payout process remains to be a constant, i.e. $\RM_t\equiv \cc$, from the beginning, so the constraint set becomes time invariant.
Mathematically speaking,
\begin{align} \label{boundary}
V(x,\cc)&=\sup\limits_{\DD\in \Pi_{[\cc,\cc]}} \E\bigg[\int_0^\tau e^{-r t} \DD_t \dt \;\bigg|\; X_0=x,~\RM_0=\cc \bigg]\nonumber\\
&=\sup\limits_{b \cc \leq \DD_t\leq \cc} \E\bigg[\int_0^\tau e^{-r t} \DD_t \dt \;\bigg|\; X_0=x\bigg],~~ x\geq 0.
\end{align}
Since the above problem \eqref{boundary} only depends on the initial state $x$, its HJB equation must be an ODE, therefore, one would expect to solve it completely. By contrast, in the general case, the HJB equation to \eqref{value} has an additional argument representing the running maximum, so it seems too ambitious to get an explicit value function or optimal dividend payout strategy. 

To get an explicit value function of the boundary problem \eqref{boundary}, we start from the following verification result, which shows that if a function is sufficiently good (i.e., satisfying the hypothesis of Proposition \eqref{prop:boundary}), then it is the value function of \eqref{boundary}. We also provides an optimal dividend payout strategy for the problem \eqref{boundary}.

The following two operators $\LL$ and $\TT$ will be used throughout this paper:
\begin{align*}
\LL v:=&~\frac{1}{2}\si^2v_{xx}+\mu v_x-r v,~~\\
\TT v:=&~\max_{b\leq d\leq 1}d(1-v_x)=b(1-v_x)+(1-b)(1-v_x)^+.
\end{align*}
Note that, unless $b=1$, $\TT$ in general is a nonlinear operator. This is the critical difference between our problem and that in \cite{GX24} where $b=1$ so that $\TT$ is a linear operator. The nonlinearity of $\TT$ not only brings tremendous difficulties to our mathematical analysis, but also change the fundamental structure of the optimal strategy.
We will face two free boundaries whereas there is only one in \cite{GX24}.

\begin{proposition}[Verification for the boundary case]\label{prop:boundary}
Suppose $g\in C^2(\R^+)$ satisfies the following ODE:
\begin{align}\label{Lg}
	\begin{cases}
		-\LL g-\cc \TT g=0, & x\in\R^{+},\\[3mm]
		g(0)=0.
	\end{cases}
\end{align}
If both $g$ and $g'$ are bounded, then $g$ is the value function of the problem \eqref{boundary}, i.e., $$V(x,\cc)=g(x),~~ x\in\R^+.$$
Moreover, an optimal dividend payout strategy is given by $\{\DD^*_{t}\}_{t\geq0}$, where
\begin{align}\label{boundarycontrol}
	\DD^*_t=
	\begin{cases}
		\cc & \hbox{\; if\; $g'(X_t)\leq 1$};\\
		b\cc & \hbox{\; otherwise.\;}
	\end{cases}
\end{align}
\end{proposition}
\begin{proof}
We put it in \appref{proofprop:boundary}.
\end{proof}

By the above result, we need to find a solution $g$ to \eqref{Lg}. 
To this end, we first define the following parameters used throughout this paper:
\begin{align*}
\ga&=\frac{\sqrt{(\mu-\cc)^2+2 \si^2 r}+(\mu-\cc)}{\si^2}, &
\la_{1,2}&=\frac{\sqrt{(\mu-b\cc)^2+2 \si^2 r}\mp (\mu-b\cc)}{\si^2},\\
k_1&=\frac{1}{\la_1+\la_2}\left[1+\la_2\left(\frac{(1-b)\cc}{r}-\frac{1}{\ga}\right)\right], & k_2&=\frac{1}{\la_1+\la_2}\left[1-\la_1\left(\frac{(1-b)\cc}{r}-\frac{1}{\ga}\right)\right].
\end{align*}
Indeed, $\ga$, $\la_{1}$ and $\la_{2}$ are, respectively, the unique positive roots for the quadratic equations:
\begin{align*}
-\frac{1}{2}\si^2 \ga^2+(\mu-\cc)\ga+r &=0,\\
-\frac{1}{2}\si^2 \la_{1}^2-(\mu-b\cc) \la_{1}+r &=0,\\
-\frac{1}{2}\si^2 \la_{2}^2+(\mu-b\cc) \la_{2}+r &=0.
\end{align*}
We will use the simple fact that $\ga,\la_1,\la_2>0$ and $\la_1 k_1+\la_2 k_2=1$ frequently in our subsequent argument without claim.

\begin{lemma}\label{lem:g}
If $2\mu\cc \leq \si^2 r$, we then define
\[g (x)=\frac{\cc}{r}(1-e^{-\ga x}),~~ x\geq 0.\]
If $2\mu\cc > \si^2 r$, then
there is a unique positive number $y_0$ such that $$k_1 e^{-\la_1 y_0}-k_2 e^{\la_2 y_0}+\frac{b\cc}{r}=0,$$
and we define
\begin{align}\label{g_def}
	\displaystyle g (x)=
	\begin{cases}
		\displaystyle k_1 e^{\la_1(x-y_0)}-k_2 e^{-\la_2(x-y_0)}+\frac{b\cc}{r}, & 0\leq x\leq y_0; \bigskip\\
		\displaystyle -\frac{1}{\ga} e^{-\ga(x-y_0)}+\frac{\cc}{r}, & x>y_0. \smallskip
	\end{cases}
\end{align}
Then $g\in C^2(\R^+)$ satisfies \eqref{Lg}, $0\leq g\leq \cc/r$,
$0\leq g'\leq g'(0)<\infty$ and $g'$ is strictly decreasing (thus $g$ is strictly concave) on $\R^{+}$.
Moreover, $g'(y_0)=1$ in the second case.
\end{lemma}

\begin{proof}
	We put it in \appref{sec:g}
\end{proof}

\begin{figure}[htp]
\centering
\begin{minipage}[t]{\textwidth}
	\centering
	\includegraphics[width=0.35\textwidth]{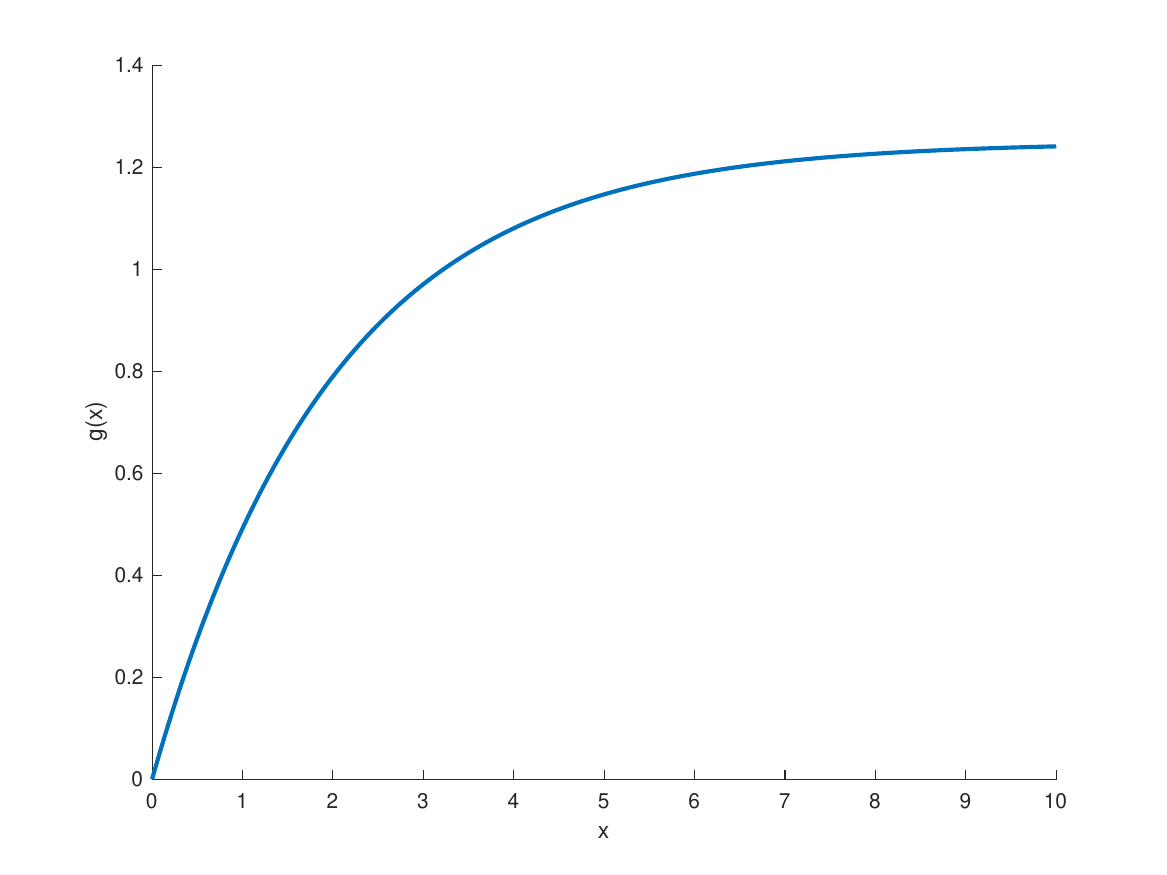}
	\caption{$2\mu \cc < \sigma^2 r$, $\mu=0.1$, $r=0.08$, $\sigma=0.8$ and $\bar{c}=0.1$}
	\label{g1}
\end{minipage}
\end{figure}

\begin{figure}[htp]
\centering
\begin{minipage}[t]{\textwidth}
	\centering
	\includegraphics[width=0.35\textwidth]{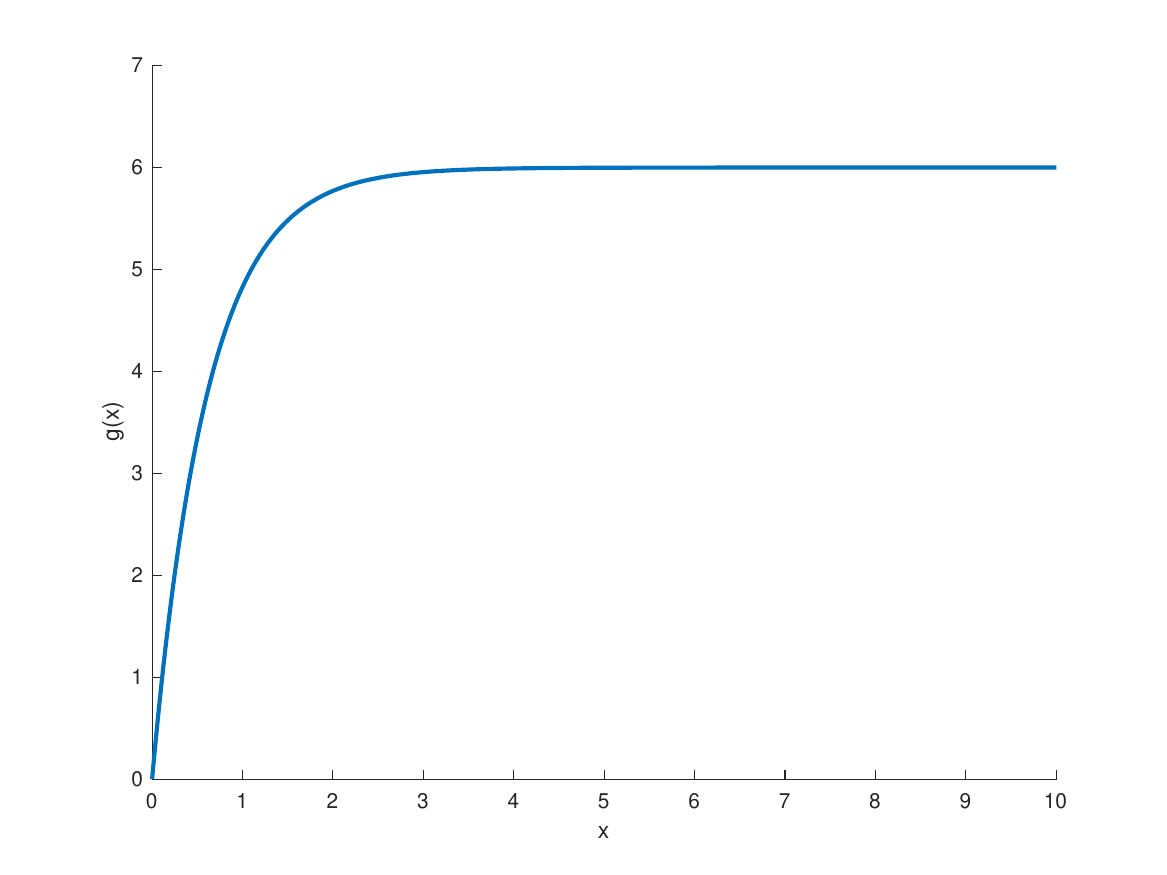}
	\caption{$2\mu \cc >\sigma^2 r$, $\mu=0.4$, $r=0.05$, $\sigma=0.4$ and $\bar{c}=0.3$ and $b=1$}
	\label{g2}
\end{minipage}
\end{figure}

We will use $g$ and $y_{0}$ defined above throughout the paper.
The pictures of $g$ are shown in Figures \ref{g1} and \ref{g2} under different parameters.

Combining \propref{prop:boundary} and \lemref{lem:g}, we solved the boundary case completely.
\begin{theorem}[Solution for the boundary case]
\label{op1}
The optimal value to the problem \eqref{boundary} is $g$ defined in \lemref{lem:g}, i.e., $V(x,\cc)=g(x)$, $x\in\R^+$, 
which, in particular, implies that $V(x,\cc)$ is concave in $x$.
Moreover, an optimal dividend payout strategy is given by $\{\DD^*_{t}\}_{t\geq0}$, where $\DD^*_t\equiv\cc$ if $2\mu\cc \leq \si^2 r$; and otherwise,
\begin{align}\label{boundarycontrol2}
	\DD^*_t=
	\begin{cases}
		\cc & \hbox{\; if\; $X_t\geq y_0$};\\
		b\cc & \hbox{\; otherwise.\;}
	\end{cases}
\end{align}
\end{theorem}

Economically speaking, since one cannot increase the value of the running maximum in the boundary case, so it is fixed.
In this case, one should pay dividends at the maximum rate if the market risk is relative high (i.e., $2\mu\cc/r \leq \si^2$) or equivalent the maximum rate is relatively small (i.e., $ \cc \leq \si^2r/(2\mu)$). The optimal strategy does not relay on the surplus level.
On the other hand, if the market risk is low (i.e., $2\mu\cc/r > \si^2$) or equivalent the maximum rate is relatively big (i.e., $ \cc> \si^2r/(2\mu)$), then the optimal strategy depends on the surplus level.
One should pay dividends either at the maximum rate $\cc$ if the surplus level is relative high (i.e., $X_t\geq y_0$), or at the minimum possible rate $b\cc$ if the surplus level is relative low (i.e., $X_t<y_0$).

We now turn to the general case, which is divided into a simple case: $2\mu\cc \leq \si^2 r$, and the reminder complicated case: $2\mu\cc > \si^2 r$. In the former case, we will give an explicit optimal value and optimal strategy. In the latter case, the problem becomes complicated and no explicit solution or strategy is available.

\subsection{The non-boundary case: $V(x,c)$}\label{sec:HJB}

We first deal with the simple case: $2\mu\cc \leq \si^2 r$.
In the absence of the drawdown constraint \eqref{admissiblestrategy}, it is known that the constant payout strategy $\DD^*_t\equiv\cc$ is optimal to the problem \eqref{value} (see \cite{albrecher2023optimal}). Since this policy also satisfies the drawdown constraint, it remains optimal even when the constraint is imposed. We provide a simply proof in \secref{sec:op2} for the easy of reading. 
\begin{theorem} 
[Solution for the simple case $2\mu\cc \leq \si^2 r$]
\label{op2}
If $2\mu\cc \leq \si^2 r$, then the optimal value to the problem \eqref{value} is $$V(x,c)=\frac{\cc}{r}(1-e^{-\ga x}),~~ (x,c)\in \Q,$$
and an optimal dividend payout strategy is given by $\{\DD^*_{t}\}_{t\geq0}$, where $\DD^*_t\equiv\cc$.
\end{theorem}

We remark that the value function in the simple case is invariant with respect $c$, i.e.,
$V(x,c)=V(x,\cc)$ for all $(x,c)\in \Q$.

From now on we study the complicated case: $2\mu\cc> \si^2 r$, which is henceforth assumed. We will adopt a novel PDE argument plus stochastic analysis method to tackle the problem.
This provides the major contributions of this paper. 

Our analysis consists of several parts. We start with introducing the HJB equation for the problem \eqref{value}. The second step is the key of approach, where we introduce and study a regime switching system, a sequence of ODEs (where are indeed single-obstacle problems), whose solvability can be established easily by standard PDE theory. 
In the next step, we establish a lot of estimates for the solution of the regime switching system (See \appref{sec:approximation} for details) and construct a solution to the HJB equation of the original problem \eqref{value} from the solution of the regime switching system.
Also, we propose an optimal dividend payout strategy depending on some free boundaries arising from the HJB equation.
The last step is to verify that the constructed solution to the HJB equation is the value function and the proposed strategy is optimal.

We introduce the following variational inequality on $v(x,c)$:
\begin{align}\label{v_pb}
\begin{cases}
	\min\{-\LL v-c \TT v , \; -v_c\}=0, & \hbox{in }\; \bQ:=\R^+\times [0,\cc),\\[3mm]
	v(0,c)=0, & c\in [0,\cc),\\[3mm]
	v(x,\cc)=g (x), & x\geq 0.
\end{cases}
\end{align}
It is indeed the HJB equation for our problem \eqref{value}. We will show its solution is the value function of \eqref{value}. 
The exact meaning of its solution is given as follows.

Denote
\begin{align*}
\A=\Big\{v:\Q\mapsto \R \;\Big|\; &v\in C(\Q), ~\mbox{$v$ is non-increasing w.r.t. $c$,}
\; \\
&v(\cdot,c)\in W^2_{p, \rm loc} [0,+\infty)\cap L^\infty [0,+\infty) ~\mbox{for each}~ c\in [0,\cc] \Big\}.
\end{align*}

\begin{definition}[Strong solution]\label{def:solution}
We call $v$ is a strong solution to \eqref{v_pb} if the follows hold:
\begin{enumerate}
	\item $v\in \A$;
	\item $v$ satisfies the boundary conditions in \eqref{v_pb};
	\item for each $c\in[0,\cc]$,
	\begin{align}\label{-Lv>=0}
		-\LL v(\cdot,c)-c \TT v(\cdot,c)\geq 0 ~\hbox{a.e. in}~ \R^+;
	\end{align}
	\item
	if there is some $(x,c)\in \bQ$ such that $v(x,c)>v(x,s)$ holds for all $s\in(c,\cc]$, then \begin{align}\label{-Lv=0}
		-\LL v(y,c)-c \TT v(y,c)=0~\hbox{for all}~y\in (0,x).
	\end{align}
\end{enumerate}

\end{definition}

\thmref{thm:u} and \thmref{thm:averi} will provide a complete answer to the problem \eqref{value} in the complicated case.
The former resolves the solvability issue of the HJB equation \eqref{v_pb} and establishes some properties of its solution. The latter gives the optimal value and optimal strategy for the problem.

\begin{theorem}\label{thm:u}
The PDE \eqref{v_pb} admits a unique strong solution $v$ in the sense of \defref{def:solution}.
Furthermore, we have $ v_{x}$ is continuous in $\Q$ and $v$ is Lipschitz continuous w.r.t. $c$, also given any $p>1$ and $\al=1-1/p$, there is a constant $K>0$ such that
\begin{gather}
	0\leq -v_{c}\leq K ~\mbox{a.e.},\label{vc}\\
	\label{v} 0\leq v\leq \frac{\cc}{r},\\
	\label{vx} 0\leq v_x\leq K,
\end{gather}
and for any $(x,c)\in \Q$ and $N\in \N$,
\begin{align}\label{vxx}
	v_x(y,c)\leq \max\{v_x(x,c), 1\},~~ \forall\; y\geq x,
\end{align}
and
\begin{align}\label{vL}
	|v(\cdot,c)|_{W^2_p[N-1,N]}+|v(\cdot,c)|_{C^{1+\al}(\R^+)}\leq K.
\end{align}
\end{theorem}
\begin{proof}
The proof is given in \appref{sec:solvability}.
\end{proof} 

In the rest part of this section,
we let $v$ denote the unique solution to \eqref{v_pb}, given in \thmref{thm:u}.

To derive the optimal strategy, our initiative thinking is as follows.
On one hand, we hope to pay dividend as much as possible to increase the objective in \eqref{value}, the total discounted dividend payouts. But this may lead to two unfavorable results: first, it may increase the running maximum so that shrinking the set of admissible strategies given by \eqref{admissiblestrategy} in the future and decreasing the value function; second, the ruin time may come earlier. To avoid these negative impacts, our strategy is to increase the running maximum to some higher level only when the value functions coincide at these two levels. This motives us to define the following
{\bf switching region}
\begin{align*}
\SS=\Big\{(x,c)\in \bQ \;\Big|\; v(x,c)=v(x,s) \text{ for some $s\in(c,\cc]$} \Big\}.
\end{align*}
In this region, if we increase the running maximum from $c$ to $s$, the value function will not change. Since $v$ is non-increasing in $c$, the complement of $\SS$ in $\bQ$, called the {\bf non-switching region}, is given by
\begin{align*}
\NS=\Big\{(x,c)\in \bQ \;\Big|\; v(x,c)>v(x,s) \text{ for all $s\in(c,\cc]$}\Big\}.
\end{align*}
In this region, increasing the running maximum is harmful, so one should not do that; in other words, one should not pay dividends exceeding the running maximum.
In particular, we have \eqref{-Lv=0} holds if $(x,c)\in\NS.$

It is nature to conjecture that one shall increase the running maximum only when his wealth is high enough. This motives us to define the {\bf switching boundary} as follows:
\begin{equation}\label{switching_boundary}
\sw(c)=\inf\Big\{x\in \R^+\;\Big|\; (x,c)\in \SS\Big\},~~ c\in[0,\cc),
\end{equation}
with the convention that $\inf\emptyset=+\infty$. Clearly it holds that $\sw(c)\leq x$ if $(x,c)\in \SS$.
The following result shows that the reverse statement is almost true. 

\begin{proposition}\label{prop:characterizationx}
The switching boundary $\sw(\cdot)$ is positive, bounded and continuous on $[0,\cc)$.
The limit $\sw(\cc):=\lim_{c\to\cc-}\sw(c)$ exits and is finite.
Also, if $(x,c)\in \SS$, then $(y,c)\in \SS$ for any $y\geq x$. As a consequence,
the switching and non-switching regions $\SS$ and $\NS$ are separated by $\sw(\cdot)$:
\begin{gather*}
	\Big\{(x,c)\in \bQ \;\Big|\; x> \sw(c)\Big\}\subseteq\SS
	\subseteq\Big\{(x,c)\in \bQ \;\Big|\; x\geq \sw(c)\Big\},\\
	\Big\{(x,c)\in \bQ \;\Big|\; x< \sw(c)\Big\}\subseteq\NS
	\subseteq\Big\{(x,c)\in \bQ \;\Big|\; x\leq \sw(c)\Big\}. 
\end{gather*}
Moreover, 
\begin{gather}\label{increasing}
	\text{$\sw(\cdot)$ is strictly increasing if and only if
		$\big\{(x,c)\in \bQ \;\big|\; x=\sw(c)\big\} \subseteq\NS$.}\\
	\label{decreasing}
	\text{$\sw(\cdot)$ is non-increasing if and only if
		$\big\{(x,c)\in \bQ \;\big|\; x=\sw(c)\big\} \subseteq\SS$.}
\end{gather}
\end{proposition}
\begin{proof}
The proof is given in \appref{sec:freeboundaries}.
\end{proof}

Unfortunately, we do not know whether
$\big\{(x,c)\in \bQ \;\big|\; x=\sw(c)\big\}$
belongs to $\SS$ or $\NS$. 
Indeed, we conjecture that $\sw(\cdot)$ is strictly increasing, which is consistent with numerical examples in \secref{sec:numerical}. Sadly, we cannot confirm this in theory except for $b=1$ which is proved in \cite{GX24}.

In order to find the optimal strategy,
we define the {\bf equivalent maximum rate} as
\begin{align*}
\swb(x,c)=\max\Big\{s\in[0,\cc]\;\Big|\; v(x,s)=v(x,c)\Big\}\in[c,\cc],~~(x,c)\in \bQ,
\end{align*}
and the {\bf equivalent minimum rate} as
\begin{align*}
\minswb(x,c)=\min\Big\{s\in[0,\cc]\;\Big|\; v(x,s)=v(x,c)\Big\}\in[0,c],~~(x,c)\in \bQ.
\end{align*}
For $(x,c)\in \bQ$, the following facts can be easily verified:
\begin{itemize}
\item $\minswb(x,c)\leq c\leq\swb(x,c)$;
\item If $s\in[\minswb(x,c),\swb(x,c)]$, then $v(x, s)=v(x,c)$, $\minswb(x,s)=\minswb(x,c)$, $\swb(x,s)=\swb(x,c)$;
\item If $\minswb(x,c)<\swb(x,c)$, then $(x, s)\in\SS$ for $s\in[\minswb(x,c),\swb(x,c))$;
\item If $\minswb(x,c)=\swb(x,c)$, then $(x,c)\in\NS$;
\item If $\swb(x,c)<\cc$, then $(x, \swb(x,c))\in\NS$;
\item $(x,c)\in\SS$ if and only if $\swb(x,c)>c$;
\item $(x,c)\in\NS$ if and only if $\swb(x,c)=c$.
\end{itemize}

Given a wealth level $x$, we can increase the running maximum dividend rate $c$ to its equivalent maximum rate $\swb(x,c)$, which is the maximum dividend rate that will not reduce the optimal value.
In particular, if $\swb(x,c)=c$, then one should not increase the running maximum dividend rate at this wealth level $x$.
Also if $\swb(x,c)=\cc$, then the optimal value of the problem \eqref{value} is known since $v(x,c)=v(x,\cc)=g(x)$.

The following result, which characterizes the equivalent maximum rate, will play a critical role in finding an optimal strategy for the problem \eqref{value}.

\begin{proposition}\label{prop:xi+}
For each $c\in[0,\cc)$, $x\mapsto \swb(x,c)$ is non-decreasing and
$x\mapsto \minswb(x,c)$ is non-increasing on $\R^{+}$.
Also, for $(x,c)\in\bQ$, it holds that
\begin{align} 
	\label{valueatboundary}
	\sw\big(\swb(x,c)\big)&=x&&\hbox{if $\minswb(x,c)<\swb(x,c)<\cc$;}\qquad\smallskip\\
	\label{lvatdecreasingpoint}
	-\LL v\big(y,\swb(x,c)\big)-\swb(x,c) \TT v\big(y,\swb(x,c)\big)&=0&&\hbox{if }\; 0 < y < x;\smallskip\\
	\label{lvatincreasingpoint}
	-\LL v\big(y,\minswb(x,c)\big)-\minswb(x,c) \TT v\big(y,\minswb(x,c)\big)&=0&&\hbox{if }\; 0<y< x\mbox{ and }\minswb(x,c)>0;\smallskip\\
	\label{xi+} \lim\limits_{y\to x+} \swb(y,c)&=\swb(x,c);\smallskip\\
	\label{xi-} \lim\limits_{y\to x+} \minswb(y,c)&=\minswb(x,c);&& \smallskip\\
	\label{v_c=0} v_c(x, s)&=0&&\hbox{if } \minswb(x,c)<s\leq\swb(x,c)<\cc.
\end{align}
Also,
\begin{gather*}
	\SS=\Big\{(x,s) \;\Big|\; \minswb(x,c)\leq s<\swb(x,c),~(x,c)\in \bQ\Big\}.
\end{gather*}
\end{proposition}
\begin{proof}
The proof is given in \appref{sec:proofofprop:xi+}.
\end{proof}

\begin{figure}[H]
\begin{center} \vspace{0pt}
	
	\begin{tikzpicture}[line width=0.4pt, scale=0.86]
		\thicklines
		\draw[black,thick,->] (0, 0) -- (0, 8);
		
		\draw[black,thick,->] (0, 0) -- (12, 0);
		\draw[black,thick] (11, 0) -- (11, 7.5);
		\node[rectangle] at (11, -0.3) {$\cc$};
		\node[rectangle] at (-0.3, -0.3) {$0$};

		\draw[blue,dashed] (0, 5) .. controls (0.8, 2) and (1, 1.6) .. (2.3, 1.4);
		\node[red] at (2.3, 1.4) {$\bullet$};
		\draw[red,dashed] (2.3, 1.4) .. controls (3, 1.6) and (4, 2) .. (5.1, 4.2);
		\node[blue] at (5.1, 4.2) {$\bullet$};
		\draw[blue,dashed] (5.1, 4.2) .. controls (6.1, 2) and (6.8, 0.8) .. (7.5, 0.7);
		\node[blue] at (7.5, 0.7) {$\bullet$};
		\draw[blue,thick] (7.5, 0.7) -- (8.2, 0.7);
		\node[red] at (8.2, 0.7) {$\bullet$};
		\draw[red,dashed] (8.2, 0.7) .. controls (9, 0.6) and (10, 1.9) .. (11, 6.5);

		\node[rectangle] at (2.4, -0.3) {$c_{1}$};
		\node[rectangle] at (7.9, -0.3) {$c_{2}$};
		
		\node[rectangle] at (0.8, 6.8) {$\minswb(x_{4},c)$};
		\node[rectangle] at (10.1, 6.8) {$\swb(x_{4},c)$};
		\draw[blue,thick] (0, 6.5) -- (11, 6.5);
		\node[blue] at (0, 6.5) {$\bullet$};
		
		\node[blue] at (5.1, 5.75) {$\SS$};
		
		\node[red] at (5.1, 1.75) {$\NS$};

		\node[rectangle] at (0.8, 5.3) {$\minswb(x_{3},c)$};
		\node[rectangle] at (9.8, 5.3) {$\swb(x_{3},c)$};
		\node[blue] at (0, 5) {$\bullet$};
		\draw[blue,thick] (0, 5) -- (10.67, 5);
		\node[red] at (10.67, 5) {$\bullet$};
		\draw[red,thick] (10.67, 5) .. controls (10.8, 4.5) .. (11, 4.3);

		\node[rectangle] at (1.4, 3.8) {$\minswb(x_{2},c_{1})$};
		\node[rectangle] at (3.9, 3.8) {$\swb(x_{2},c_{1})$};
		
		\node[rectangle] at (6.8, 2.7) {$\minswb(x_{2},c_{2})$};
		\node[rectangle] at (8.9, 2.7) {$\swb(x_{2},c_{2})$};

		\node[red] at (0, 4) {$\bullet$};
		\draw[red,thick] (0, 4) .. controls (0.3, 3.6) .. (0.4, 3.5);
		\node[blue] at (0.4, 3.5) {$\bullet$};
		\draw[blue,thick] (0.4, 3.5) -- (4.75, 3.5);
		\node[red] at (4.75, 3.5) {$\bullet$};
		\draw[red,thick] (4.75, 3.5) .. controls (5, 2.8) .. (6, 2.4);
		\node[blue] at (6, 2.4) {$\bullet$};
		\draw[blue,thick] (6, 2.4) -- (9.8, 2.4);
		\node[red] at (9.8, 2.4) {$\bullet$};
		\draw[red,thick] (9.8, 2.4) .. controls (10, 1.8) .. (11, 1);

		\node[red] at (0, 2) {$\bullet$};
		\draw[red,thick] (0, 2.1) .. controls (0.2, 0.9) and (5, 0.8) .. (7.5, 0.7);
		\node[rectangle] at (6.8, 0.3) {$\minswb(x_{1},c_{2})$};
		\node[rectangle] at (9, 0.3) {$\swb(x_{1},c_{2})$};
		\draw[red,thick] (8.2, 0.7) .. controls (10, 0.6) and (10.5, 0.47) .. (11, 0.4);
		\node[blue] at (7.5, 0.7) {$\bullet$};

	\end{tikzpicture}
\end{center}
\caption{
	The value functions $c\mapsto v(x_{i},c)$ are in solid curve, from below to top corresponding to $x_{1}<x_{2}<x_{3}<x_{4}$.
	The dashed curve is $\sw(\cdot)$. Its blue parts stand for $\minswb(x,c_{1})$ and $\minswb(x,c_{2})$ at all different levels of $x$, whereas the red parts stand for $\swb(x,c_{1})$ and $\swb(x,c_{2})$.
	If a point belongs to $\SS$ (resp. $\NS$), then so does any point above (resp. below) it.
	Hence, all points above (resp. below) dashed curve belong to $\SS$ (resp. $\NS$).
	The blue (resp. red) parts of the dashed curve belong to $\SS$ (resp. $\NS$). Therefore, the blue parts of the solid curve are expanding as $x$ gets bigger and eventually the whole curve becomes blue above certain threshold (which is indeed $\sup_{c\in[0,\cc)}\sw(c)$).
	Although the dashed curve in above demonstration figure is $W$-shaped, we believe it should be increasing as $c$ gets bigger. Unfortunately, we cannot prove this, but numerical examples in \secref{sec:numerical} strongly support our conjecture.
}
\end{figure}
On the other hand, the objective is linear in the dividend payout, so one may expect a bang-bang dividend payout strategy: paying dividend at the maximum possible rate if the surplus is relatively high; at the minimum rate otherwise.
If the marginal utility is bigger than one, then one should keep the surplus as high as possible to wait for a better future, so one should save money and pay dividend at the minimum possible rate; otherwise, at the maximum rate.
The above thinking motives us to define the following
{\bf minimum dividend payout region}
\begin{align*}
\Dmin=\Big\{(x,c)\in \Q \;\Big|\; v_x(x,c)> 1\Big\}, 
\end{align*}
and the {\bf maximum dividend payout region}
\begin{align*}
\Dmax=\Big\{(x,c)\in \Q \;\Big|\; v_x(x,c)\leq 1\Big\}. 
\end{align*}
These two regions will be separated by the following {\bf converting boundary}
\begin{equation}\label{converting boundary}
\sd(c)=\inf\Big\{x\in\R^{+}\;\Big|\; v_x(x,c)\leq 1\Big\},~~ c\in[0,\cc].
\end{equation}
The following result characterizes the converting boundary.

\begin{proposition}\label{prop:characterizationy}
The converting boundary $\sd(\cdot)$ is continuous on $[0,\cc]$, and satisfies
$\sd(\cc)=y_0$,
$ v_x(\sd(c),c)=1$,
and
\begin{align}\label{y<x}
	0<\sd(c)<\sw(c),~~ c\in[0,\cc].
\end{align}
Also, it separates the minimum and the maximum dividend payout regions $\Dmin$ and $\Dmax$:
\begin{align*}
	\Dmin &=\Big\{(x,c)\in \Q \;\Big|\; x< \sd(c)\Big\},~~
	\Dmax=\Big\{(x,c)\in \Q \;\Big|\; x\geq \sd(c)\Big\}.
\end{align*}
\end{proposition}

Using the above results, we now give the optimal value and the optimal strategy for the problem \eqref{value} in the complicated case as follows.
\begin{theorem}[Optimal value and optimal strategy in the complicated case.]
\label{thm:averi}
Let $v$ be the unique solution to \eqref{v_pb}, given in \thmref{thm:u}. Then it coincides with the optimal value to the problem \eqref{value}.
Moreover, $\{\DD^*_{t}\}_{t\geq 0}$ is an optimal strategy for the problem \eqref{value} with $(x,c)\in\bQ$, where the triple $(X^*_t, \RM^*_t, \DD^*_t)$ is determined by
\[\begin{cases}
	\displaystyle{X^*_t=x+\int_0^t(\mu-\DD^*_s) \ds+\int_0^t\sigma \dw_s,}\\[5mm]
	\RM^*_t=\swb\Big(\max\limits_{s\in[0,t]}X^*_s,c\Big),
	\\[5mm]
	\DD^*_t=\Big(b+(1-b)1_{\{X^*_t\geq\sd(\RM^*_t)\}}\Big)\RM^*_t.
\end{cases}
\]
\end{theorem} 

From an economic perspective, this result indicates that dividend payouts should be adjusted according to the surplus level. Specifically, when the surplus is relatively low ($X^*_t < \sd(\RM^*_t)$), dividends should be paid at the minimum permitted rate, $b\RM^*_t$. If the surplus falls within the intermediate range ($\sd(\RM^*_t) \leq X^*_t \leq \sw(\RM^*_t)$), dividends should be paid at the historical maximum payout rate, $\RM^*_t$. Finally, when the surplus is sufficiently high ($X^*_t > \sw(\RM^*_t)$), dividends should exceed the historical maximum payout rate, $\RM^*_t$, thereby establishing a new historical maximum payout rate $\swb\Big(\max\limits_{s\in[0,t]}X^*_s,c\Big)$.


\section{The no-ceiling model: $\cc=+\infty$} \label{sec:no ceiling}

In this section, we investigate the no-ceiling model, where the payout rate may be arbitrarily large, i.e., $\cc = \infty$. 
We first derive the value function for $\cc = +\infty$ by taking the limit $\cc \to +\infty$. However, many key estimates from the ceiling model depend implicitly on $\cc$, preventing a direct extension of some proofs and conclusions to the no-ceiling case. Consequently, novel arguments are required to address these challenges. All proofs for the results in this section are deferred to \appref{proof:no ceiling}.

We denote by $v^{\cc}(x, c)$ the value function with the finite ceiling $\cc$ (that is \eqref{value}) on $c$ and by $v^{\infty}$ the no-ceiling problem.
Starting from the definition of the value function, it can be directly shown (see Property 3.3 and Property 3.5 in \cite{albrecher2023optimal}) that
$v^{\cc}(x, c)$ is non-decreasing w.r.t $\cc$ and converges to $v^{\infty}(x, c)$, meaning:
\begin{align}\label{inftyvv_def}
v^{\infty}(x, c)=\lim_{\cc \to +\infty}v^{\cc}(x, c),~~(x,c)\in \bQ^{\infty}.
\end{align}
and the following estimates hold 
\begin{gather}
\label{inftyvv}
x\leq v^{\infty}(x,c)\leq x+\mu/r,\\
\label{inftyvv_limc} \lim\limits_{c \to +\infty}v^{\infty}(x, c) = x.
\end{gather}

\subsection{The value function and HJB equation} \label{sec:vvalue}
Now, we introduce the variational inequality problem for the no-ceiling model:
\begin{align}\label{inftyvv_pb}
\begin{cases}
	\min\{-\LL v^{\infty}-c \TT v^{\infty} , \; -v^{\infty}_c\}=0, & \hbox{in }\; \bQ^{\infty}:=\R^+\times [0,+\infty),\\[3mm]
	v^{\infty}(0,c)=0, & c\in [0,+\infty),\\[3mm]
	\lim\limits_{c\to+\infty}v^{\infty}(x,c)=x, & x\geq 0,
\end{cases}
\end{align}
Because we cannot prove the no-ceiling value function satisfies $v^{\infty}(\cdot,c)\in W^2_{p,\rm loc}(\R^+)$ for each $c>0$, we need to relax the constraints on the solution and define its solution as in \defref{def:solution2}.

Denote 
\begin{align*}
\A^{\infty}_0=\Big\{v:\bQ^{\infty}\mapsto \R \;\Big|\; &v\in C(\bQ^{\infty}), ~\mbox{$v$ is non-increasing w.r.t. $c$,}
\; \\
&v(\cdot,c)\in C^1 [0,+\infty)\cap L^\infty [0,+\infty) ~\mbox{for each}~ c\in [0,+\infty) \Big\}.
\end{align*}

\begin{definition}[Weak solution]\label{def:solution2}
We call $v^{\infty}$ is a weak solution to \eqref{inftyvv_pb} if the follows hold:
\begin{enumerate}
	\item $v^{\infty}\in \A^{\infty}_0$;
	\item $v^{\infty}$ satisfies the boundary conditions for $x=0$ and the growth condition for $c\to +\infty$ in \eqref{inftyvv_pb};
	\item for each $c\in[0,+\infty)$,
	\begin{align}\label{infty-Lv>=0}
		-\LL v^{\infty}(\cdot,c)-c \TT v^{\infty}(\cdot,c)\geq 0 ~\hbox{weakly in}~ \R^+;
	\end{align}
	\item
	if there is some $(x,c)\in \bQ^{\infty}$ such that $v^{\infty}(x,c)>v^{\infty}(x,s)$ holds for all $s\in(c,+\infty)$, then \begin{align}\label{infty-Lv=0}
		-\LL v^{\infty}(y,c)-c \TT v^{\infty}(y,c)=0~\hbox{for all}~y\in (0,x).
	\end{align}
\end{enumerate}
\end{definition}

Similar to the ceiling model, we have the following result on the no-ceiling value function $v^{\infty}$. 

\begin{theorem}\label{inftythm:uno-ceiling}
The function $v^{\infty}$ defined in \eqref{inftyvv_def} is the unique weak solution, in the sense of \defref{def:solution2}, to the PDE \eqref{inftyvv_pb}.
Furthermore, we have $ v^{\infty}_{x}$ is continuous in $\bQ^{\infty}$ and $v^{\infty}$ is Lipschitz continuous w.r.t. $c$, also, there is a constant $K>0$ such that
\begin{gather}
	0\leq -v^{\infty}_{c}\leq K ~\mbox{a.e.},\label{inftyvc}\\
	\label{inftyvv}
	x\leq v^{\infty}(x,c)\leq x+\mu/r,\\
	\label{inftyvv_limc} \lim\limits_{c \to +\infty}v^{\infty}(x, c) = x,\\
	\label{inftyvx} 0\leq v^{\infty}_x\leq K,
\end{gather}
and for any $(x,c)\in \bQ^{\infty}$, 
\begin{align}\label{inftyvxx}
	v^{\infty}_x(y,c)\leq \max\{v^{\infty}_x(x,c), 1\},~~ \forall\; y\geq x.
\end{align} 
\end{theorem} 

Accordingly, we can define the {\bf switching region}:
\begin{align*}
\SS^{\infty}=\Big\{(x,c)\in \bQ^{\infty} \;\Big|\; v^{\infty}(x,c)=v^{\infty}(x,s) \text{ for some $s\in(c,+\infty)$} \Big\},
\end{align*}
the {\bf non-switching region}:
\begin{align*}
\NS^{\infty}=\Big\{(x,c)\in \bQ^{\infty} \;\Big|\; v^{\infty}(x,c)>v^{\infty}(x,s) \text{ for all $s\in(c,+\infty)$}\Big\},
\end{align*}
the {\bf equivalent maximum rate}:
\begin{align*}
\swb^{\infty}(x,c)=\max\Big\{s\in[0,+\infty)\;\Big|\; v^{\infty}(x,s)=v^{\infty}(x,c)\Big\}\in[c,+\infty),~~(x,c)\in \bQ^{\infty},
\end{align*}
{\bf equivalent minimum rate}:
\begin{align*}
\minswb^{\infty}(x,c)=\min\Big\{s\in[0,+\infty)\;\Big|\; v^{\infty}(x,s)=v^{\infty}(x,c)\Big\}\in[0,c],~~(x,c)\in \bQ^{\infty},
\end{align*}
the {\bf switching boundary}:
\begin{equation}\label{inftysswitching_boundary}
\sw^{\infty}(c)=\inf\Big\{x\in \R^+\;\Big|\; (x,c)\in \SS^{\infty}\Big\},~~ c\in[0,+\infty),
\end{equation}
and the {\bf converting boundary}:
\begin{align*}
\sd^{\infty}(c)=\inf\Big\{x\in\R^{+}\;\Big|\; v_x(x,c)\leq 1\Big\},~~ c\in[0,+\infty).
\end{align*}

We have the following characterization. 
\begin{proposition}\label{inftyprop:characterizationx-noceiling}
The switching boundary $\sw^{\infty}(\cdot)$ is positive, bounded and continuous on $[0,+\infty)$, and satisfies the estimate 
\begin{align}\label{inftyxl_ub}
	\limsup\limits_{c\to +\infty}\sw^{\infty}(c)\leq \frac{\mu}{r b} \left(1 + \sqrt{1 + 2b - 2b^2}\right).
\end{align}
Also, if $(x,c)\in \SS^{\infty}$, then $(y,c)\in \SS^{\infty}$ for any $y\geq x$. As a consequence,
the switching and non-switching regions $\SS^{\infty}$ and $\NS^{\infty}$ are separated by $\sw^{\infty}(\cdot)$:
\begin{gather*}
	\Big\{(x,c)\in \bQ^{\infty} \;\Big|\; x> \sw^{\infty}(c)\Big\}\subseteq\SS^{\infty}
	\subseteq\Big\{(x,c)\in \bQ^{\infty} \;\Big|\; x\geq \sw^{\infty}(c)\Big\},\\
	\Big\{(x,c)\in \bQ^{\infty} \;\Big|\; x< \sw^{\infty}(c)\Big\}\subseteq\NS^{\infty}
	\subseteq\Big\{(x,c)\in \bQ^{\infty} \;\Big|\; x\leq \sw^{\infty}(c)\Big\}. 
\end{gather*}
The converting boundary $\sd^{\infty}(\cdot)$ is continuous on $[0, +\infty)$, and satisfies
\begin{align}\label{inftyy<xnoceiling}
	0<\sd^{\infty}(c)<\sw^{\infty}(c),~~ c\in[0,+\infty),
\end{align}
and 
\begin{align}\label{inftyyl}
	\lim\limits_{c\to +\infty}\sd^{\infty}(c)=\frac{\mu}{r}.
\end{align}
\end{proposition} 

Similar to \thmref{thm:averi}, we can get the optimal control for the no-ceiling model. Due to page limit, we leave this to the interested readers.

\section{Numerical study} \label{sec:numerical}
In this section, we present the numerical results of our proposed model. The benchmark parameters are set as follows: $\mu=0.3$, $\bar{c}=0.3$, $\sigma=0.3$, $r=0.05$ and $b=0.5$. The subsequent analysis demonstrates the model's behavior when each parameter is varied individually. The numerical solution is obtained by solving the variational inequalities associated with the approximating regime-switching ODE systems (see Appendix \ref{sec:approximation}) through the standard penalty method.

To validate our approach, we conduct a comprehensive comparison with established results in the literature on optimal dividend payout problems under various constraints. Figure \ref{fig:b=0 and 1} illustrates the model predictions for two limiting cases.

First, we consider the case where $b=0$, corresponding to the absence of constraints on dividend payments. For this scenario, we benchmark our results against the classical analytical solution derived by \cite{asmussen1997controlled} and \cite{gerber2004optimal}. These works establish that the optimal threshold below which no dividends are paid is given by
\begin{equation}\label{analytical_threshold_GS04}
	x^{*}
	= \frac{1}{\theta_{1}(0) - \theta_{2}(0)}
	\log \frac{\theta_{2}(0)\bigl(\theta_{2}(0) - \theta_{2}(\bar{c})\bigr)}
	{\theta_{1}(0)\bigl(\theta_{1}(0) - \theta_{2}(\bar{c})\bigr)}.
\end{equation}

Second, we examine the case where $b=1$, which represents the optimal dividend strategy under ratcheting constraints. Here, we compare our results with the numerical solution obtained from \cite{albrecher2022optimal}, where the governing ODE is solved using an explicit Euler scheme. The boundary conditions for this comparison are numerically derived based on the specific functional forms defined in their work.

It is important to note the distinct characteristics of each case: when $b=0$, the absence of constraints on adjusting the running maximum of historical dividend payments necessitates reporting only the converting boundary $\sd(\cdot)$. Conversely, when $b=1$, the ratcheting constraint permits only increases or maintenance of the running maximum payout rate, requiring us to report solely the switching boundary $\sw(\cdot)$. The excellent agreement between our numerical results and the existing solutions in both limiting cases validates the accuracy and reliability of our computational approach.

\begin{figure}[H]
	\centering
	\begin{subfigure}[b]{0.48\linewidth}
		\centering
		\includegraphics[width=\linewidth]{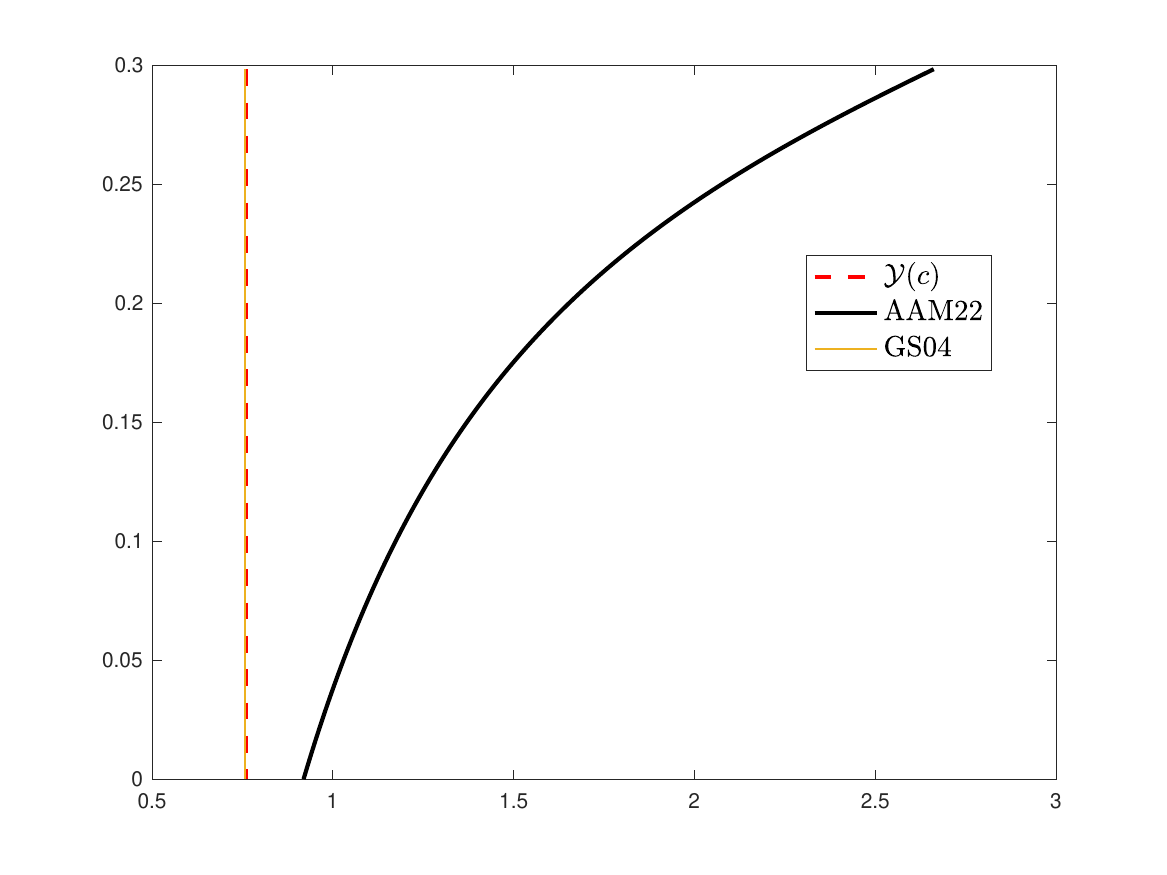}
		\caption{$b=0$, orange line is the case for no constraint collected from \eqref{analytical_threshold_GS04}, and the red curve is the numerical solution from AAM22 with ratcheting ($b=1$).The red dashed line is the output of the PDE with $b=0$,}
		\label{fig:b00}
	\end{subfigure}
	\hfill
	\begin{subfigure}[b]{0.48\linewidth}
		\centering
		\includegraphics[width=\linewidth]{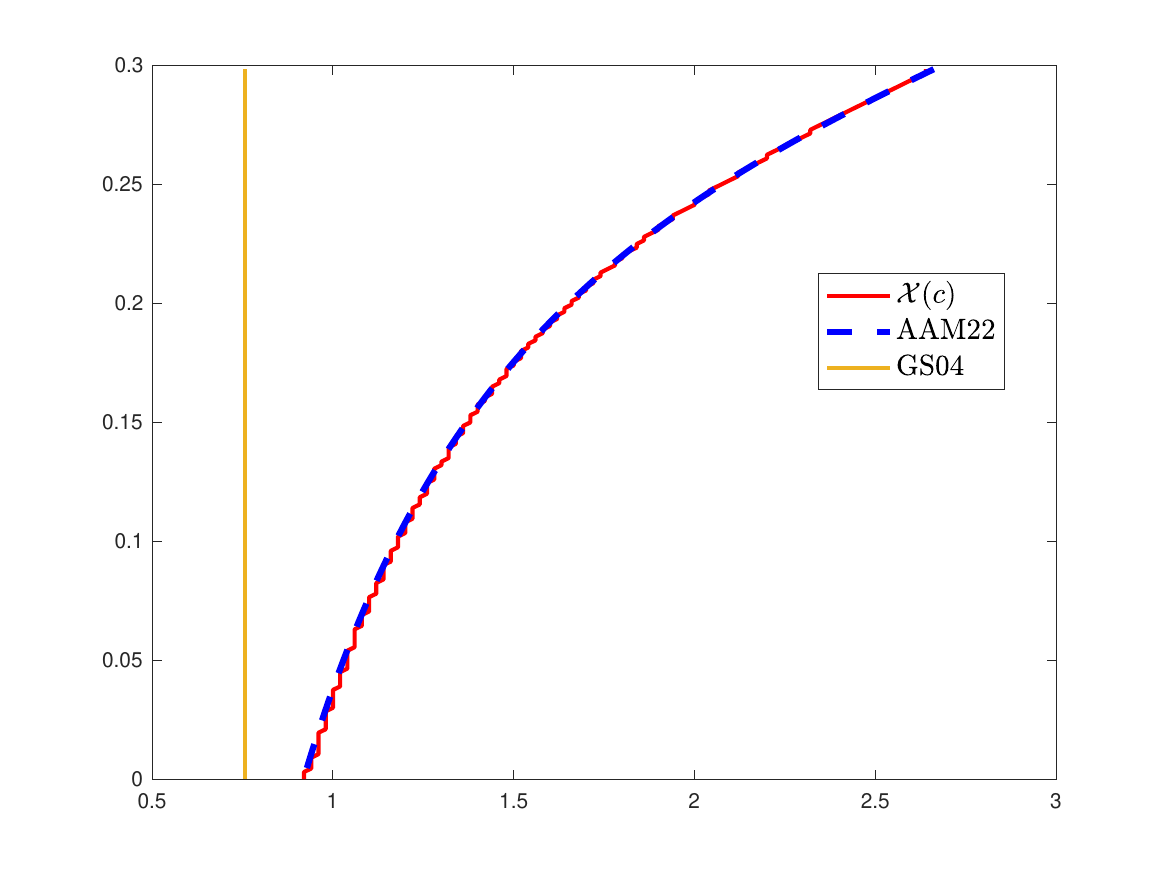}
		\caption{$b=0$, orange line is the case for no constraint collected from \eqref{analytical_threshold_GS04}, and the blue dashed line is the numerical solution from AAM22 with ratcheting ($b=1$).The red line is the output of the PDE with $b=1$.}
		\label{fig:b1}
	\end{subfigure}
	\caption{Comparison of numerical solutions and cases with no constraints for $b=0$ and $b=1$.}
	\label{fig:b=0 and 1}
\end{figure}

\subsection{Optimal dividend payout strategies}
Having validated our numerical approach against established benchmarks, we now turn to a detailed characterization of the optimal dividend payout strategies determined by the switching boundary $\sw(\cdot)$ and the converting boundary $\sd(\cdot)$. Figure~\ref{fig:phase_diagram} provides a graphical illustration of how these boundaries govern the evolution of dividend policies across different regions of the state space.

\begin{figure}[H]
	\centering
	\begin{subfigure}[b]{0.48\linewidth}
		\centering
		\includegraphics[width=\linewidth]{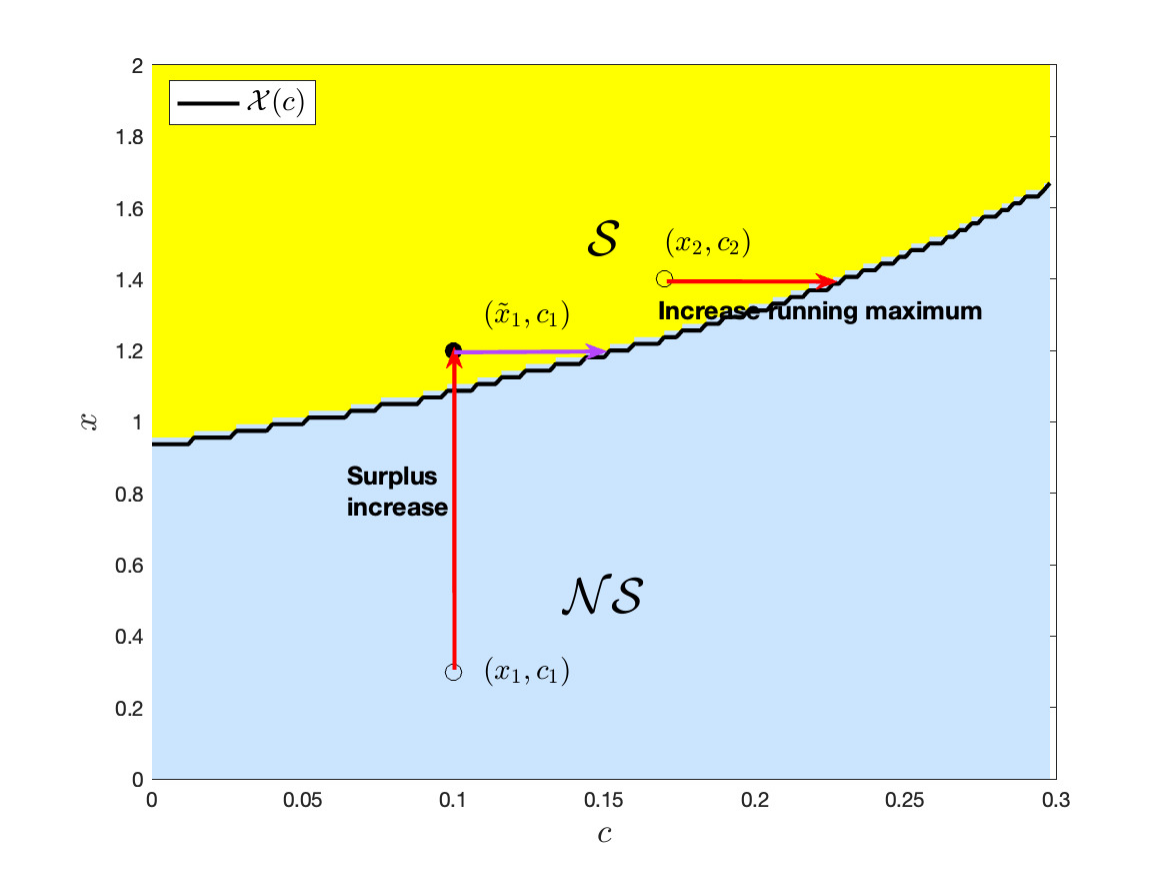}
		\caption{Switching boundary $\sw(\cdot)$.}
		\label{fig:switching_boundary}
	\end{subfigure}
	\hfill
	\begin{subfigure}[b]{0.48\linewidth}
		\centering
		\includegraphics[width=\linewidth]{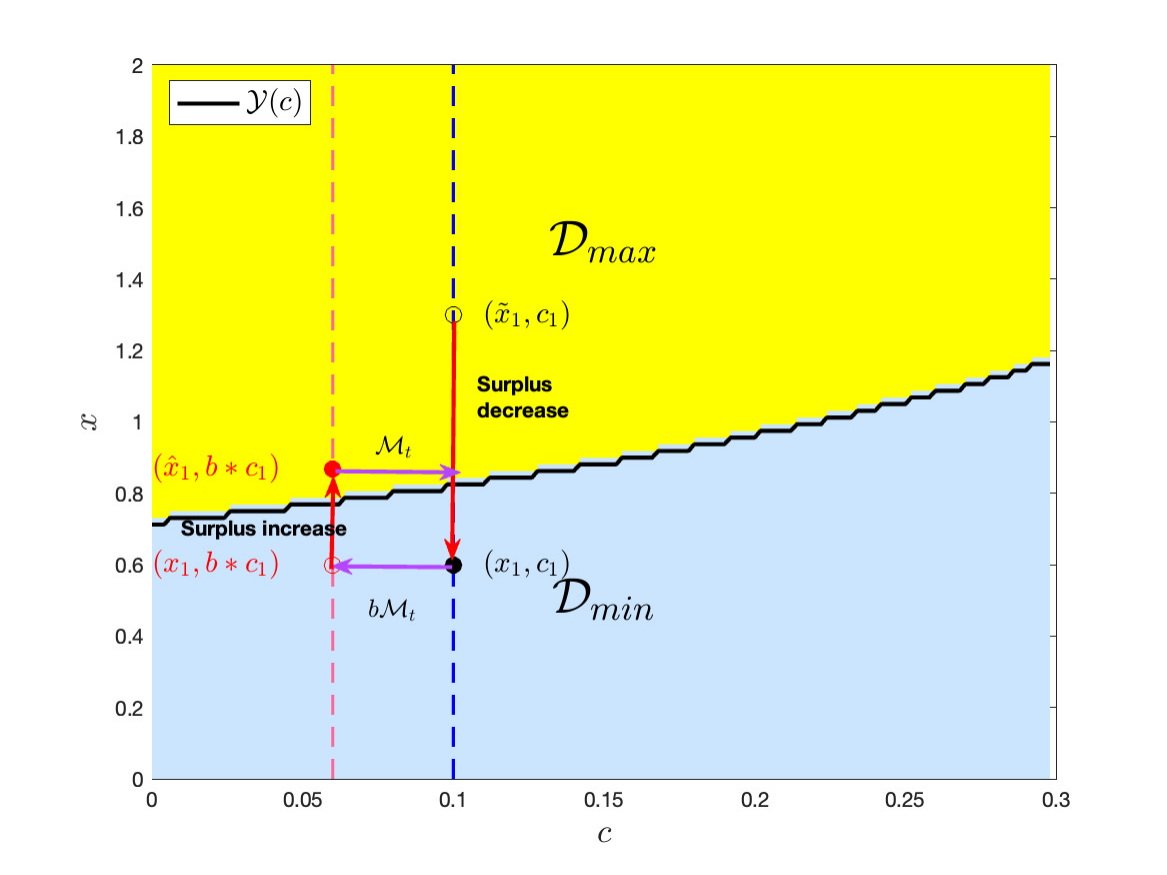}
		\caption{Converting boundary $\sd(\cdot)$.}
		\label{fig:converting_boundary}
	\end{subfigure}
	\caption{Illustration of optimal trajectories for different starting points $(x,c)$ across regions. The arrows indicate the evolution paths from initial states in different domains. Parameters: $\mu=0.3$, $b=0.6$, $r=0.05$, $\sigma=0.3$, $\bar{c}=0.3$.}
	\label{fig:phase_diagram}
\end{figure}

Figure~\ref{fig:phase_diagram} illustrates the optimal dividend payment strategies and their adjustments as the surplus level evolves across different regions. In both panels, hollow circles represent the initial state $(x,c)$, while solid circles indicate the state after surplus changes.

Panel~\ref{fig:switching_boundary} depicts the dynamics near the switching boundary $\sw(\cdot)$. For an $(x,c)$ pair initially located in the non-switching region $\mathcal{NS}$, the surplus increases vertically and approaches the switching boundary $\sw(\cdot)$. At the instant the surplus crosses this boundary upward, the running maximum dividend payment rate $\mathfrak{M}(x,c)$ is immediately increased to the level that touches the boundary, and the instantaneous dividend rate adjusts to this new running maximum level. Conversely, for an $(x,c)$ pair starting in the switching region $\mathcal{S}$, increasing the running maximum does not deteriorate the value function. Therefore, the dividend payment rate increases horizontally until it reaches the switching boundary.

Panel~\ref{fig:converting_boundary} demonstrates a complete cycle of cross-regional evolution near the converting boundary $\sd(\cdot)$. Starting from an $(x,c)$ pair with relatively high surplus in the $\mathcal{D}_{\max}$ region, as the surplus decreases, it moves vertically downward toward the converting boundary. At the instant the surplus crosses $\sd(\cdot)$ from above, the dividend payment rate is immediately reduced to the minimum permissible level $b\mathfrak{M}_t$. Subsequently, as the surplus recovers and increases, it moves vertically upward toward the converting boundary again. Once the surplus crosses $\sd(\cdot)$ from below, the dividend payment rate instantaneously jumps back to the running maximum level $\mathfrak{M}_t$.

These trajectories illustrate the bang-bang nature of the optimal dividend policy, where adjustments occur instantaneously at the boundaries, and the evolution between boundaries follows paths driven by surplus dynamics.

\begin{figure}[H]
	\centering
	
	\begin{subfigure}[b]{0.48\textwidth}
		\centering
		\includegraphics[width=\linewidth]{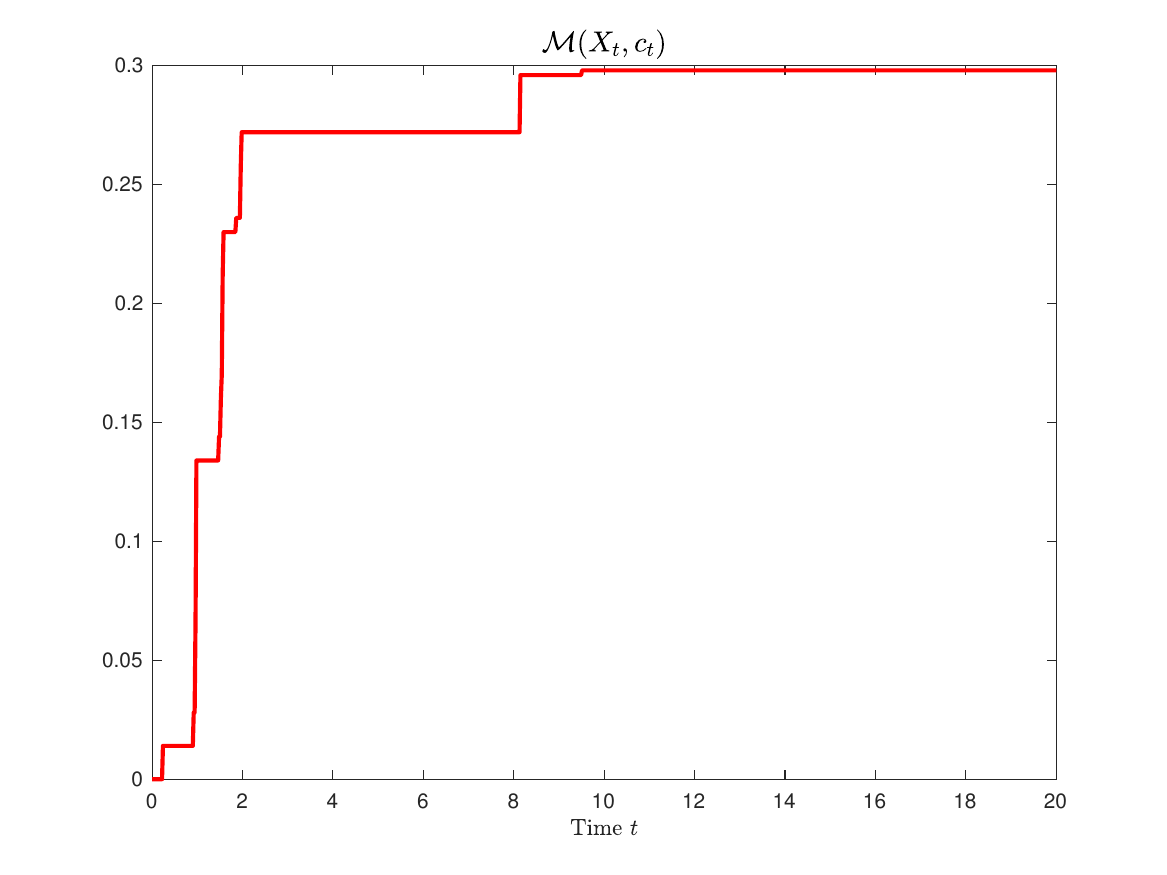}
		\caption{Evolution of running maximum $\mathfrak{M}(X_t,\mathcal{C}_t)$ over time.}
		\label{fig:running_max_path}
	\end{subfigure}~~~
	\begin{subfigure}[b]{0.48\textwidth}
		\centering
		\includegraphics[width=\linewidth]{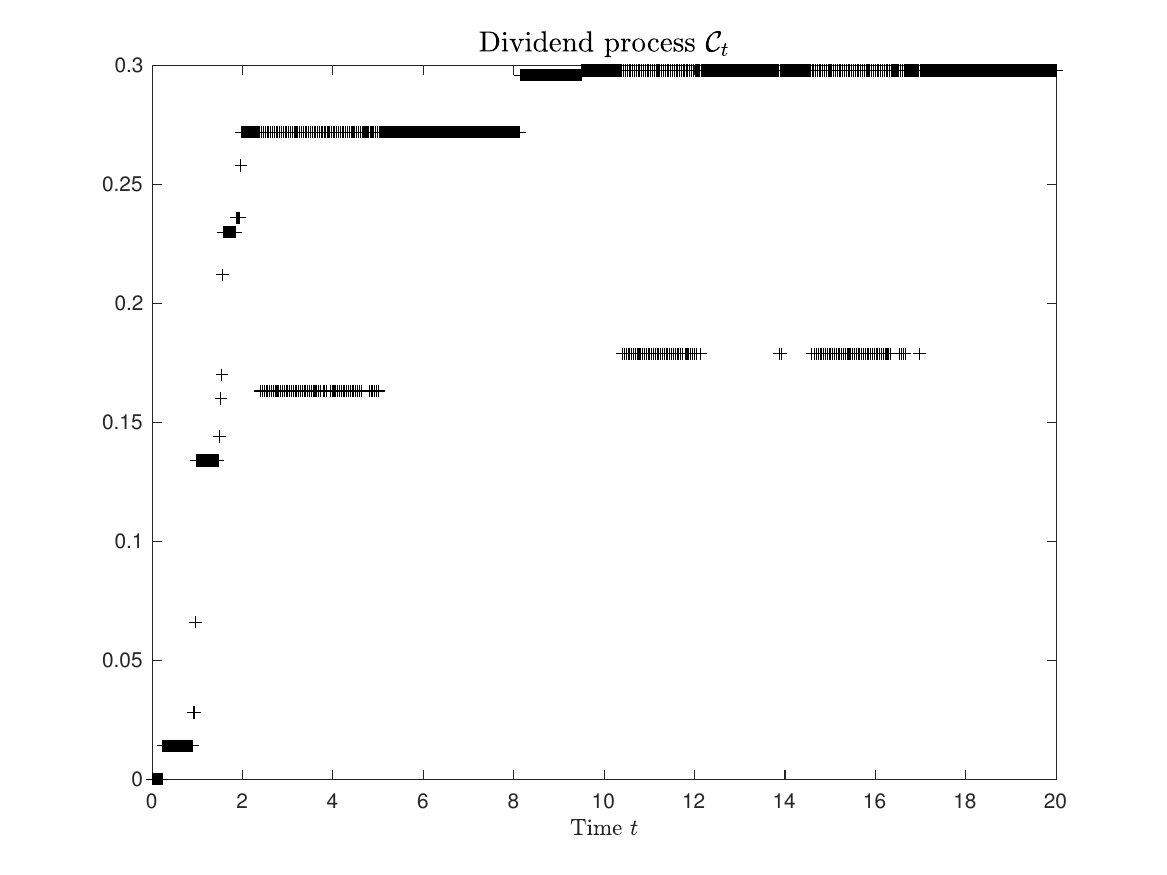}
		\caption{Evolution of instantaneous dividend payment rate over time.}
		\label{fig:dividend_rate_path}
	\end{subfigure}
	
	\par\bigskip
	
	\begin{subfigure}[b]{0.65\textwidth}
		\centering
		\includegraphics[width=\linewidth]{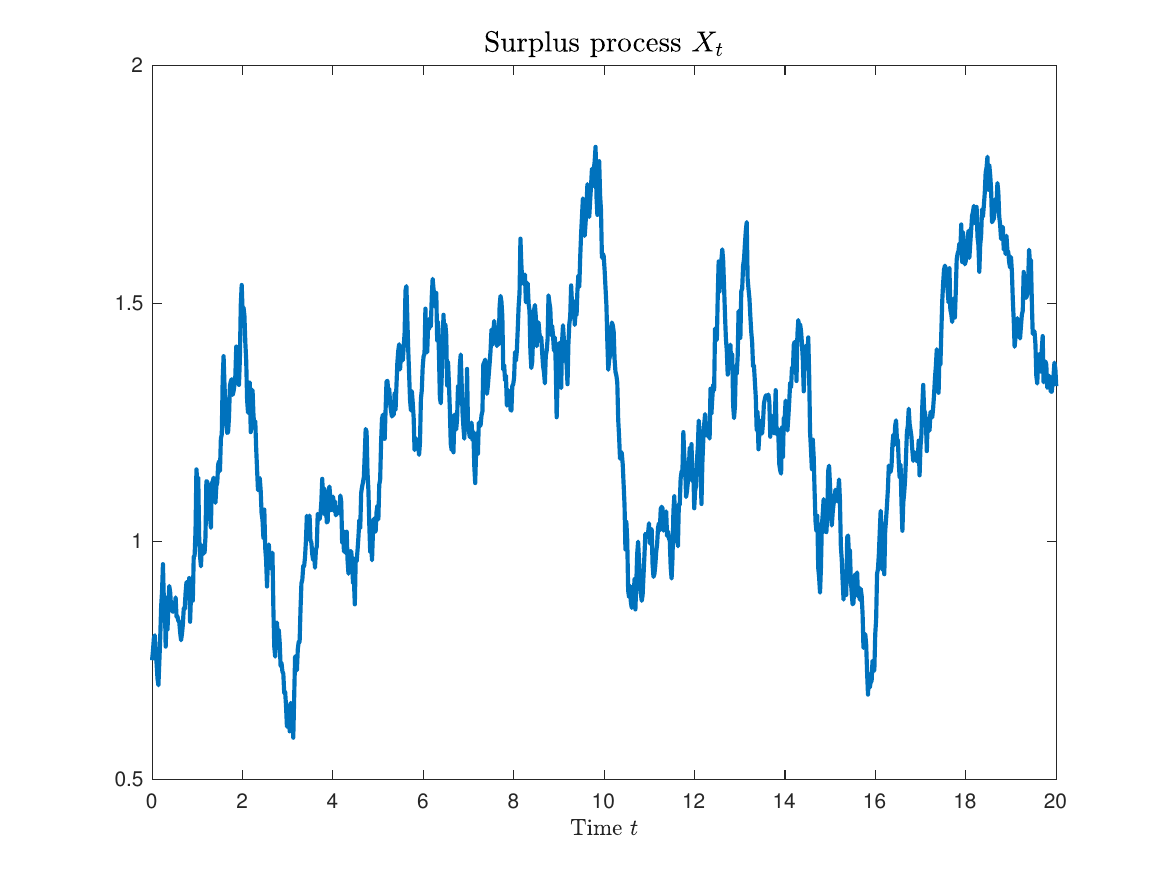}
		\caption{Evolution of surplus level $X_t$ over time.}
		\label{fig:surplus_path}
	\end{subfigure}
	
	\caption{Simulated sample paths under the optimal dividend payment policy determined by $\mathfrak{M}_t$. Panel \ref{fig:running_max_path} shows the evolution of the running maximum $\mathfrak{M}(X_t,\mathcal{C}_t)$, Panel \ref{fig:dividend_rate_path} displays the corresponding instantaneous dividend payment rate, and Panel \ref{fig:surplus_path} illustrates the dynamics of the surplus level. Parameters: $\mu=0.3$, $b=0.6$, $r=0.05$, $\sigma=0.3$, $\bar{c}=0.3$.}
	\label{fig:simulation_paths}
\end{figure}

Using the same parameter configuration as in Figure~\ref{fig:phase_diagram}, we conduct a simulation to illustrate the dynamic evolution of the optimal dividend policy. Figure~\ref{fig:simulation_paths} presents sample paths that demonstrate how the running maximum $\mathfrak{M}_t$ evolves with the surplus dynamics and how the dividend payment rate responds to these changes. Panel~\ref{fig:running_max_path} shows that $\mathfrak{M}(X_t,\mathcal{C}_t)$ increases progressively as the surplus fluctuates, reflecting the adjustments at the switching boundary $\sw(\cdot)$ characterized in Figure~\ref{fig:phase_diagram}. Panel~\ref{fig:dividend_rate_path} illustrates the corresponding instantaneous dividend payment rate, which exhibits jump discontinuities when the surplus crosses either the switching or converting boundaries. The bang-bang nature of the policy is evident as the dividend rate alternates between the running maximum $\mathfrak{M}_t$ and the minimum permissible level $b\mathfrak{M}_t$ in response to surplus movements across the converting boundary $\sd(\cdot)$. Panel~\ref{fig:surplus_path} displays the stochastic evolution of the surplus level $X_t$, which drives all the strategic adjustments in the dividend policy.

\subsection{Comparative statics}
The simulation results presented above demonstrate the operational mechanics of the optimal dividend policy in a single representative scenario. To provide a comprehensive understanding of how the optimal boundaries and resulting dividend strategies respond to changes in model parameters, we now conduct an extensive comparative statics analysis. The subsequent numerical experiments examine the sensitivity of both the switching boundary $\sw(\cdot)$ and the converting boundary $\sd(\cdot)$ to key parameters, with systematic comparisons to existing results in the literature.

Figure \ref{fig:Varying drawdown ratio b} presents a comprehensive comparison of the converting boundary and switching boundary with the two-curve strategy reported in \cite{albrecher2023optimal}, where the drawdown ratio $b$ is varied from 0.3 to 0.9. To provide a complete perspective on the evolution of the free boundaries, we include the two limiting cases ($b=0$ and $b=1$) in our analysis. This allows for a clear visualization of how the free boundaries transition as the drawdown ratio increases monotonically. The comparison reveals excellent agreement between our results and the numerical ODE solutions presented in \cite{albrecher2023optimal}, with no significant discrepancies observed across the entire range of $b$ values examined. This consistency further validates the robustness of our numerical approach. We can observe that increasing value of $b$ shall result bigger slopes of both the switching and converting curve, this is consistent with the financial intuition that stricter drawdown constraint should discourage the further increase on the maximum dividend payout rate. 

Additionally, Figure \ref{fig:varying_b_value} illustrates the corresponding value functions for different drawdown ratios, revealing a clear economic trade-off inherent in the drawdown constraint. Panel \ref{fig:value_b} demonstrates that as the drawdown ratio $b$ increases from 0.1 to 0.9, the value function uniformly decreases across all surplus levels. This monotonic relationship indicates that more restrictive drawdown constraints (higher $b$ values) systematically reduce the firm's ability to optimize dividend distributions, thereby diminishing shareholder value. Panel \ref{fig:delta_value_b} quantifies this value erosion by displaying the differences in value functions relative to the baseline case of $b=0.1$. The magnitude of value loss intensifies as $b$ approaches unity, with the most pronounced deterioration occurring in the intermediate surplus range. Notably, the value difference reaches its maximum negative impact around surplus levels of 2-4 units, suggesting that firms with moderate capital reserves are most affected by stringent drawdown constraints. Panel \ref{fig:delta_3d_b} provides a three-dimensional visualization that captures the joint dynamics of surplus level, dividend rate, and the resulting value differentials across various drawdown ratios, further emphasizing how tighter constraints fundamentally reshape the optimization landscape and impose substantial costs on dividend flexibility.

\begin{figure}[H]
	\centering
	\begin{subfigure}[b]{0.47\linewidth}
		\centering
		\includegraphics[width=\linewidth]{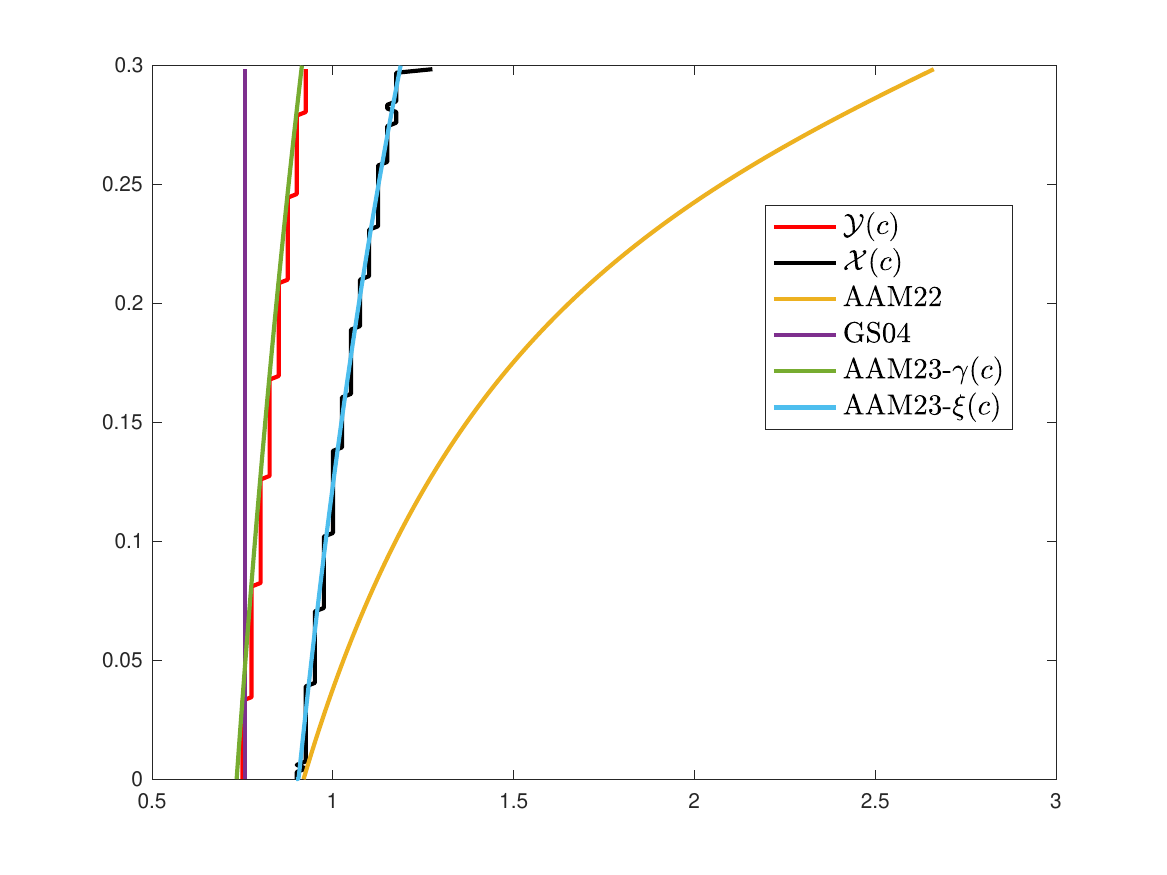}
		\caption{$b=0.3$}
		\label{fig:large11}
	\end{subfigure}~~
	\begin{subfigure}[b]{0.47\linewidth}
		\centering
		\includegraphics[width=\linewidth]{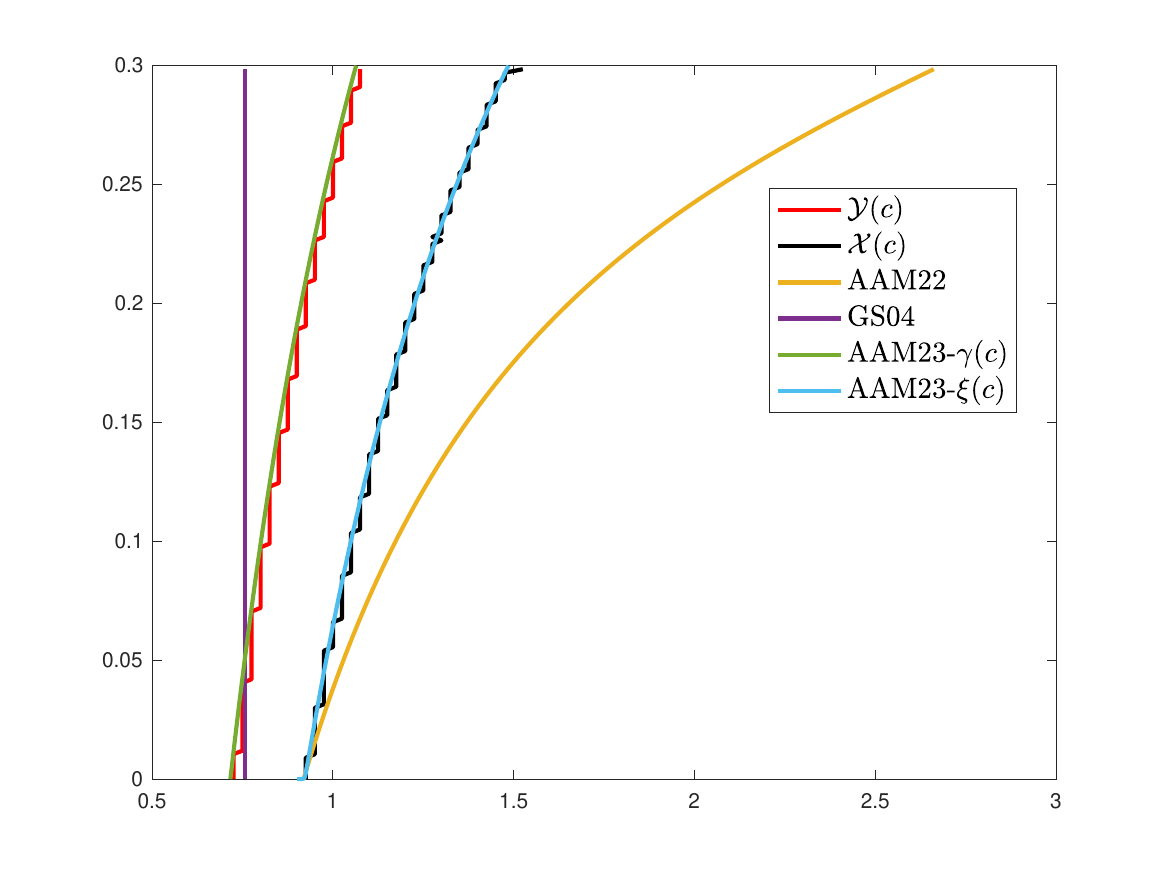}
		\caption{$b=0.5$}
		\label{fig:b1gx1}
	\end{subfigure}\\
	\begin{subfigure}[b]{0.47\linewidth}
		\centering
		\includegraphics[width=\linewidth]{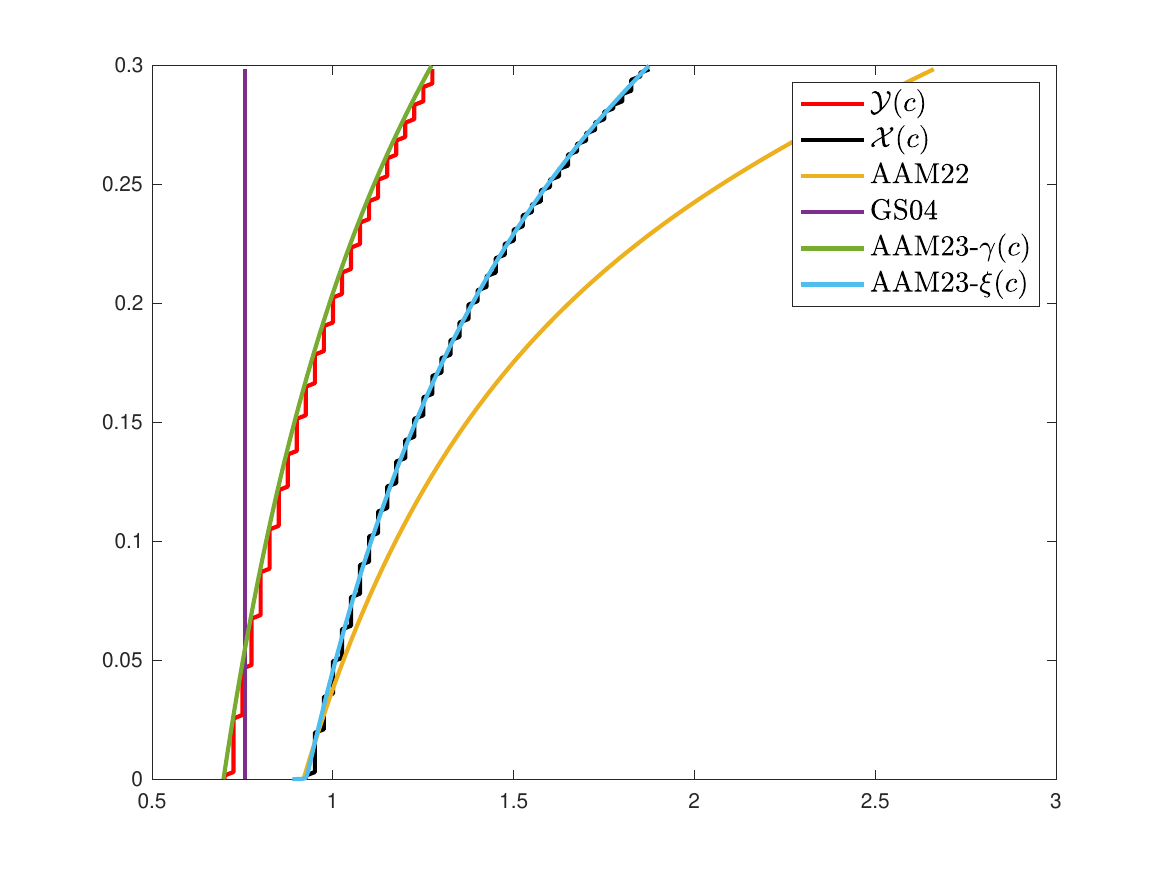}
		\caption{$b=0.7$}
		\label{fig:b0}
	\end{subfigure}~~
	\begin{subfigure}[b]{0.47\linewidth}
		\centering
		\includegraphics[width=\linewidth]{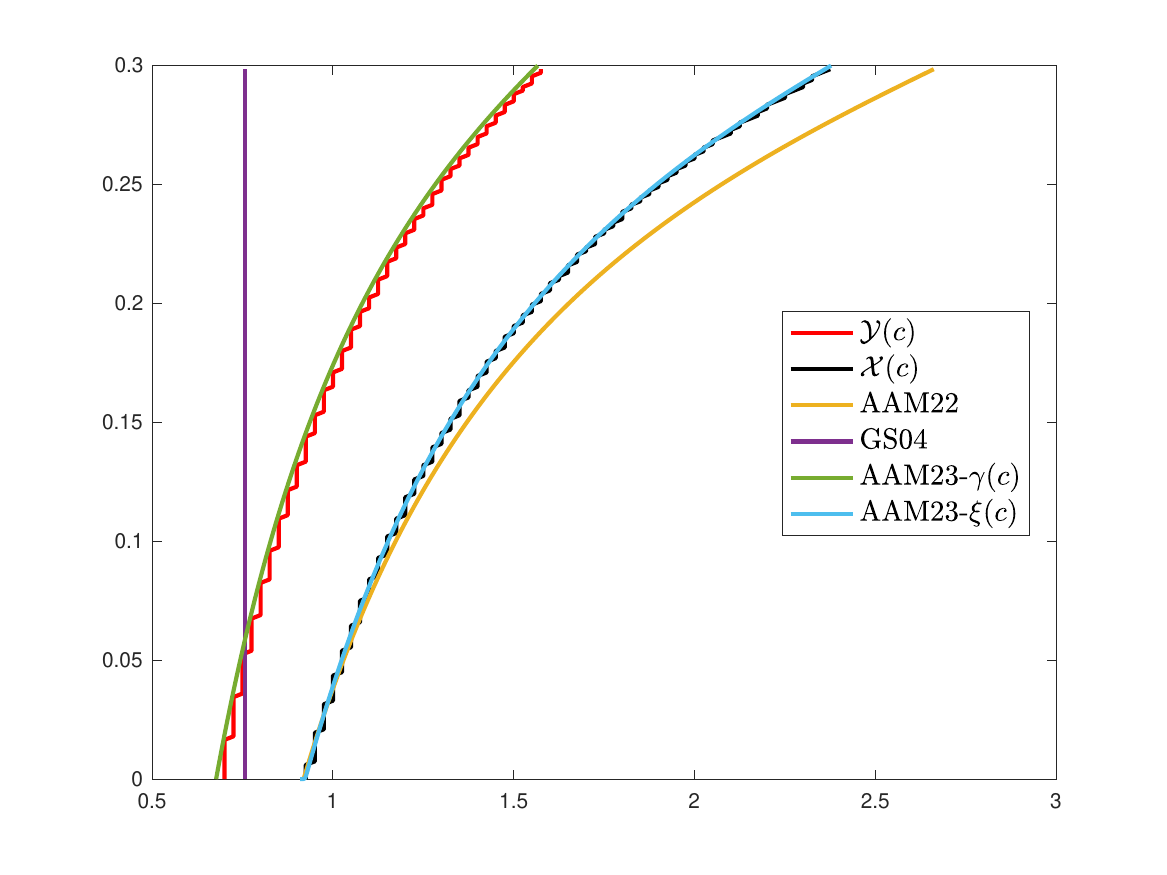}
		\caption{$b=0.9$}
		\label{fig:b02}
	\end{subfigure}
	\caption{Comparison of the optimal dividend payout strategies with different drawdown ratios. The left purple and the right orange curve corresponding to two extreme cases $b=0$ and $b=1$ derived from \eqref{analytical_threshold_GS04} and \cite{albrecher2022optimal}, the green and blue cure were solved from the two curved strategies by \cite{albrecher2023optimal}, the red and black curve were the numerical out of our PDEs.}
	\label{fig:Varying drawdown ratio b}
\end{figure}

\begin{figure}[H]
	\centering
	
	\begin{subfigure}[b]{0.48\textwidth}
		\centering
		\includegraphics[width=\linewidth]{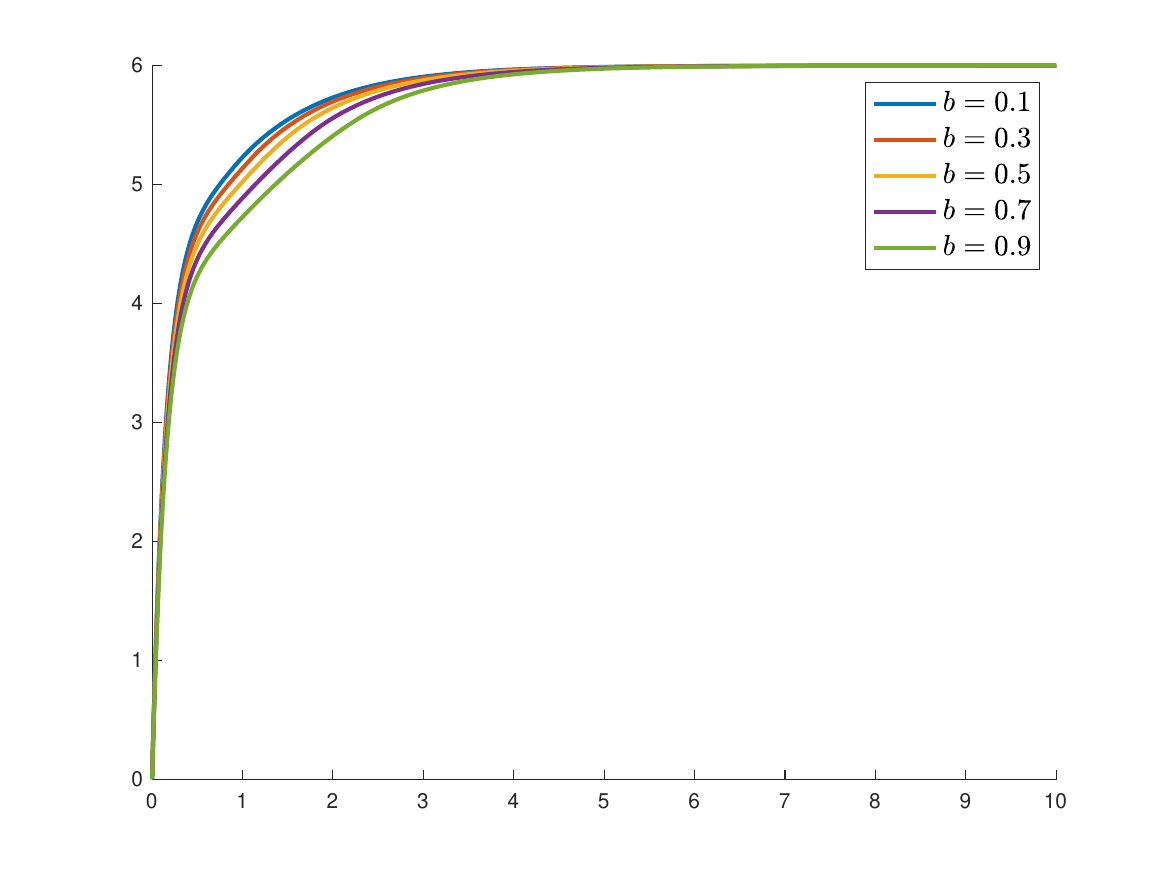}
		\caption{Values at $c=0$ across $b\in[0.1,0.9]$.}
		\label{fig:value_b}
	\end{subfigure}~~~
	\begin{subfigure}[b]{0.48\textwidth}
		\centering
		\includegraphics[width=\linewidth]{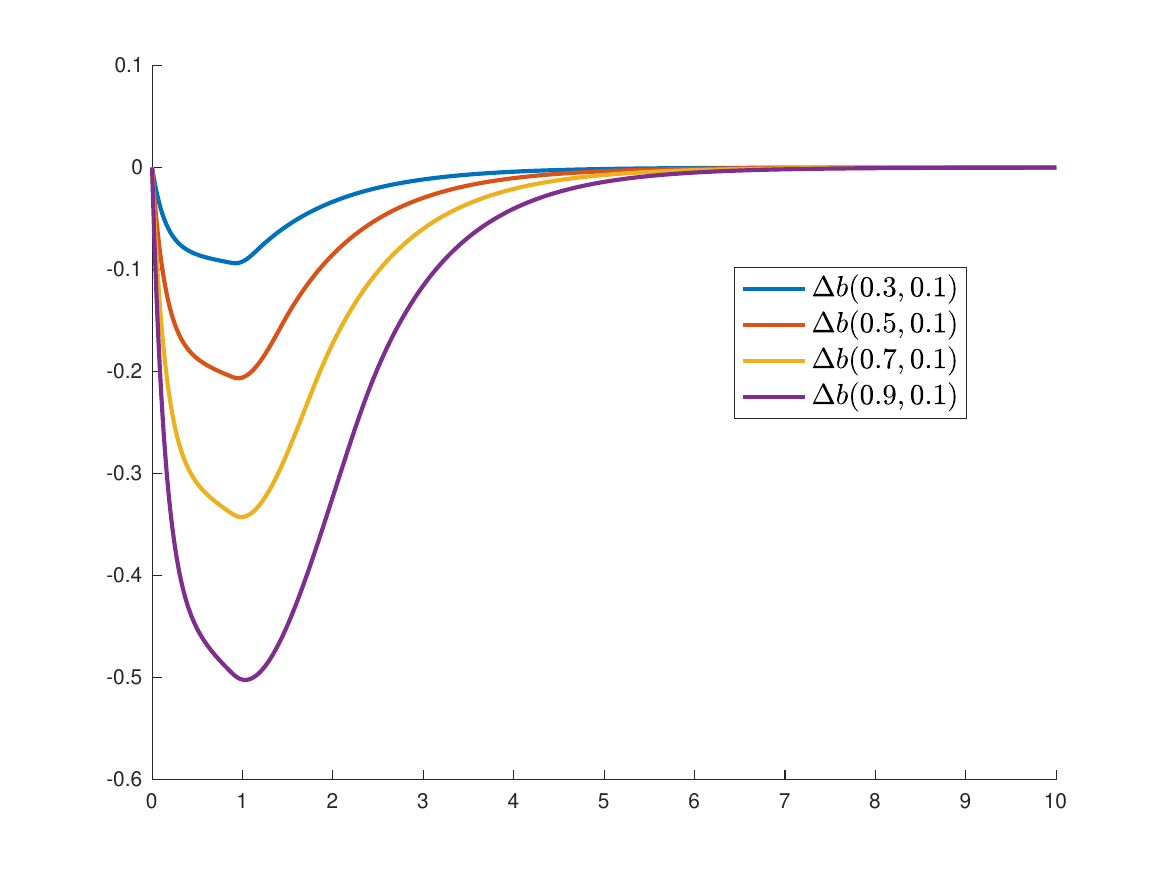}
		\caption{Value differences at $c=0$ relative to $b=0.1$.}
		\label{fig:delta_value_b}
	\end{subfigure}
	
	\par\bigskip
	
	\begin{subfigure}[b]{0.65\textwidth}
		\centering
		\includegraphics[width=\linewidth]{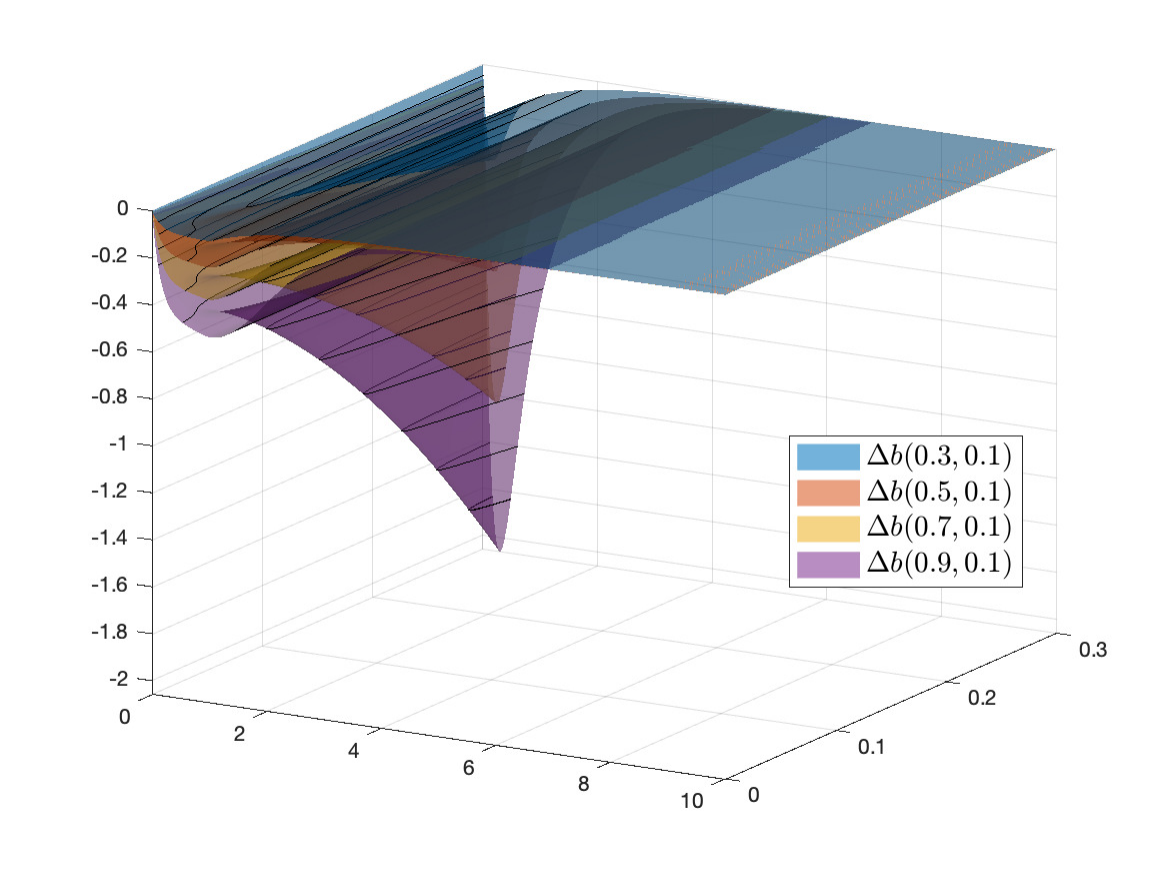}
		\caption{3D differences over $(x,c)$; $x\in[0,10]$, $c\in[0,0.3]$, measured relative to $b=0.1$.}
		\label{fig:delta_3d_b}
	\end{subfigure}
	
	\caption{Effect of $b$ on the value function. X: wealth $x$; Y: value; Z (in 3D): difference. For (a)–(b) we fix $c=0$; (b) shows differences vs. $b=0.1$; (c) shows differences over $x\in[0,10]$ and $c\in[0,0.3]$. We vary $b$ from $0.1$ to $0.9$.}
	\label{fig:varying_b_value}
\end{figure}

Figure \ref{fig:Maximum Dividend payout ration} depicts the free boundaries associated with various maximum dividend payout ratios. Consistent with the observations in \cite{albrecher2023optimal}, the parameter $\bar{c}$ significantly influences the location of both the converting boundary $\sd(\cdot)$ and the switching boundary $\sw(\cdot)$. The switching boundary $\sw(\cdot)$ determines when to increase the running maximum dividend payout ratio, while the converting boundary $\sd(\cdot)$ distinguishes between paying dividends at the minimum allowable rate (given by the drawdown ratio $b$ times the running maximum) and the maximum allowable rate (the running maximum itself). As $\bar{c}$ increases from 0.2 to 0.4, both boundaries shift rightward, requiring higher surplus levels before updating the running maximum or transitioning between minimum and maximum payout rates. This pattern reveals a key economic insight: larger values of $\bar{c}$ allow for more generous maximum dividend distributions, but the firm must accumulate greater reserves before it becomes optimal to either increase the running maximum or switch from the constrained minimum rate to the full maximum rate. This reflects the inherent tension between dividend capacity and capital adequacy under the drawdown constraint framework. 

\begin{figure}[H]
	\centering
	\begin{subfigure}[b]{0.48\textwidth}
		\centering
		\includegraphics[width=\linewidth]{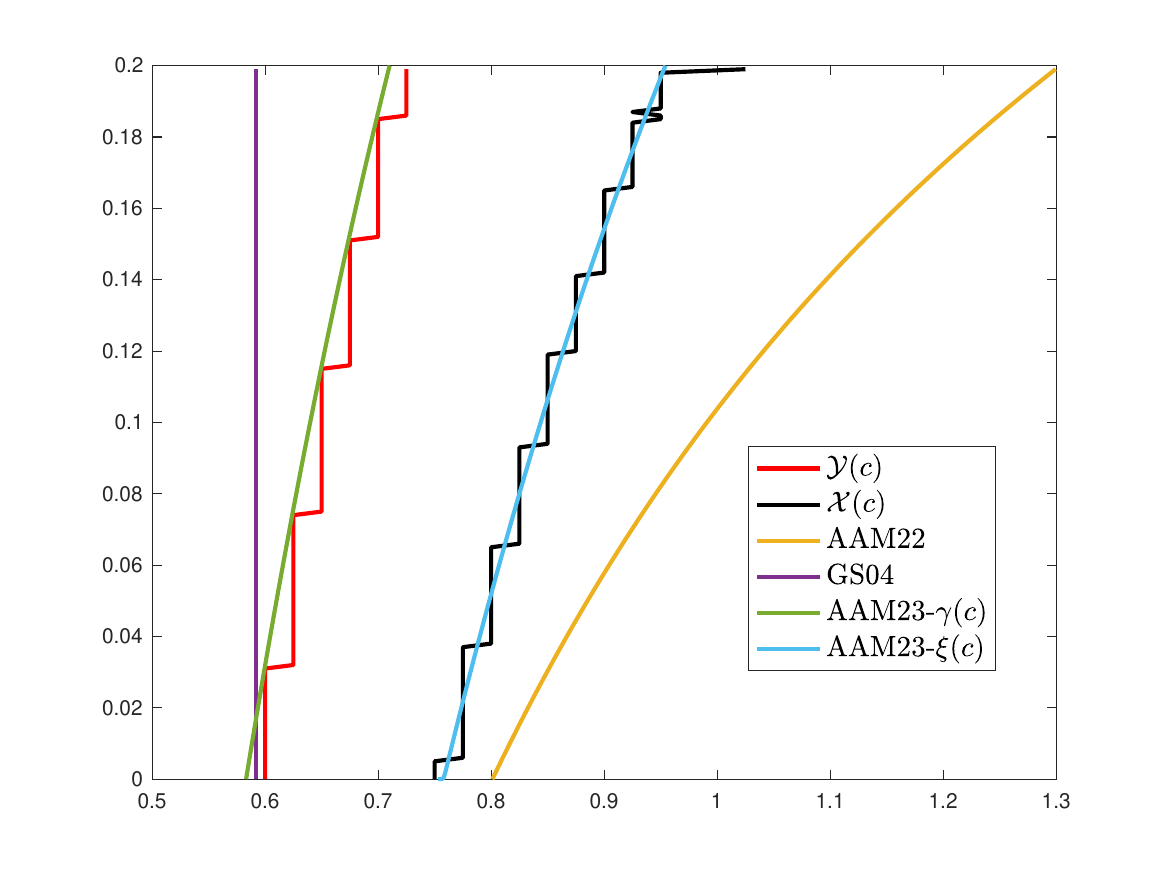}
		\caption{$\bar{c}=0.2$, }
		\label{fig:c=0.2 boundary}
	\end{subfigure}
	~~
	\begin{subfigure}[b]{0.48\textwidth}
		\centering
		\includegraphics[width=\linewidth]{b05_graph_new.pdf}
		\caption{$\bar{c}=0.3$, }
		\label{fig:c=0.3 boundary}
	\end{subfigure}
	\bigskip\\
	\begin{subfigure}[b]{0.48\textwidth}
		\centering
		\includegraphics[width=\linewidth]{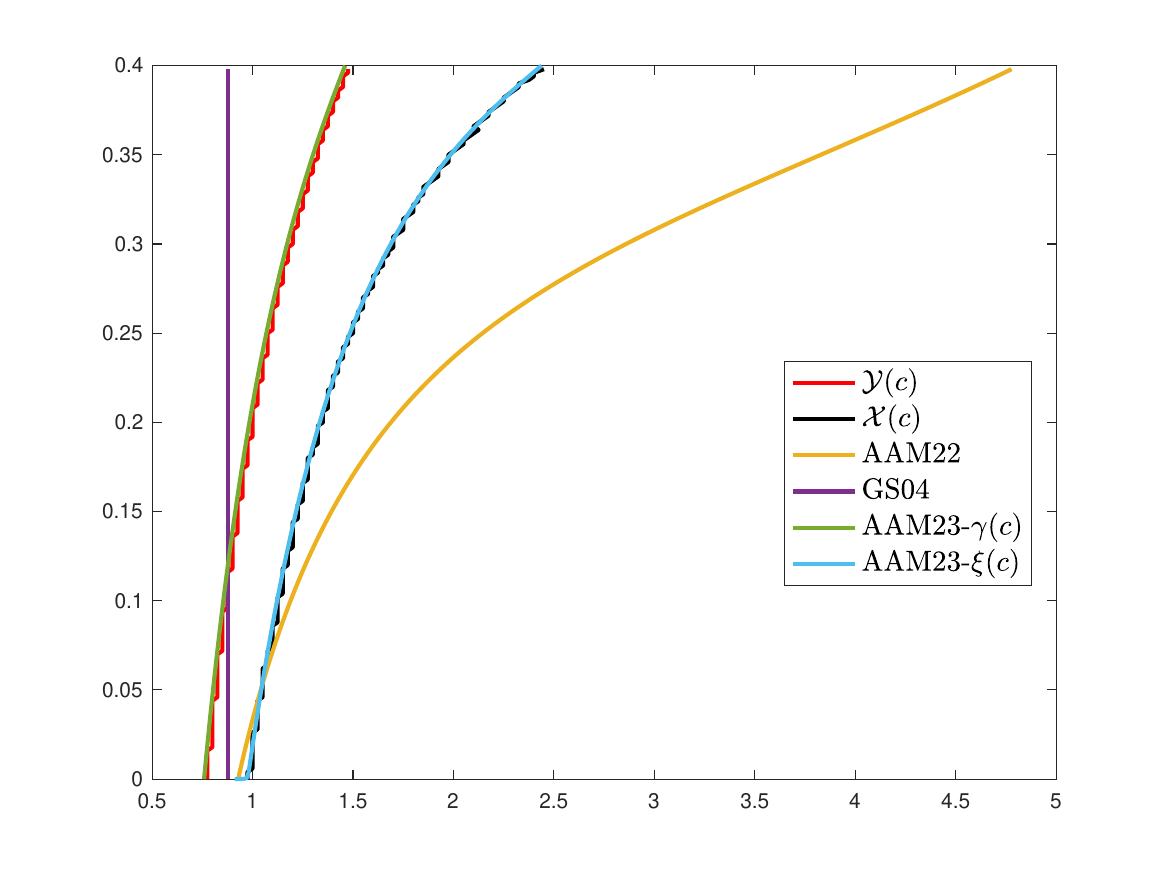}
		\caption{$\bar{c}=0.4$. }
		\label{fig: c=0.4 boundary}
	\end{subfigure}
	\caption{Comparison of the optimal dividend payout strategies with different maximum dividend payout ratio $\bar{c}$. The left purple and the right orange curve corresponding to two extreme cases $b=0$ and $b=1$ derived from \eqref{analytical_threshold_GS04} and \cite{albrecher2022optimal}, the green and blue cure were solved from the two curved strategies by \cite{albrecher2023optimal}, the red and black curve were the numerical out of our PDEs. }
	\label{fig:Maximum Dividend payout ration}
\end{figure}	

Figure \ref{fig:varying_cbar_value} quantifies the economic implications of these structural changes. Panel \ref{fig:value_cbar} demonstrates that increasing $\bar{c}$ monotonically enhances the value function across all surplus levels, with the most dramatic improvements occurring at higher surplus values. Panel \ref{fig:delta_value_cbar} reveals that the value differential relative to $\bar{c}=0.2$ becomes increasingly pronounced as surplus grows, indicating that firms with substantial capital reserves benefit disproportionately from relaxed payout constraints. The three-dimensional visualization in Panel \ref{fig:delta_3d_cbar} further illustrates this nonlinear relationship, showing how the value function's sensitivity to $\bar{c}$ intensifies at higher surplus levels. While the drawdown constraint itself may not impose severe efficiency losses for a given $\bar{c}$ (as noted in \cite{albrecher2023optimal} for $\bar{c}=3$), the magnitude of $\bar{c}$ fundamentally determines the scale of achievable dividend values, underscoring its critical role in dividend optimization under regulatory constraints.

\begin{figure}[H]
	\centering
	
	\begin{subfigure}[b]{0.48\textwidth}
		\centering
		\includegraphics[width=\linewidth]{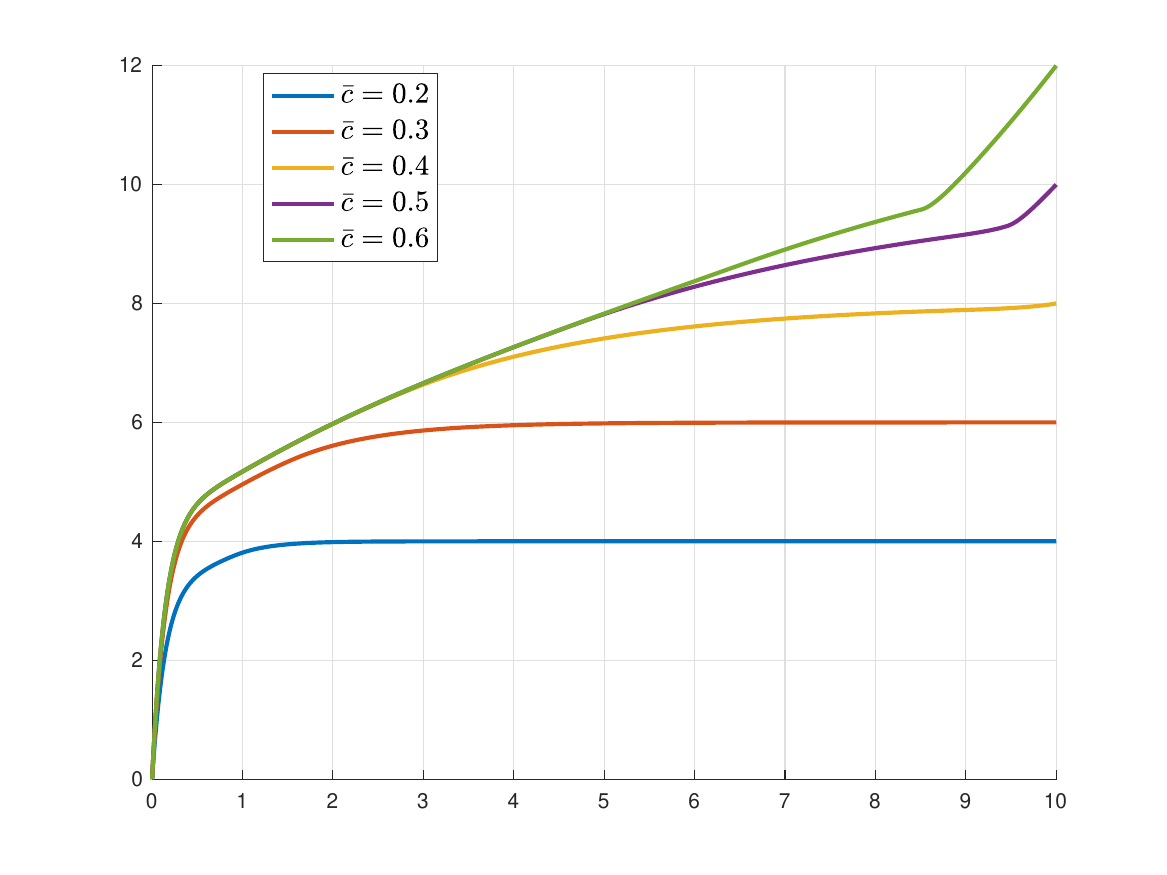}
		\caption{Values at $c=0$ across $\bar{c}\in[0.2,0.6]$,}
		\label{fig:value_cbar}
	\end{subfigure}\hfill
	\begin{subfigure}[b]{0.48\textwidth}
		\centering
		\includegraphics[width=\linewidth]{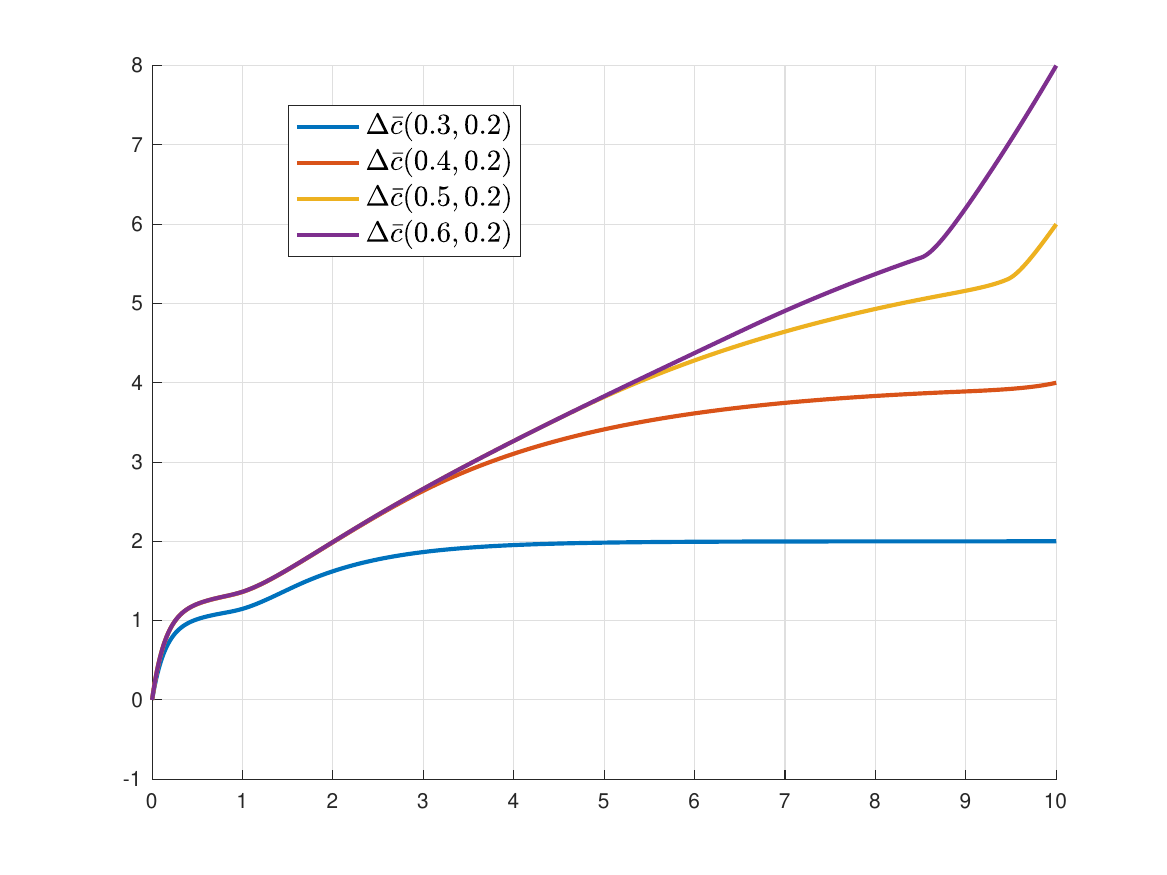}
		\caption{Value differences at $c=0$ relative to $\bar{c}=0.2$,}
		\label{fig:delta_value_cbar}
	\end{subfigure}
	
	\par\bigskip
	
	\begin{subfigure}[b]{0.72\textwidth}
		\centering
		\includegraphics[width=\linewidth]{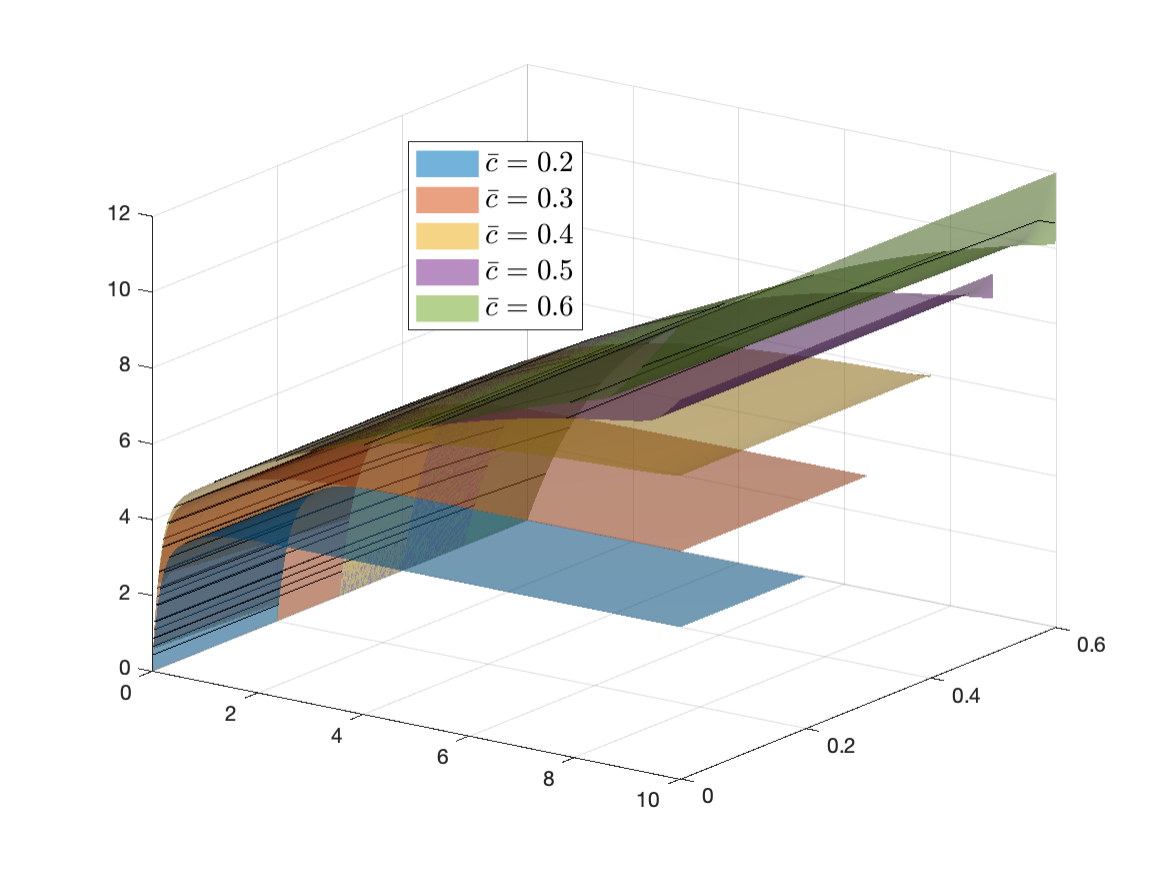}
		\caption{3D values over $(x,c)$; $x\in[0,10]$with maximum $\bar{c}$ range from $[0.2,0.6]$.}
		\label{fig:delta_3d_cbar}
	\end{subfigure}
	
	\caption{Effect of $\bar{c}$ on the value function. X: wealth $x$; Y: value; Z (in 3D): value. For (a)–(b) we fix $c=0$; (b) shows differences vs. $\bar{c}=0.2$; (c) shows values over $x\in[0,10]$ and maximum $\bar{c}$ range from $[0.2,0.6]$. We fix $b=0.6$ in these cases.}
	\label{fig:varying_cbar_value}
\end{figure}	

Figures \ref{fig:varying volatility}, \ref{fig:varying mu}, and \ref{fig:varying r} present comprehensive sensitivity analyses of the optimal dividend payout policies under varying market conditions. 

Figure \ref{fig:varying volatility} examines the impact of volatility levels ($\sigma \in \{0.2, 0.3, 0.4\}$). As volatility increases, both the switching boundary $\sw(\cdot)$ and converting boundary $\sd(\cdot)$ shift rightward, requiring higher surplus thresholds before adjusting dividend strategies. This conservative shift reflects the heightened risk environment: greater volatility necessitates larger capital buffers before firms can optimally increase their running maximum or transition from minimum to maximum payout rates. The boundaries also become more separated as volatility rises, indicating that the distinction between different dividend regimes becomes more pronounced under increased uncertainty. 

Figure \ref{fig:varying mu} illustrates the effect of drift coefficients ($\mu \in \{0.2, 0.3, 0.4\}$). Higher drift values, representing improved expected profitability, enable more aggressive dividend policies with boundaries shifting leftward. This allows firms to initiate dividend adjustments at lower surplus levels, capitalizing on favorable growth prospects. Notably, the impact is most pronounced for the converting boundary, suggesting that expected profitability particularly influences the decision between minimum and maximum payout rates. 

Figure \ref{fig:varying r} demonstrates the influence of interest rates ($r \in \{0.03, 0.05, 0.08\}$). Rising interest rates compress both boundaries toward lower surplus values, reflecting the increased opportunity cost of retained earnings. The effect is particularly striking at $r=0.03$, where the boundaries extend significantly rightward, indicating that low-rate environments encourage capital retention before dividend optimization. 

Across all scenarios, our PDE-based solutions (shown in red and black) align closely with the analytical benchmarks for extreme cases ($b=0$ and $b=1$) and the curved strategies from existing literature, validating our numerical approach while revealing the nuanced interplay between market parameters and optimal dividend policies under drawdown constraints.

\begin{figure}[H]
	\centering
	\begin{subfigure}[b]{0.48\textwidth}
		\centering
		\includegraphics[width=\linewidth]{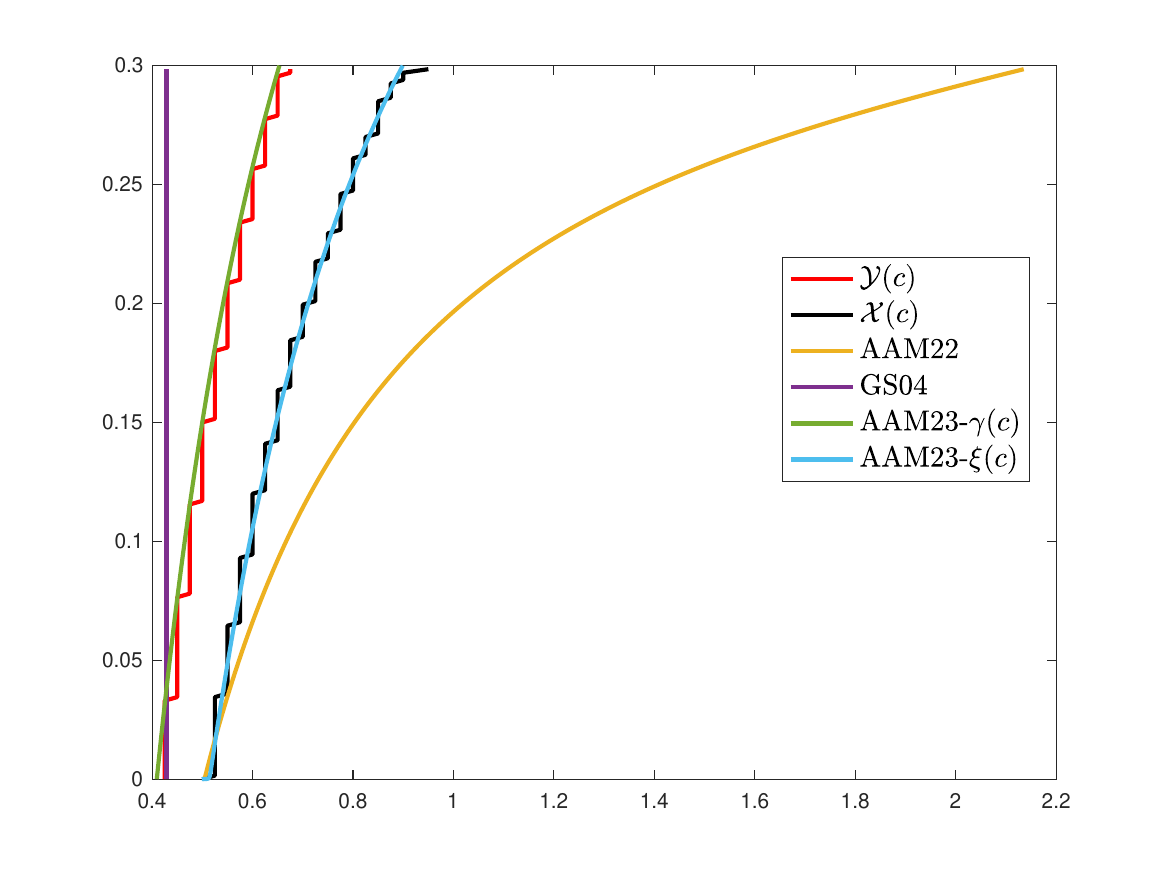}
		\caption{$\sigma=0.2$,}
		\label{fig:sigma=0.2}
	\end{subfigure}
	~~
	\begin{subfigure}[b]{0.48\textwidth}
		\centering
		\includegraphics[width=\linewidth]{b05_graph_new.pdf}
		\caption{$\sigma=0.3$,}
		\label{fig:sigma=0.3}
	\end{subfigure}
	\bigskip\\
	\begin{subfigure}[b]{0.48\textwidth}
		\centering
		\includegraphics[width=\linewidth]{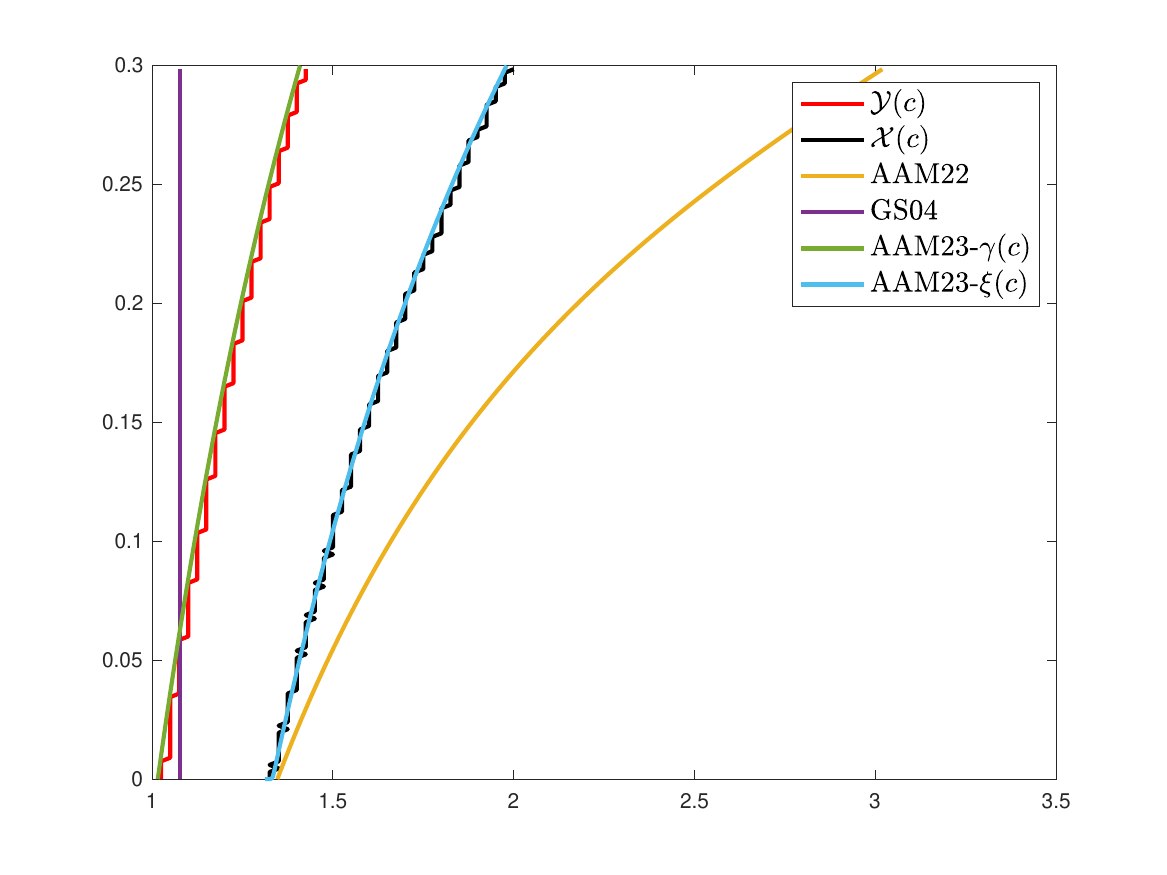}
		\caption{$\sigma=0.4$.}
		\label{fig:sigma=0.4}
	\end{subfigure}
	\caption{Comparison for dividend payout polices with different volatility levels. The left purple and the right orange curve corresponding to two extreme cases $b=0$ and $b=1$ derived from \eqref{analytical_threshold_GS04} and \cite{albrecher2022optimal}, the green and blue cure were solved from the two curved strategies by \cite{albrecher2023optimal}, the red and black curve were the numerical out of our PDEs.}
	\label{fig:varying volatility}
\end{figure}

\begin{figure}[H]
	\centering
	\begin{subfigure}[b]{0.48\textwidth}
		\centering
		\includegraphics[width=\linewidth]{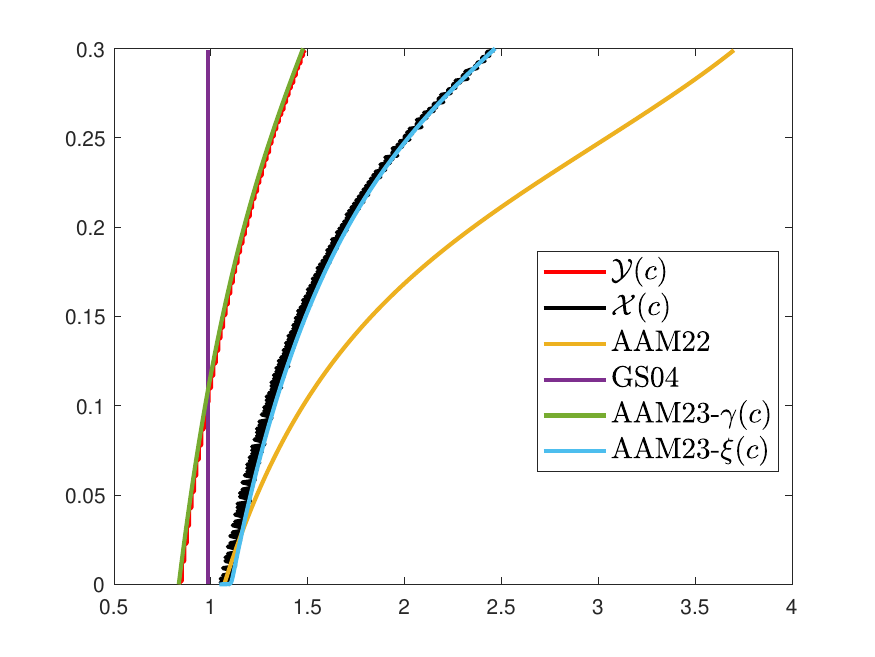}
		\caption{$\mu=0.2$}
		\label{fig:mu=0.2}
	\end{subfigure}
	~~
	\begin{subfigure}[b]{0.48\textwidth}
		\centering
		\includegraphics[width=\linewidth]{b05_graph_new.pdf}
		\caption{$\mu=0.3$}
		\label{fig:mu=0.3}
	\end{subfigure}
	\bigskip\\
	\begin{subfigure}[b]{0.48\textwidth}
		\centering
		\includegraphics[width=\linewidth]{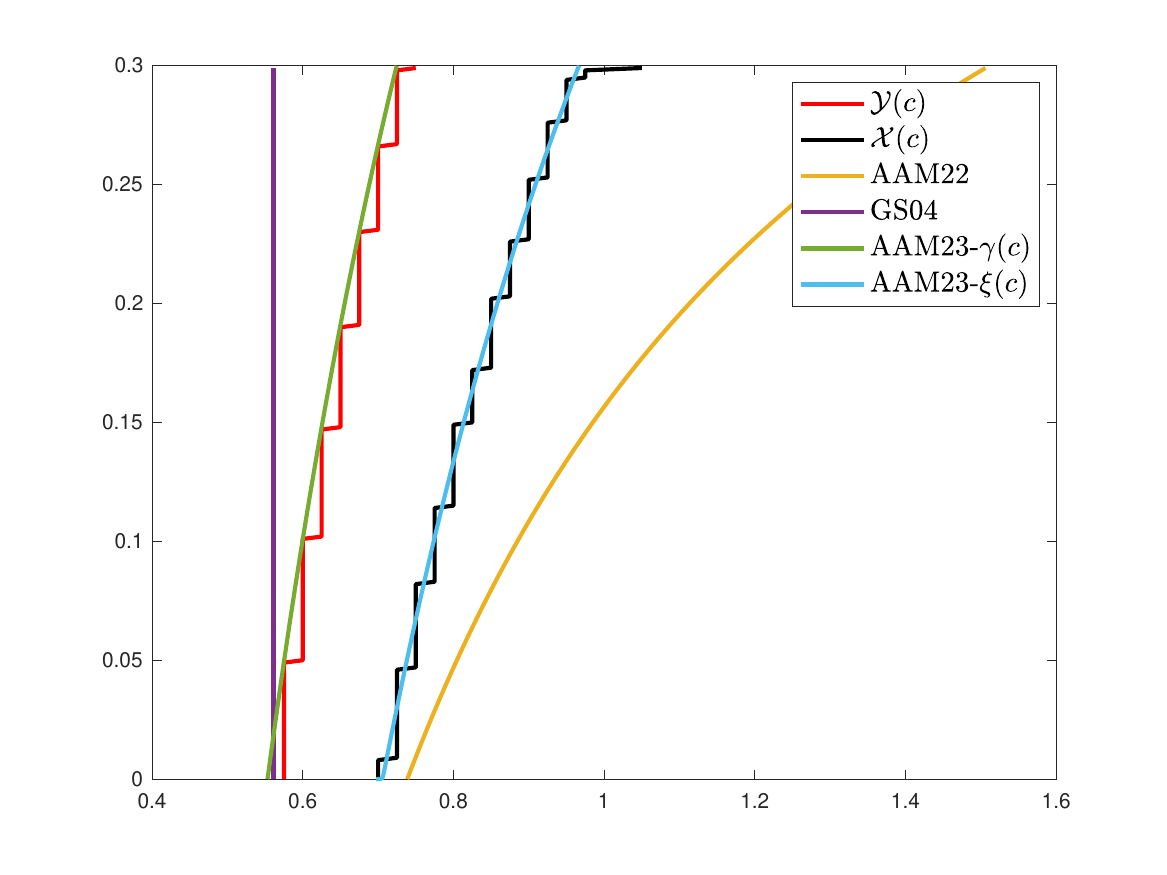}
		\caption{$\mu=0.4$}
		\label{fig:mu=0.4}
	\end{subfigure}
	\caption{Comparison for dividend payout polices with different dfift levels. The left purple and the right orange curve corresponding to two extreme cases $b=0$ and $b=1$ derived from \eqref{analytical_threshold_GS04} and \cite{albrecher2022optimal}, the green and blue cure were solved from the two curved strategies by \cite{albrecher2023optimal}, the red and black curve were the numerical out of our PDEs.}
	\label{fig:varying mu}
\end{figure}

\begin{figure}[H]
	\centering
	\begin{subfigure}[b]{0.48\textwidth}
		\centering
		\includegraphics[width=\linewidth]{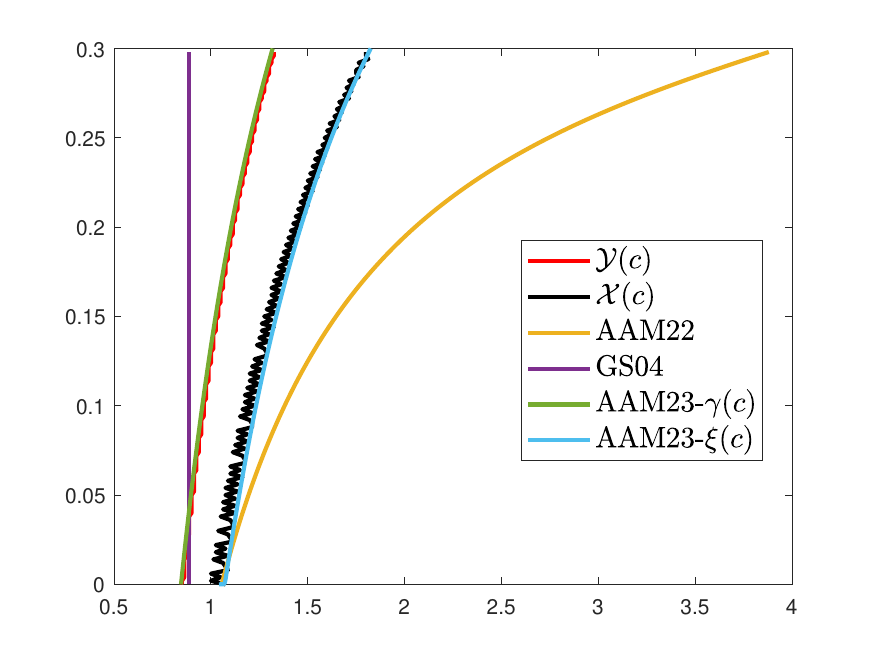}
		\caption{$r=0.03$}
		\label{fig:r=0.03}
	\end{subfigure}
	~~
	\begin{subfigure}[b]{0.48\textwidth}
		\centering
		\includegraphics[width=\linewidth]{b05_graph_new.pdf}
		\caption{$r=0.05$}
		\label{fig:r=0.05}
	\end{subfigure}
	\bigskip\\
	\begin{subfigure}[b]{0.48\textwidth}
		\centering
		\includegraphics[width=\linewidth]{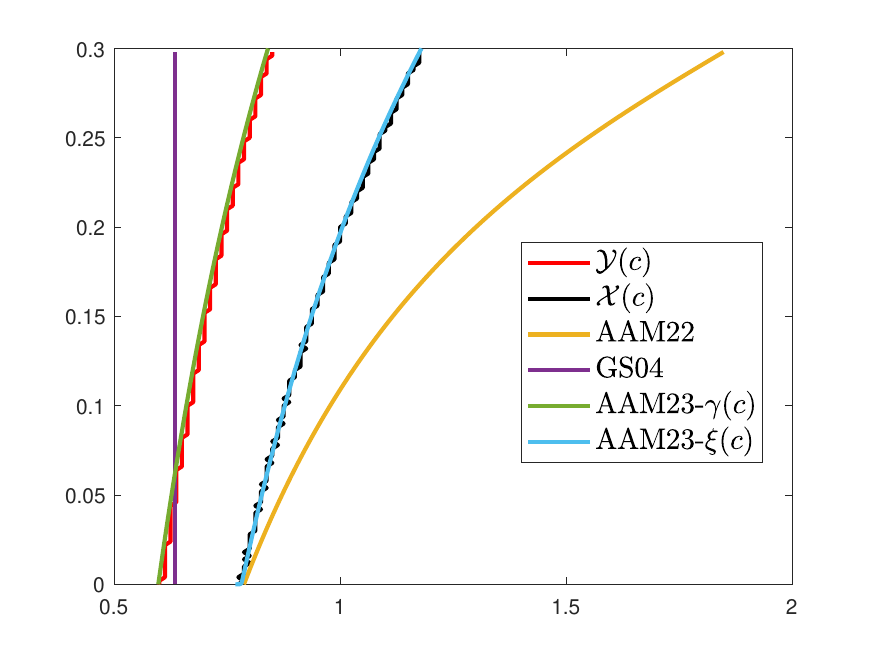}
		\caption{$r=0.08$}
		\label{fig:r=0.08}
	\end{subfigure}
	\caption{Comparison for dividend payout polices with different interest rate levels. The left purple and the right orange curve corresponding to two extreme cases $b=0$ and $b=1$ derived from \eqref{analytical_threshold_GS04} and \cite{albrecher2022optimal}, the green and blue cure were solved from the two curved strategies by \cite{albrecher2023optimal}, the red and black curve were the numerical out of our PDEs.}
	\label{fig:varying r}
\end{figure}

\section{Concluding remarks} 
\label{sec:comparison}
Our work introduces a {novel PDE-based framework} that provides stronger theoretical results and resolves foundational questions left open in \citeauth{albrecher2023optimal}. The key advantages are not merely methodological but substantive, leading to three major improvements:

\begin{enumerate} 
\item {\bf Stronger solutions enabling optimal strategies.} 
The viscosity solution technique in \citeauth{albrecher2023optimal} is a powerful existence tool but is inherently abstract. In contrast, our method proves the existence and uniqueness of a {\bf strong solution} (for the ceiling model) and a {\bf weak solution} (for the no-ceiling model) to the HJB variational inequality. These solution concepts possess significantly higher regularity than a viscosity solution. This is not just a technicality; it is the very reason we can {explicitly derive the optimal feedback control}, which is generally not possible from a viscosity solution alone. This is a fundamental advancement in the analysis of this problem.

\item {\bf A more general and robust proof technique.}
Our approach replaces the delicate parameter-based calculations and verification arguments common in the stochastic control literature with a systematic {PDE a priori estimation method}. We first prove the existence and regularity of the solution, from which the existence and properties of the free boundaries naturally follow. This framework is more general and robust, as it focuses on the intrinsic structure of the PDE rather than on guesswork about the solution's form. {The generality of our method is not merely theoretical; we have already begun applying this same PDE framework to jump-diffusion models, confirming its transferability to settings beyond the pure diffusion case studied in this paper.} This demonstrates that our approach provides a versatile blueprint for a wider class of path-dependent control problems.

\item {\bf Resolution of an open problem on free boundaries.}
The referee correctly identifies the structure of the optimal boundaries as a critical point. \citeauth{albrecher2023optimal} transforms the problem of finding the free boundaries into solving a system of ODEs, but {rigorously proves the existence of this system's solution only under restrictive parameter conditions} (i.e., when $\bar{c}$ is very large). Our method provides a {complete and unconditional proof} of the existence, boundedness, and continuity of the two free boundaries for the full range of model parameters. This resolves a significant theoretical gap in the existing literature.
\end{enumerate}

\citeauth{GX24} corresponds to the special case $b=1$ (non-decreasing dividends), where the linear structure of the problem permits stronger regularity results. For $b<1$ in the current work, the nonlinearity introduced by the drawdown constraint necessitates entirely new analytical techniques, particularly for establishing free boundary regularity. This structural difference is further reflected in the problem dimension: while \citeauth{GX24} deals with a single free boundary, our current model requires the characterization of two free boundaries, representing a substantially more complex free boundary problem. Additionally, our current work covers the singular control case $\bar{c}=\infty$, which was not addressed in \citeauth{GX24}.
\bigskip\\
{\bf Acknowledgements.} The authors would like to thank the editor and reviewers for their valuable comments and suggestions that have led to a much improved version of this paper.

\section*{Declarations}

\textbf{Conflict of interest} The authors have no conflict of interest to declare.


\begin{appendices}

\section{Proofs of \propref{prop:boundary}, \lemref{lem:g}, and \thmref{op2}}\label{sec:A}
\subsection{Proof of \propref{prop:boundary}}\label{proofprop:boundary}
Suppose $g$ fulfills the hypothesis of the proposition.
Then for any admissible strategy $\{\DD_t\}_{t\geq 0}$ in $\Pi_{[c, \cc]}$, we have
\begin{align}
\LL g(X_t)- g'(X_t)\DD_t+\DD_t
&\leq \LL g(X_t)+\sup\limits_{b\cc\leq d\leq \cc}d (1-g'(X_t))=\LL g(X_t)+\cc \TT g(X_t)=0.\label{suppermartingale}
\end{align}
This together with It\^o's formula gives
\begin{align*}
\d (e^{-rt}g(X_t)) &=e^{-rt}(\LL g(X_t)-g'(X_t) \DD_t)\dt+e^{-rt}g'(X_t)\sigma\dw_t\\
&\leq -e^{-rt} \DD_t\dt+e^{-rt}g'(X_t)\sigma\dw_t.
\end{align*}
Since $g'$ is bounded, after integrating on $[0,\tau]$ and taking expectation, it follows
\begin{align*}
g(x)\geq \E\bigg[e^{-r\tau}g(X_{\tau})+\int_0^{\tau}e^{-rt} \DD_t\dt\;\bigg|\; X_0=x\bigg].
\end{align*}
Because $r>0$ and $g$ is bounded, it follows
$e^{-r\tau}g(X_{\tau})1_{\tau=+\infty}=0;$
and since $g(0)=0$, it gives
$$e^{-r\tau}g(X_{\tau})1_{\tau<\infty}=e^{-r\tau}g(0)1_{\tau<\infty}=0.$$
So
\begin{align*}
g(x)\geq \E\bigg[\int_0^{\tau}e^{-rt} \DD_t\dt\;\bigg|\; X_0=x\bigg].
\end{align*}
This shows that $g$ is an upper bound for the value of the problem \eqref{boundary}, namely $g\geq V$.

If we choose the feedback dividend payout strategy $\DD^*_t$ as in \eqref{boundarycontrol}, then
$$\DD^*_t=\argmax\limits_{b\cc\leq d\leq \cc}d (1-g'(X_t)).$$
Also, all the above inequalities become equations, so
$$g(x)=\E\bigg[\int_0^{\tau}e^{-rt} \DD^*_t\dt\;\bigg|\; X_0=x\bigg].$$
This together with $g\geq V$ shows that $\{\DD^*_{t}\}_{t\geq0}$ is an optimal strategy for the problem \eqref{boundary}, and $g$ is the value function.

\subsection{Proof of \lemref{lem:g}}\label{sec:g}
We first show that the inequality $\cc \ga \leq r$ holds if and only if $2\mu\cc \leq \si^2 r$.
Indeed, the quadratic function $f(y)=-\frac{1}{2}\si^2 y^2 +(\mu-\cc) y+r$ satisfies $f(y)> 0$ when $y\in (0,\ga)$ and $f(y)\leq 0$ when $y\geq \ga$, so $\cc \ga \leq r$ is equivalent to $f(r/\cc)\leq 0$, i.e. $2\mu\cc\leq \si^2r$.

The estimates $0\leq g\leq \cc/r$ and $0\leq g'\leq g'(0)<\infty$ can be easily deduced from the strictly monotonicity of $g$ and $g'$ on $\R^{+}$. So we only need to establish the latter and $g\in C^2(\R^+)$ satisfies \eqref{Lg}.

Suppose $2\mu\cc \leq \si^2 r$. Then $\cc \ga \leq r$, and evidently $g\in C^2(\R^+)$ and $g'$ is strictly decreasing on $\R^{+}$. Since $\ga>0$, we have
$$g'(x)=\frac{\cc \ga}{r}e^{-\ga x}< 1,~~ x>0.$$
It thus yields
$$\TT g=b(1-g')+(1-b)(1-g')^+=1-g',$$
and
\begin{align*}
-\LL g-\cc\TT g=-\frac{1}{2}\si^2 g''-\mu g'+r g -\cc(1-g')=0,
\end{align*}
where the last equation is due to that $\ga$ is a root for the equation
$$-\frac{1}{2}\si^2 \ga^2+(\mu-\cc)\ga+r=0.$$
So $g$ satisfies \eqref{Lg}. This completes the proof for the first case $\cc \ga \leq r$.

In the rest of the proof, we assume $2\mu\cc > \si^2 r$ so that $\cc \ga > r$.

We first prove that
\begin{align}\label{k2la2}
\la_2 k_2 \geq\frac{\la_1}{\la_1+\la_2},
\end{align}
which in particular implies $k_2>0$.
In fact, from the trivial inequality
$$\mu-\cc \geq -\sqrt{(\mu-\cc)^2+2 \si^2 r},$$
we get 
\begin{align*}
\frac{\mu-b\cc}{r}-\frac{(1-b)\cc}{r}=\frac{2(\mu-\cc)}{2r}\geq\frac{(\mu-\cc)-\sqrt{(\mu-\cc)^2+2 \si^2 r}}{2r}=-\frac{1}{\ga}.
\end{align*}
Since
$\la_{1}$ and $-\la_{2}$ are the two roots for the equations
$$ -\frac{1}{2}\si^2 \la^2-(\mu-b\cc) \la+r=0,$$
we get from the Viete theorem that 
\begin{align*}
\frac{\la_1+(-\la_2)}{\la_1(-\la_2)}=\frac{\mu-b\cc}{r}\geq \frac{(1-b)\cc}{r}-\frac{1}{\ga}.
\end{align*}
This together with $\la_1, \la_2>0$ gives
\begin{align*}
\frac{\la_2}{\la_1+\la_2}\left[1-\la_1\left(\frac{(1-b)\cc}{r}-\frac{1}{\ga}\right)\right]
\geq\frac{\la_2}{\la_1+\la_2}\left[1-\la_1 \frac{\la_1+(-\la_2)}{\la_1(-\la_2)}\right]
=\frac{\la_1}{\la_1+\la_2},
\end{align*}
which yields the desired inequality \eqref{k2la2}.

We next prove the existence of $y_0$.
Write
\begin{align} \label{funf}
f(y)=k_1 e^{-\la_1 y}-k_2 e^{\la_2 y}+\frac{b\cc}{r}.
\end{align}
Since $\la_1, \la_2, k_{2}>0$ and $\la_1 k_1+\la_2 k_2=1$, it follows
\begin{align*}
f'(y)&=-\la_1 k_1 e^{-\la_1 y}-\la_2 k_2 e^{\la_2 y}
\leq -\la_1 k_1 e^{-\la_1 y}-\la_2 k_2 e^{-\la_1 y}=-e^{-\la_1 y}<0,~~y\geq 0.
\end{align*}
Thanks to $\la_1, \la_2, k_{2}>0$ and $\cc \ga > r$, we have $f(+\infty)=-\infty$ and 
\begin{align*}
f(0)&=k_1 -k_2+\frac{b\cc}{r}=\frac{\cc}{r}-\frac{1}{\ga}>0.
\end{align*}
Hence, $f$ admits unique positive root $y_0$.

By the definitions of $y_0$ and $g$, we have $g(0)=0.$ Because
$$
k_1 -k_2+\frac{b\cc}{r}=\frac{\cc}{r}-\frac{1}{\ga}, \quad \la_1 k_1+\la_2 k_2=1,
$$
we see that $g$ and $g'$ are continuous at $y_0$, and so are on $\R^+$. In particular, $g'(y_0)=1$.

Using the same argument for the previous case, one can show that $g$ satisfies \eqref{Lg} on $[y_0, \infty)$ and clearly $g'$ is strictly decreasing there.
We now show that the same properties hold on $(0,y_0)$ too.
Using the fact that $\la_{1}$ and $-\la_{2}$ are the roots for the equations
$ -\frac{1}{2}\si^2 \la^2-(\mu-b\cc) \la+r=0,$
one can show that
\begin{align}\label{Lg2}
-\frac{1}{2}\si^2 g''-\mu g'+r g -b\cc(1-g')=0,~~ x\in[0, y_0].
\end{align}
We now suppose that $g'$ is strictly decreasing on $(0,y_0)$, then since $g'(y_0)=1$, we have $g'\geq 1$ on $[0,y_{0}]$, which together with \eqref{Lg2} implies that $g$ satisfies \eqref{Lg}.

To show that $g'$ is strictly decreasing on $(0,y_{0})$, we fix arbitrary $x$ and $y$ such that $0<x<y<y_{0}$. Then $e^{-\la_2 t}>1>e^{\la_1t}$ for any $t\in (x-y_{0}, y- y_{0})$, so
\begin{align*}
\frac{e^{-\la_2(x-y_0)}-e^{-\la_2(y-y_0)}}{\la_2}=\int_{x-y_0}^{y-y_0} e^{-\la_2 t}\dt>\int_{x-y_0}^{y-y_0} e^{\la_1t}\dt
=\frac{e^{\la_1(y-y_0)}-e^{\la_1(x-y_0)} }{\la_1}.
\end{align*}
After rearrangement it reads
\begin{align*}
\frac{\la_1}{\la_2 }> \frac{e^{\la_1(y-y_0)}-e^{\la_1(x-y_0)} }{e^{-\la_2(x-y_0)}-e^{-\la_2(y-y_0)}},
\end{align*}
which yields
\begin{align*}
\frac{\la_1}{\la_1+\la_2}> \frac{e^{\la_1(y-y_0)}-e^{\la_1(x-y_0)} }{e^{-\la_2(x-y_0)}-e^{-\la_2(y-y_0)}+e^{\la_1(y-y_0)}-e^{\la_1(x-y_0)}}.
\end{align*}
Recalling the inequality \eqref{k2la2}, we deduce
\begin{align*}
\la_2 k_2> \frac{e^{\la_1(y-y_0)}-e^{\la_1(x-y_0)} }{e^{-\la_2(x-y_0)}-e^{-\la_2(y-y_0)}+e^{\la_1(y-y_0)}-e^{\la_1(x-y_0)}},
\end{align*}
which can also be rewritten as
\begin{align*}
(1-\la_2 k_2) (e^{\la_1(x-y_0)}-e^{\la_1(y-y_0)})+\la_2 k_2 (e^{-\la_2(x-y_0)}-e^{-\la_2(y-y_0)})>0.
\end{align*}
Since $\la_1 k_1+\la_2 k_2=1$, it follows
\begin{align*}
\la_1 k_1 (e^{\la_1(x-y_0)}-e^{\la_1(y-y_0)})+\la_2 k_2 (e^{-\la_2(x-y_0)}-e^{-\la_2(y-y_0)})>0,
\end{align*}
leading to
\begin{align*}
g'(x)=\la_1 k_1 e^{\la_1(x-y_0)}+\la_2 k_2 e^{-\la_2(x-y_0)}>\la_1 k_1 e^{\la_1(y-y_0)}+\la_2 k_2 e^{-\la_2(y-y_0)}=g'(y).
\end{align*} This confirms that $g'$ is strictly decreasing on $(0,y_{0})$.

Finally, since $g$ and $g'$ are continuous and $g$ satisfies \eqref{Lg} on $\R^{+}$,
we conclude that $g\in C^2(\R^{+})$. This completes the proof.

\subsection{Proof of \thmref{op2}}\label{sec:op2}

If $2\mu\cc \leq \si^2 r$, then $\cc \ga \leq r$ and $g'\leq 1$.
Therefore, for any strategy $\{\DD_t\}_{t\geq 0}\in\Pi_{[c,\cc]}$, we have
\begin{align*}
\LL g(X_t)- g'(X_t)\DD_t+\DD_t
&\leq \LL g(X_t)+\sup\limits_{0\leq \DD\leq \cc}\DD (1-g'(X_t))\\[3mm]
&=\LL g(X_t)+\cc (1-g'(X_t))=\LL g(X_t)+\cc \TT g(X_t)=0.
\end{align*}
Repeating the proof of \propref{prop:boundary} with
\eqref{suppermartingale} replaced by the above inequality, one can prove the claim.

\section{An approximation regime switching problem}\label{sec:approximation}
The proof of \thmref{thm:u} is very delicate. To this end, we first introduce an approximation regime switching problem and study its properties in this section.

We notice that \eqref{v_pb} is a PDE, which is generally harder to deal with than ODEs of the same order. Following Albrecher et al. \cite{albrecher2022optimal}, we reduce it to the study of a sequence of ODEs. We use a pure differential equation methods to study it, so that a better solution will be given.

Let us explain our approach in details as follows. Since $\DD_t$ is increasing, the admissible strategy set is shrinking as the increasing of $\DD_t$. If there are only finite discrete values for $\DD_t$, one can treat the cases one by one backwardly: from a lager value to a smaller one. We have already done the first step where $\DD_t$ achieves its maximum $\cc$ that is the boundary case studied in \secref{sec:boundary}.
We can deal with the case where $\DD_t$ is fixed as a smaller value than $\cc$. Since the running maximum is fixed, the HJB equation reduces to an ODE, which can be treated similarly to the boundary case studied in \secref{sec:boundary}. By mathematical induction, we can solve the discrete value case. Once this is done, we will take a limit argument to derive the desired solution to the original HJB equation \eqref{v_pb}.

We now realize the above idea.
Suppose the dividend payout rate can only take the following discrete values $$c_i=\cc- i\Dc ,~~ i=0,1,2,\cdots,n,$$ with $\Dc=\cc/n>0$.
Let $v_{-1}=0$, we consider the following regime switching ODE system:
\begin{align}\label{vi_pb}
\begin{cases}
	\min\{-\LL v_i-c_i \TT v_i, \; v_i-v_{i-1}\}=0, & x>0,~~ i=0, 1,2,\cdots,n,\\[3mm]
	v_i(0)=0, & i=0,1,2,\cdots,n,
\end{cases}
\end{align}
under bounded growth condition.
Here, we can think of $v_i(x)$ as an approximation of $v(x,c_i)$.

Clearly, the system \eqref{vi_pb} can be solved recursively: for each $i$, it is a single-obstacle ODE problem if $v_{i-1}$ is known. One could easily construct a numerical scheme (such as the standard penalty method in \cite{forsyth2002quadratic}) to generate the solution $v_i(x)$ to \eqref{vi_pb} given $v_{i-1}$ is known. The free boundaries $\sw(\cdot)$ and $\sd(\cdot)$ can thereby be connected for all $i=1,2,...,n$ as shown in the previous section.

\begin{lemma}\label{lem:vi}
The system \eqref{vi_pb} has a solution $v_i\in W^2_{p,\rm loc}(\R^+)\cap C^{1+\al}(\R^+),\;i=0,1,2,\cdots,n$ for any $p>1$ and $\al=1-1/p, $ which satisfies
\begin{gather}\label{vi}
	0\leq v_i \leq \frac{\cc}{r},\\\label{vix}
	v_i' \geq 0,
\end{gather}
and
\begin{align}\label{vixx}
	v_i'(y)\leq \max\{v_i'(x),1\}, \quad\forall\; 0\leq x\leq y.
\end{align}
Moreover, for each $i=1,2,\cdots, n,$ there is a free boundary point $x_i>0$ such that $v_i(x)=v_{i-1}(x)$ if $x\geq x_i$ and $v_i(x)>v_{i-1}(x)$ if $x<x_i$.
Also, $v'_i(x)\leq 1$ if $x\geq x_i$.
\end{lemma}
\begin{proof}
It follows from \eqref{Lg}, $c_0=\cc$ and $g\geq 0=v_{-1}$ that
\begin{align*}
	\begin{cases}
		\min\{-\LL g-c_0 \TT g, \; g-v_{-1}\}=0, & x>0, \\
		g(0)=0. &
	\end{cases}
\end{align*}
Therefore, the solution to \eqref{vi_pb} if $i=0$ is $v_0=g$.

By the standard penalty method and $L_p$ theorem, we can prove
$v_i\in W^2_{p,\rm loc}(\R^+)$
step by step from $i=1,2,\cdots,n$.
By the Sobolev embedding theorem, we also have
$v_i\in C^{1+\al}(\R^+),$
since the process is standard (see, e.g., \cite{friedman1975parabolic}), we omit the details.

From \eqref{vi_pb}, we have
\begin{align}\label{increasingvi}
	v_i\geq v_{i-1}\geq\cdots \geq v_0\geq 0,
\end{align}
so the lower bound in \eqref{vi} holds. Furthermore, since $v_i(x)-v_{i}(0)\geq v_{0}(x)-v_{0}(0)=g(x)-g(0)$ for $i\geq 0$, it follows that $v'_i(0)\geq g'(0)>g'(y_0)=1$
by recalling that $\cc \ga > r$ and \lemref{lem:g}.

We now prove the remaining claims by mathematical induction.
By \lemref{lem:g}, we have $v_0(x)=g(x) $ satisfies \eqref{vi}-\eqref{vixx}.
Suppose all the desired results of the lemma hold for some $i=j-1$ ($1\leq j<n$), we are going to prove they also hold if $i=j$.

By the induction hypothesis, we have $v_{j-1}\leq \cc/r$. It is clearly that the constant function $\cc/r$ is a super solution to \eqref{vi_pb} for $v_j$, so $v_j$ also satisfies the upper bound in \eqref{vi}.

We now prove the set $$ E_{j}:=\{x> 0\;|\;v_j(x)=v_{j-1}(x) \}$$ is not empty. Suppose, on the contrary it was empty, then by \eqref{vi_pb} we would have $-\LL v_j- c_j \TT v_j=0$ on $\R^+$. Applying the comparison principle we can prove that $v_j\leq c_j/r$ under the bounded growth condition and $v_j(0)=0$, which contradicts to the order \eqref{increasingvi} since $v_0(+\infty)=\cc/r>c_j/r$.
Therefore, $E_{j}$ is not empty, and consequently, $x_j:=\inf E_{j}<\infty.$

We next prove that
\begin{align}\label{claim1}
	\mbox{there is no $x>0$ such that $v_j(x)=v_{j-1}(x)$ and $v_{j-1}'(x)> 1$.}
\end{align}
Suppose, on the contrary, there is some $x>0$ such that $v_j(x)=v_{j-1}(x)$ and $ v_{j-1}'(x) > 1$.
There are two cases: either there is $1\leq k\leq j-1$ such that $$v_{j}(x)=v_{j-1}(x)=\cdots=v_{j-k}(x)>v_{j-k-1}(x);$$
or $$v_j(x)=v_{j-1}(x)=\cdots=v_{0}(x).$$
In the former case, since $v_{j-k}$ and $v_{j-k-1}$ are continuous, there is a neighborhood $I_\ep=(x-\ep, x+\ep)$ such that $v_{j-k}>v_{j-k-1}$ in $I_\ep,$
so that
\begin{align}\label{Iep}
	-\LL v_{j-k}-c_{j-k}\TT v_{j-k}=0~ \hbox{ in } I_\ep.
\end{align}
Whereas in the latter case, we set $k=j$ and the above equation still holds since $v_{j-k}\equiv v_{0}$.

Combining \eqref{vi_pb} and \eqref{Iep} yields
\begin{align}\label{v''}
	\frac{\si^2}{2}(v_j-v_{j-k})''\leq -\mu (v_j'-v_{j-k}')+r(v_j-v_{j-k})-c_j(\TT v_j-\TT v_{j-k})+k\Dc \TT v_{j-k}
	~\hbox{a.e. in } I_\ep.
\end{align}
Note that both functions $v_j-v_{j-k}$ and $v_{j-1}-v_{j-k}$ attain their minimum value $0$ at $x$, so
$$v_j(x)=v_{j-k}(x),~~ v_j'(x)=v_{j-k}'(x)=v_{j-1}'(x)>1.$$
It implies the right hand side (RHS) of \eqref{v''} is negative at $x$. Since the RHS of \eqref{v''} is continuous in $x$, we have $(v_j-v_{j-k})''<0$ a.e. in $I_\ep$ if $\ep$ is sufficiently small. It follows that $v_j-v_{j-k}$ is strictly concave in $I_\ep$, which contradicts that $v_j-v_{j-k}$ attains its minimum value at $x$. Hence we have \eqref{claim1}.

We claim
\begin{align}\label{v'jxj<=1}
	v_{j-1}'(x_{j})\leq 1.
\end{align}
Indeed, if $x_j>0$, then $x_j\in E_j$ and \eqref{claim1} implies $v_{j-1}'(x_{j})\leq 1$. If $x_j=0$, then there exists $y_n\in E_j$ such that $y_n\to x_{j}$, which also implies $v_{j-1}'(x_{j})\leq 1$ since $v_{j-1}'(y_n)\leq 1$ and $v_{j-1}'$ is continuous. 
By the induction hypothesis, \eqref{vixx} holds at $i=j-1$, so $v_{j-1}'(x)\leq \max\{v_{j-1}'(x_{j}),1\}$ for any $x\geq x_{j}$.
Combining these estimates we conclude
\begin{align}\label{vjx<=1}
	v_{j-1}'(x)\leq 1,~x\geq x_j.
\end{align}

We next prove
\begin{align}\label{vjj}
	v_j(x)=v_{j-1}(x),~~ \forall x\geq x_j.
\end{align} 
Note that the function $\nu(x)=v_j(x),\;x\geq x_j$ satisfies
\begin{align}\label{vjj_pb}
	\begin{cases}
		\min\{-\LL \nu- c_j \TT \nu , \; \nu-v_{j-1}\}=0, & x>x_j,\\[3mm]
		\nu(x_j)=v_{j-1}(x_j).
	\end{cases}
\end{align}
Notice \eqref{vjx<=1} implies $\TT v_{j-1}(x)\geq 0$ for $x\geq x_j$.
It thus follows
\begin{align*}
	-\LL v_{j-1}(x)- c_j \TT v_{j-1}(x)
	\geq -\LL v_{j-1}(x)- c_{j-1}\TT v_{j-1}(x) \geq 0,~~ x>x_j,
\end{align*}
so $\nu(x)=v_{j-1}(x)$ is another solution to \eqref{vjj_pb}. By the uniqueness of solution to \eqref{vjj_pb}, we conclude \eqref{vjj}. Therefore, $x_j$ is the unique free boundary point for $v_j$.

We come to prove $v_j'\geq 0$. Since $v_j=v_{j-1}$ in $[x_j,+\infty)$ and $ v_{j-1}'\geq 0$ by the induction hypothesis, we only need to prove
\begin{align}\label{vjx>=0}
	v_j'\geq 0~\hbox{in}\; [0, x_j].
\end{align}
Differentiating the equation in \eqref{vi_pb} yields
\begin{align}\label{vjx_eq}
	-\LL ( v_j')+c_j \big(b+(1-b)h\big) v_j''=0~~ \hbox{a.e. in}\; (0,x_j),
\end{align}
where $h=1_{\{ v_j'\leq 1\}}.$
Becasue $v_j'(0)\geq 0$ and $v_j'(x_j)=v_{j-1}'(x_j)\geq 0$,
the comparison principle leads to \eqref{vjx>=0}.

Finally, we prove \eqref{vixx} when $i=j$. By \eqref{vjx<=1}, we know \eqref{vixx} holds for $y\in [x_j,+\infty)$. Fix any $0<x<x_{j}$.
Applying \eqref{vjx<=1}, it is easy to check that the constant function $\max\{v_j'(x),1\}$ is a super solution of \eqref{vjx_eq} in $[x,x_j]$, so \eqref{vixx} also holds for $y\in [x,x_j]$. 
The proof is complete.
\end{proof}

\begin{lemma}\label{lem:vix_ub}
The solution $v_i$ to \eqref{vi_pb} satisfies
\begin{align}\label{vix_ub}
	0\leq v_i'\leq K,~~ i=0,1,\cdots, n,
\end{align}
where $K>0$ is a constant independent of $n$ and $i$.
\end{lemma}
\begin{proof}
The lower bound comes from \eqref{vix}.
By \lemref{lem:g}, we also have $v'_0=g'\leq g'(0)<\infty$.

The following variational inequality
\begin{align*} 
	\begin{cases}
		\min\{-\LL \ov, \; \ov'-1\}=0, & x>0,\\[3mm]
		\ov(0)=0,
	\end{cases}
\end{align*}
admits a unique solution
(see \cite{taksar2000optimal})
\begin{align*} 
	\ov(x)=
	\begin{cases}
		K_1 ({\rm e}^{\theta_2 x}-{\rm e}^{-\theta_1 x}),& 0\leq x<x_\infty;\\[3mm]
		K_2+x, & x\geq x_\infty,
	\end{cases}
\end{align*}
where
\[
x_\infty=\frac{2}{\theta_2+\theta_1}\ln \frac{\theta_1}{\theta_2},\\
\]
where $\theta_{1}, \theta_2>0$ satisfy
\[
-\frac{1}{2}\si^2\theta_{1}^2+\mu\theta_{1}+r=0,\quad
-\frac{1}{2}\si^2\theta_{2}^2-\mu\theta_{2}+r=0,
\]
and
\[
K_1=(\theta_2e^{\theta_2x_\infty}+\theta_1e^{-\theta_1 x_\infty})^{-1},~~
K_2=K_1 ({\rm e}^{-\theta_1 x_\infty}-{\rm e}^{\theta_2x_\infty})-x_\infty.
\]
Note that $\ov'\geq 1$, so $\TT \ov\leq 0$. It follows, for $i=1,2,\cdots, n$,
$-\LL \ov-c_i \TT \ov\geq -\LL \ov\geq 0,$
which means $\ov$ is a super solution to \eqref{vi_pb}.
Since $v_i(0)=\ov(0)=0$, we have $ v_i'(0)\leq \ov'(0)<\infty.$
By \eqref{vixx}, we conclude $ v_i'\leq K=\max\{\ov'(0),\; 1\}$.
\end{proof}

\begin{lemma}\label{lem:vi_w2p}
For each $p>1$, we have
\begin{align}\label{vi_w2p}
	|v_i|_{W^2_p[N-1,N]}\leq K, \quad i=0,1,\cdots, n,
\end{align}
where $K>0$ is a constant independent of $n$, $i$ and $N\in \N$.
\end{lemma}
\begin{proof}
We can rewrite \eqref{vi_pb} as
\begin{align}\label{Lvi}
	-\LL v_i=f:=\Big(\sum\limits_{j=1}^i c_j 1_{\{ v_{j-1}< v_i\leq v_j\}}+\cc 1_{\{ v_i=v_0\}} \Big)\TT v_i~~ \text{a.e.}.
\end{align}
From \eqref{vix_ub}, we know the RHS is uniformly bounded, so the $L_p$ estimation (See e.g., \cite{GT11} Theorem 9.13 on page 239) gives
$$
|v_i|_{W^2_p[N-1,N]}\leq C(|v_i|_{L_p[N-1,N]}+|f|_{L_p[N-1,N]})
$$
for some constant $C>0$, together with \eqref{vi} gives \eqref{vi_w2p}.
\end{proof}

The Sobolev embedding theorem and \lemref{lem:vi_w2p} imply the following.
\begin{corollary}\label{cor:vi_c1}
For each $0<\al<1$, we have
\begin{align}\label{vi_c1}
	|v_i|_{C^{1+\al}(\R^+)}\leq K, \quad i=0,1,\cdots, n,
\end{align}
where $K>0$ is a constant independent of $n$ and $i$.
\end{corollary}

Since we will take a limit argument to construct a solution to \eqref{v_pb} from that of \eqref{vi_pb}, we first derive more properties of the solution to \eqref{vi_pb}.

Define
$$u_i:=\frac{v_i-v_{i-1}}{\Dc },~~ i=1,2,\cdots,n.$$

\begin{lemma}\label{lem:ui_b}
We have
\begin{align}\label{ui_b}
	0\leq u_i\leq K,~~ i=1,\cdots, n,
\end{align}
where $K>0$ is a constant independent of $n$ and $i$.
\end{lemma}
\begin{proof}
The lower bound comes from \eqref{vi_pb}.
Let $K$ be given in \lemref{lem:vix_ub}.
Thanks to \eqref{vix_ub}, we have $\TT v_{i-1}\geq b(1-v'_{i-1})\geq b(1-K).$
Let $\ov_i=v_{i-1}+\frac{b(K-1)}{r} \Dc $, then
\begin{align*}
	-\LL\ov_i-c_i \TT \ov_i
	=&-\LL v_{i-1}-c_i \TT v_{i-1}+b(K-1) \Dc\\
	=&-\LL v_{i-1}-c_{i-1} \TT v_{i-1}+(b(K-1)+\TT v_{i-1}) \Dc\geq 0,
\end{align*}
namely, $\ov_i$ is a super solution to \eqref{vi_pb}, so $v_i\leq \ov_i=v_{i-1}+\frac{b(K-1)}{r} \Dc $. The claim thus follows by redefining $K$.
\end{proof}

We use the following function in the rest of this paper:
\begin{align}\label{H}
H(x, y)=-\frac{(1-x)^+-(1-y)^+}{x-y} 1_{\{x\neq y\}},~~ x, \ y\in\R.
\end{align}
It is clear that $0\leq H\leq 1$.

\begin{lemma}\label{lem:uix_b}
If there exists a constant $a>0$ such that $x_i>a$ for all $i=1,\cdots,n$,
we then have
\begin{align}\label{uxi_b}
	| u_i'|\leq K \quad\mbox{ in }\; [0, x_{i-1}], \quad i=1,\cdots, n,
\end{align}
where $K>0$ is a constant independent of $n$ and $i$.
\end{lemma}
\begin{proof}
Since $ u_i'=0$ in $[x_i,+\infty)$, we only need to prove \eqref{uxi_b} holds in $[0,x_i]\cap [0,x_{i-1}].$
Denote $\wx_i=\min\{x_i,x_{i-1}\}$.
The difference of the equations of $v_i$ and $v_{i-1}$ in $[0,\wx_i]$ leads to
\begin{align*}
	-\LL u_i+c_{i-1} [b+(1-b)h_i] u_i'=-\TT v_i~\hbox{in}\; [0, \wx_i],
\end{align*}
where $h_i=H( v_i', v_{i-1}')\in [0,1]$.
\lemref{lem:vix_ub} and \lemref{lem:ui_b} imply that $|\TT v_i|$ and $|u_i|$ are uniformly bounded and independent of $n$ and $i$, { so for any $[x,x+a]\subset [0,\wx_i]$, applying the $L_p$ estimation, we obtain $|u_i|_{W^2_p[x,x+a]}$ is bounded by some constant independent of $n$ and $i$.} Then the Sobolev embedding theorem shows that $|u_i|_{C^{1+\al} [x,x+a]}$ is also bounded and independent of $n$ and $i$. The claim then follows.
\end{proof}

\begin{lemma}\label{lem:uci_b}
If the condition of \lemref{lem:uix_b} holds,
we then have
\begin{align}\label{ui_b2}
	u_{i-1}\leq u_i+B \Dc, \quad i=2,\cdots, n,
\end{align}
where $B>0$ is a constant independent of $n$ and $i$.

\end{lemma}
\begin{proof}
Since $u_{i-1}=0$ by \lemref{lem:vi} and $u_i\geq0$ in $[x_{i-1},\infty)$, we only need to prove that \eqref{ui_b2} holds in $[0,x_{i-1}]$.
By the equation of $v_{i-1}$ we know
\begin{align*} 
	-\LL v_{i-1}-c_{i-1} [b+(1-b)h ](1- v_{i-1}')=0~\hbox{in}\; [0,x_{i-1}],
\end{align*}
where $h=1_{\{ v_{i-1}'\leq 1\}}.$
On the other hand, by the equation of $v_i$ in \eqref{vi_pb} and noting that $(1- v_i')^+\geq (1- v_i')h$, we have
$-\LL v_i- c_i [b+(1-b) h] (1- v_i')\geq 0.$
Dividing the difference of the above two estimates by $\Dc$ gives
$$-\LL u_i+c_{i-1} [b+(1-b) h ] u_i'+[b+(1-b) h ] (1- v_i')\geq0~\hbox{in}\; [0,x_{i-1}].$$
Similarly, by the equation of $v_{i-2}$ in \eqref{vi_pb} and $(1- v_{i-2}')^+\geq (1- v_{i-2}')h$
we have
\begin{align*} 
	-\LL v_{i-2}- c_{i-2} [b+(1-b) h] (1- v_{i-2}')\geq 0,
\end{align*}
we can deduce
$$-\LL u_{i-1}+c_{i-2} [b+(1-b) h ] u_{i-1}'+[b+(1-b) h ](1- v_{i-1}')\leq0~\hbox{in}\; [0,x_{i-1}].$$
So $w_i:=(u_i-u_{i-1})/\Dc$ satisfies
$$-\LL w_i+c_{i-2} [b+(1-b) h ] w_i'\geq 2[b+(1-b) h ] u_i'\geq -2K,~~\hbox{in}\; [0,x_{i-1}],$$
where the last inequality is due to \eqref{uxi_b}. Moreover, because $w_i(0)=0$ and $w_i(x_{i-1})=u_i(x_{i-1})/\Dc\geq 0$, by the comparison principle, we have $w_i\geq -2K/r$ in $[0,x_{i-1}]$, which gives \eqref{ui_b2} with $B=2K/r$.
\end{proof}

\section{Proof of \thmref{thm:u}}\label{sec:solvability}

Now we are ready to prove \thmref{thm:u}.
Fix any $p>1$. Then $\al=1-1/p\in(0,1)$.

For each $n\in \N$, rewrite $v_i(x)$ as $v^n_i(x)$ if $\Dc=\Dc^n:=\cc/n$.
Let $v^n(x,c)$ be the linear interpolation function of $v^n_i(x)$, i.e.
$$
v^n(x,c)= \frac{c^n_{i-1}-c}{\Dc^n} v^n_i(x) + \frac{c-c^n_{i}}{\Dc^n} v^n_{i-1}(x),\quad \hbox{if}\; c^n_{i}< c \leq c^n_{i-1},
$$
where $c^n_{i}=\cc- i \Dc^n$.
\lemref{lem:vix_ub} and \lemref{lem:ui_b} imply $v^n$ is uniform Lipschitz continuous in $\Q$. Apply the Arzela-Ascoli theorem, there is $v\in C(\Q )$, and a subsequence $v^{n_k}$ such that, for each $L>0$,
$v^{n_k}\longrightarrow v~\hbox{in}\; C([0,L]\times [0,\cc]).$
\lemref{lem:vi_w2p} implies $v(\cdot,c)\in W^2_{p, \rm loc} [0,+\infty)$ for any $c\in [0,\cc]$; the estimate \eqref{vi} implies that $v$ is bounded in $\Q $; \eqref{ui_b} implies that $v$ is non-increasing in $c$. The above shows $v\in \A$.

Furthermore, the estimate \eqref{vc} comes from \eqref{ui_b}. The estimates \eqref{v}-\eqref{vxx} follow from \eqref{vi}-\eqref{vixx} and \eqref{vix_ub}. The estimate \eqref{vL} is a consequence of \eqref{vi_w2p} and \eqref{vi_c1}.

We come to prove that $v$ satisfies \eqref{v_pb}. The boundary conditions are trivially satisfied. 
We now show
$$-\LL v-c \TT v\geq 0.$$
For each $(x,c)\in \Q$, by the construction of $v$, there is $c^k= c^{n_k}_{i_k}$ such that $c^k\to c,$ $v^{n_k}_{i_k}(x)\to v(x,c).$
Moreover, from \lemref{lem:vi_w2p} and \corref{cor:vi_c1} we also have $v^{n_k}_{i_k}(\cdot)$(or its subsequence)$\to v(\cdot,c)$ weakly in $W^2_{p }[N-1,N]$ and uniformly in $C^{1+\al}[N-1,N]$ for each $N\in \N$.
Let $k\to \infty$ in the inequality
$$-\LL v^{n_k}_{i_k}(x)-c^k \TT v^{n_k}_{i_k}(x)\geq 0$$
we get
$$-\LL v(x,c)-c \TT v(x,c)\geq 0.$$

It is only left to show the last requirement \eqref{-Lv=0} in \defref{def:solution}.
Suppose $v(x,c)>v(x,s)$ for all $s\in(c,\cc]$. For each fixed $s\in(c,\cc]$, there is a sequence $s^k= c^{n_k}_{j_k}$ such that $s^k\to s$ and $v^{n_k}_{j_k}(x)\to v(x,s)$.
So for $k$ large enough, $v^{n_k}_{i_k}(x)> v^{n_k}_{j_k}(x)$.
Since $ c^{n_k}_{i_k} \to c < s\leftarrow c^{n_k}_{j_k}$, we have $i_k > j_k$.
There must exist some integer $l_k\in (j_k, i_k]$ such that $v^{n_k}_{l_k}(x)> v^{n_k}_{l_k-1}(x)$.
Thanks to \lemref{lem:vi}, we have
$$v^{n_k}_{l_k}(y)> v^{n_k}_{l_k-1}(y), ~~ y\leq x.$$
It then follows from \eqref{vi_pb} that
$$v^{n_k}_{l_k}(y)-t^k \TT v^{n_k}_{l_k}(y)=0, ~~ y\leq x,$$
where $t^k= c^{n_k}_{l_k}$.
The sequence $t^k$ has a subsequence converging to some $t\in [c,s]$, so that
$$-\LL v(\cdot,t)-t \TT v(\cdot,t)=0~~ \hbox{a.e. in}~[0,x].$$
Let $s\to c$,
we have $$ -\LL v(\cdot,c)-c \TT v(\cdot,c)=0~~ \hbox{a.e. in}~ [0,x],$$
which implies \eqref{-Lv=0}.
So we conclude that $v$ is a solution to \eqref{v_pb}.

We now prove the uniqueness.
Suppose $v,\;w\in \A$ are two solutions to the problem \eqref{v_pb}.
Let
$\phi(x)=e^{\eta x},$
where $\eta>0$ is a small constant such that
\begin{align}\label{Lphi}
-\LL \phi=\Big(-\frac{\si^2}{2}\eta^2-\mu \eta+r\Big)e^{\eta x}\geq 0.
\end{align}
We come to prove
\begin{align}\label{v<=w ep}
v\leq w+\ep \phi
\end{align}
for any $\ep>0.$
Suppose it is not true, then
\begin{align*} 
M=\sup\limits_{(x,c)\in \Q}(v-w-\ep \phi)(x,c)>0.
\end{align*}
Because $v$ and $w$ are bounded, $v-w-\ep \phi$ tends to $-\infty$ uniformly for all $c\in [0,\cc]$ if $x\to+\infty$.
By virtue of its continuity, $v-w-\ep \phi$ must attain its maximum value at some point $(x^*, c^*)$. Without loss of generality, we may assume $c^*$ is the largest $c\in [0,\cc]$ such that
$$(v-w-\ep \phi)(x^*,c)=M.$$
Because $$(v-w-\ep \phi)(x^*,\cc)=-\ep \phi(x^*,\cc)<0<M$$ and $$(v-w-\ep \phi)(0,c^*)=-\ep \phi(0,c^*)<0<M,$$
we have $c^*<\cc$ and $x^*>0$, and consequently,
\begin{align*} 
(v-w-\ep \phi)(x^*,c)<M=(v-w-\ep \phi)(x^*,c^*), \quad c\in(c^*, \cc].
\end{align*}
Since $w$ is non-increasing w.r.t. $c$, the above inequality implies
$$v(x^*,c^*)> v(x^*,c),~~ c\in(c^*, \cc].$$
According to the definition of the solution to \eqref{v_pb}, this implies
\begin{align}\label{strictless}
-\LL v(x,c^*)-c^* \TT v(x,c^*)=0,\quad x\in(0,x^*).
\end{align}
Combining
$
-\LL w(\cdot,c^*)-c^* \TT w(\cdot,c^*)\geq 0 \quad\hbox{a.e. in}\;(0,x^*)
$
and \eqref{Lphi} we conclude
$$
-\LL (v-w-\ep \phi)(\cdot,c^*)\leq c^* [\TT v(\cdot,c^*)- \TT w(\cdot,c^*)]\quad\hbox{a.e. in}\;(0,x^*),
$$
i.e., a.e. in $(0,x^*)$, 
\begin{align}\label{vwxx}
&\frac{\si^2}{2} \p_{xx}(v-w-\ep \phi)(\cdot,c^*)\nonumber\\
\geq &-\mu \p_x (v-w-\ep \phi)(\cdot,c^*)
+ r (v-w-\ep \phi)(\cdot,c^*)- c^* [\TT v(\cdot,c^*)- \TT w(\cdot,c^*)].
\end{align}
Because $x^*$ is an internal maximum point of $(v-w-\ep \phi)(\cdot,c^*)$, we know that
$$\p_x (v-w-\ep \phi)(x^*,c^*)=0.$$
Moreover, due to $\phi'\geq0$,
\begin{align*}
\TT v(x^*,c^*)=\TT (w+\ep \phi)(x^*,c^*)\leq \TT w(x^*,c^*).
\end{align*}
Recalling $(v-w-\ep \phi)(x^*,c^*)=M>0$, so the RHS of \eqref{vwxx} is positive at $x^*$, since the RHS of \eqref{vwxx} is continuous, we can infer that $$\p_{xx}(v-w-\ep \phi)(\cdot,c^*)> 0$$ a.e. in $(x^*-\ep, x^*)$ when $\ep$ is sufficiently small. Then, for $x\in (x^*-\ep, x^*)$, we have $$\p_x(v-w-\ep \phi)(x,c^*)<\p_x(v-w-\ep \phi)(x^*,c^*)=0,$$ which contradicts the fact that $x^*$ is a maximum point of $(v-w-\ep \phi)(\cdot,c^*)$.
Hence, the solution to \eqref{v_pb} is unique.

From now on, we let $v$ be the unique solution to the HJB equation \eqref{v_pb} defined above.
We next show $ v_x$ is continuous in $\Q$, which will complete the proof of \thmref{thm:u}.

By \eqref{vL}, there is a constant $K>0$ such that
\begin{align}\label{vLx}
|v_x(x,c)-v_x(y,c)|\leq K|x-y|^\al,~~ \forall\; x,y\in \R^+,\ c\in[0,\cc].
\end{align}
Now we come to prove there is another constant $K'>0$, such that
\begin{align}\label{vLc}
|v_x(x,c)-v_x(x,s)|\leq K'|c-s|,~~ \forall\; x\in \R^+,\ c, s\in[0,\cc].
\end{align}
Since $v_{xc}=0$ in $\SS$, we only need to prove \eqref{vLc} holds for $(x,c),\;(x,s)\in \NS.$

Let $$u(x)=\frac{v(x,s)-v(x,c)}{c-s},$$ then
\begin{align*}
-\LL u+s [b+(1-b) h ] u_x=-\TT v(\cdot,c)~\hbox{in}\;[0,x],
\end{align*}
where
$h=H(v_x(x,c),v_x(x,s))\in[0,1].$
Note that { \eqref{vc}} implies $u$ is bounded, apply the $L_p$ estimation we have $u$ is bounded in $W^2_p[N-1,N]$ for each $N\in\N$, then apply the embedding theorem, we know $u_x$ is bounded, which gives \eqref{vLc}.
Combining \eqref{vLx} and \eqref{vLc}, we conclude that $v_x$ is continuous in $\Q$, completing the proof of \thmref{thm:u}.


\section{Proof of \propref{prop:characterizationx} }\label{sec:freeboundaries}

In this section, we establish \propref{prop:characterizationx}.
It will be a consequence of the below Lemmas \ref{lem:N} $-$ \ref{lemma:xcb}.

Recall that
\begin{align*}
\SS&=\Big\{(x,c)\in \bQ \;\Big|\; v(x,c)=v(x,s) \text{ for some $s\in(c,\cc]$} \Big\}; \hspace{-32pt}&&\\
\sw(c)&=\inf\Big\{x\in \R^+\;\Big|\; (x,c)\in \SS\Big\}, && c\in[0,\cc);\\
\swb(x,c)&=\max\Big\{s\in[0,\cc]\;\Big|\; v(x,s)=v(x,c)\Big\}\in[c,\cc], &&(x,c)\in \bQ;\\
\minswb(x,c)&=\min\Big\{s\in[0,\cc]\;\Big|\; v(x,s)=v(x,c)\Big\}\in[0,c], &&(x,c)\in \bQ.
\end{align*}

\begin{lemma}\label{lem:N}
If $(x,c)\in\SS$, then
\begin{align}\label{vx<=1}
	v_x(x ,c)\leq 1
\end{align}
and $v(y,s)=v(y, c)$ for all $(y,s)\in [x,+\infty)\times[\minswb(x,c),\swb(x,c)].$ As a consequence, we have
\begin{align}\label{vxxc}
	v_x(\sw(c),c)\leq 1,~c\in [0,\cc),
\end{align}
and
\begin{gather*}
	\Big\{(x,c)\in \bQ \;\Big|\; x> \sw(c)\Big\}\subseteq\SS
	\subseteq\Big\{(x,c)\in \bQ \;\Big|\; x\geq \sw(c)\Big\},\\
	\Big\{(x,c)\in \bQ \;\Big|\; x< \sw(c)\Big\}\subseteq\NS
	\subseteq\Big\{(x,c)\in \bQ \;\Big|\; x\leq \sw(c)\Big\}. 
\end{gather*}
\end{lemma}

\begin{proof}
Since $(x,c)\in\SS$, we have $\os:=\swb(x,c)>c$. If $\os<\cc$, then $(x,\os)\in\NS$ and thanks to \eqref{-Lv=0},
$$-\LL v(\cdot ,\os)-\os\TT v(\cdot ,\os)=0~\hbox{in } (0,x).$$
If $\os=\cc$, then $v(\cdot ,\os)=g(\cdot)$ and $g$ satisfies \eqref{Lg}, so the above equation also holds.
So by \eqref{v_pb},
\begin{align}\label{vxx-vxx}
	\frac{\si^2}{2}(v_{xx}(\cdot ,c)-v_{xx}(\cdot ,\os))
	&\leq -\mu (v_x (\cdot ,\os)-v_x (\cdot ,\os))+r(v(\cdot ,c)-v(\cdot ,\os))\nonumber\\
	&\quad\;-\os(\TT v(\cdot ,c)-\TT v(\cdot ,\os))-(c-\os)\TT v(\cdot ,\os)
	~\hbox{a.e. in } (0,x).
\end{align}
Since $v(\cdot,c)-v(\cdot,\os)$ attains its minimum value 0 at $x$, it gives
$v(x ,c)=v(x ,\os)$ and $v_x(x ,c)=v_x(x ,\os).$
Suppose $v_x(x ,c)>1$. Then the RHS of \eqref{vxx-vxx} is negative at $x$. Since it is continuous at $x$, we have $v_{xx}(\cdot ,c)-v_{xx}(\cdot ,\os)<0$ a.e. in $(x-\ep,x)$ for some small $\ep>0$.
It then follows $v_x(\cdot ,c)-v_x(\cdot ,\os)>0$ in $(x-\ep,x)$, which
together with $v(x ,c)=v(x ,\os)$ yields $v(\cdot ,c)<v(\cdot ,\os)$ in $(x-\ep,x)$,
contradicting that $v$ is non-increasing in $c$. Therefore, we conclude \eqref{vx<=1}.

By \eqref{vxx} and \eqref{vx<=1} we have $v_x(y,\os)\leq 1$ if $y\geq x$ and thus $\TT v(y,\os)\geq 0$ if $y\geq x .$
Therefore, the function $\nu$, defined by
\begin{align*}
	\nu(y,s):=v(y,\os),~~ (y,s)\in [x ,+\infty)\times[\minswb(x,c),\os],
\end{align*}
satisfies
$$-\LL \nu(y,s)-s \TT \nu(y,s) \geq -\LL v(y,\os)-\os \TT v(y,\os)\geq 0.$$
Since $\nu_s=0$, we have
\begin{align*}
	\begin{cases}
		\min\{-\LL \nu-s \TT \nu , \; -\nu_s\}=0, & (y,s)\in [x ,+\infty)\times[\minswb(x,c),\os],\smallskip\\
		\nu(x ,s)=v(x ,s), & s\in [\minswb(x,c),\os],\smallskip\\
		\nu(y,\os)=v(y,\os), & y\geq x.
	\end{cases}
\end{align*}
Similar to the proof of the uniqueness of the solution of the problem \eqref{v_pb}, we can prove that the solution of the above problem is unique, so we have $\nu=v$ in $[x ,+\infty)\times[\minswb(x,c),\os]$.

The estimate \eqref{vxxc} is a consequence of \eqref{vx<=1} and the continuity of $v_{x}$.
The last two claims come from the fact that $(x,c)\in\SS$ implies $(y,c)\in\SS$ for any $y\geq x$.
\end{proof}

\begin{lemma}\label{lemma:xc}
We have
$\sw(c)>0$ for all $c\in[0,\cc).$
\end{lemma}
\begin{proof}
For $x>0$ and $c\in[0,\cc)$, by monotonicity and $v(0,c)=g(0)=0$, we have
\[v(x,c)-v(0,c)=v(x,c)-g(0)\geq v(x,\cc)-g(0)=g(x)-g(0).\]
It hence follows
\begin{align}\label{vx0>1}
	v_x(0,c)\geq g'(0)>g'(y_0)=1
\end{align}
by recalling that $\cc \ga > r$.
Comparing to \eqref{vxxc}, we conclude that $\sw(c)>0$.
\end{proof}

To prove the continuity of $\sw(\cdot)$, we need the following estimate.

\begin{lemma}\label{lem:ucL}
There exists a constant $B>0$ such that, when $$0\leq c_*<c_*+\Dc_*\leq c^*-\Dc^*<c^*\leq \cc,$$ it holds that 
\begin{align}\label{ucL}
	\frac{v(x,c^*-\Dc^*)-v(x,c^*)}{\Dc^*}\leq \frac{v(x,c_*)-v(x,c_*+\Dc_*)}{\Dc_*}+B(c^*-c_*).
\end{align} 
\end{lemma}
\begin{proof}
Suppose that $\{v^{n_k}\}$ is the sequence constructed in \secref{sec:solvability} and it converges to $v$. Moreover, let $x^{n_k}_i$ denote the free boundary point corresponding to $v^{n_k}_i$, where $i = 1,2,\cdots,n_k$.

We first prove that there exists a constant $a > 0$, which is independent of both $k$ and $i$, such that
\begin{align}\label{xnk}
	x^{n_k}_i\geq a.
\end{align}
Conversely, assume that there exists a subsequence $\{y_m\}$, where $y_m := x^{n_{k_m}}_{i_m}$, which converges to $0$ as $m \to \infty$. By virtue of \eqref{v'jxj<=1}, we are aware that $(v^{n_{k_m}}_{i_m})'(y_m)\leq 1$.
Denote $s_m = c^{n_{k_m}}_{i_m}$. Given that $s_m$ is bounded within the interval $[0,\cc]$, there must exist an $s_0 \in [0,\cc]$ and a subsequence of $\{s_m\}$ (which we still denote by $\{s_m\}$ for simplicity) such that $s_m \to s_0$. Thanks to \eqref{vi_c1}, there exists a subsequence of the functions $v^{n_{k_m}}_{i_m}(\cdot,s_m)$ that converges to $v(\cdot,s_0)$ in the space $C^1[0,L]$. Consequently, we can infer that $v_x(0,s_0)\leq 1$, which stands in contradiction to \eqref{vx0>1}.

Denote $$u^{n_k}_i=\frac{v^{n_k}_i-v^{n_k}_{i-1}}{\Dc^{n_k}}, \quad i=1,\cdots, n_k,$$ and suppose
$$c_*\in [c^{n_k}_{i+j+M+1},c^{n_k}_{i+j+M}],\; c_*+\Dc_*\in [c^{n_k}_{i+j},c^{n_k}_{i+j-1}],\; c^*-\Dc^*\in [c^{n_k}_{i},c^{n_k}_{i-1}],
\; c^*\in [c^{n_k}_{i-N+1},c^{n_k}_{i-N}],$$
for some $i,j, M, N>0$.
Since the condition of Lemma \ref{lem:uci_b} is satisfied by virtue of \eqref{xnk}, we are able to apply Lemma \ref{lem:uci_b} and obtain
\begin{align*}
	u^{n_k}_{i-l}\leq u^{n_k}_{i+j+m+1}+(l+j+m+1) B \Dc^{n_k}
\end{align*}
for $l, m\geq 0$, so
\begin{align*}
	\frac{1}{N}\sum_{l=0}^{N-1} u^{n_k}_{i-l}\leq \frac{1}{M}\sum_{m=0}^{M-1} u^{n_k}_{i+j+m+1}+(N+M+j+1) B \Dc^{n_k},
\end{align*}
i.e.
\begin{align*}
	\frac{v^{n_k}_{i}-v^{n_k}_{i-N}}{N \Dc^{n_k}}\leq \frac{v^{n_k}_{i+j+M}-v^{n_k}_{i+j}}{M \Dc^{n_k}}+(N+M+j+1) B \Dc^{n_k},.
\end{align*}
Since $v^{n_k}(x,c)$ is non-increasing w.r.t. $c$, we have
\begin{align*}
	&\frac{v^{n_k}(x,c^*-\Dc^*)-v^{n_k}(x,c^*)}{\Dc^*}
	\leq\frac{v^{n_k}(x,c^{n_k}_{i})-v^{n_k}(x,c^{n_k}_{i-N})}{(N-2) \Dc^{n_k}}\\[3mm]
	\leq& \frac{N}{N-2}\Bigg[\frac{v^{n_k}(x,c^{n_k}_{i+j+M})-v^{n_k}(x, c^{n_k}_{i+j})}{M \Dc^{n_k}}+(N+M+j+1) B \Dc^{n_k}\Bigg]\\[3mm]
	\leq& \frac{N}{N-2}\Bigg[\frac{M+2}{M}\cdot\frac{v^{n_k}(x,c_*)-v^{n_k}(x, c_*+\Dc_*)}{\Dc_*}+\frac{N+M+j+1}{N+M+j-1}B (c^*-c_*)\Bigg].
\end{align*}
Let $k$ approach infinity, and observe that $N$, $M$, $i$, and $j$ also tend to infinity simultaneously, we are then able to derive \eqref{ucL}.
\end{proof}

\begin{lemma}\label{lemma:xcc}
The switching boundary $\sw(\cdot)$ is continuous in $[0,\cc)$.
\end{lemma}
\begin{proof}
We first prove the right limit exits at any point $c\in[0,\cc)$, namely
\begin{align}\label{x_lim}
	\liminf\limits_{s\to c+}\sw(s)=\limsup\limits_{s\to c+}\sw(s).
\end{align}
Suppose, on the contrary, there exist $c\in[0,\cc)$, $x_*$ and $x^*$ such that
$$\liminf\limits_{s\to c+}\sw(s)<x_*<x^*<\limsup\limits_{s\to c+}\sw(s).$$
Then by definition, there are two sequences $c_n\to c+$ and $\Dc_n\to 0+$ such that
\begin{align}\label{cn}
	v(x,c_n+\Dc_n)=v(x,c_n),~~ x\geq x_*,
\end{align}
and a decreasing sequence $s_1>s_2>\cdots>s_n\to c$ such that
$$-\LL v(x,s_n)-s_n \TT v(x,s_n)=0,~~ x< x^*.$$

Let $$u^n(x)=\frac{v(x,s_n)-v(x,s_{n-1})}{s_{n-1}-s_n}\geq0, $$ then it satisfies
\begin{align}\label{uDsn}
	-\LL u^n+s_{n-1} [b+(1-b) h_n] u^n_x=-\TT v(\cdot,s_n),~~ x< x^*,
\end{align}
where $$h_n=H(v_x(x,s_n),v_x(x,s_{n-1}))\in[0,1].$$
Due to \eqref{vL} and \eqref{vc}, $\TT v(\cdot,s_n)$ and $u^n$ are bounded and independent of $n$, by the $L_p$ estimation we know $u^n$ is uniformly bounded in $W^2_p[0,x^*]$.
On the other hand, since $c_m +\Dc_m \to c$ and $s_n>c$, we can pick $m$ sufficiently large to make $c_m +\Dc_m< s_n$. Using \lemref{lem:ucL} and \eqref{cn} we have
$$0\leq u^n(x)\leq \frac{v(x,c_m)-v(x,c_m+\Dc_m)}{\Dc_m} + B(s_{n-1}-c) = B(s_{n-1}-c),~~ x>x_*.$$
So we have $u^n(x) \to 0,$ $x>x_*.$
so it has a subsequence converges to $0$ weekly in $W^2_p[0,x^*]$ and uniformly in $C^{1+\al}[0,x^*]$.
Moreover, note that $v(\cdot,s_n)$ is bounded in $W^2_p[0,x^*]$ and converges to $v(\cdot,c)$ in $C[0,x^*]$, so it has a subsequence converges to $v(\cdot,c)$ in $C^{1+\al}[0,x^*]$.

Letting $n\to \infty$ in \eqref{uDsn}, we have $\TT v(\cdot,c)=0$ in $(x_*,x^*),$ i.e. $v_x(\cdot,c)=1$ in $(x_*,x^*).$
However, substituting $v_x(\cdot,c)=1$ into $-\LL v(x,c)-c \TT v(x,c)=0$ we conclude $v(\cdot,c)=\mu/r$ and thus $v_x(\cdot,c)=0$ in $(x_*,x^*),$ leading to a contradiction.

In a similar way, we can prove $\liminf\limits_{s\to c-}\sw(s)=\limsup\limits_{s\to c-}\sw(s)$ and $\liminf\limits_{s\to c-}\sw(s)\geq \limsup\limits_{s\to c+}\sw(s) $
so that
$$\lim\limits_{s\to c-}\sw(s)\geq \lim\limits_{s\to c+}\sw(s).$$

We come to prove
$$\lim\limits_{s\to c-}\sw(s)\leq \lim\limits_{s\to c+}\sw(s).$$
Suppose, on the contrary, there are $x_*$ and $x^*$ such that
$$\lim\limits_{s\to c+}\sw(s)<x_*<x^*<\lim\limits_{s\to c-}\sw(s).$$
Then there is $\ep>0$ such that $v(\cdot,c+\ep)=v(\cdot,c)$ in $(x_*,x^*),$ note that by \eqref{vx<=1} we have $\TT v(\cdot,c)\geq 0$ in $(x_*,x^*)$, so
$$0\leq -\LL v(\cdot,c+\ep)-(c+\ep)\TT v(\cdot,c+\ep)\leq -\LL v(\cdot,c)-c\TT v(\cdot,c)=0,$$
then we have $\TT v(\cdot,c)=0$ i.e. $v_x(\cdot,c)=1$ in $(x_*,x^*).$ Similar to above, it will lead to a contradiction.

It is left to prove
$\lim\limits_{s\to c}\sw(s)=\sw(c).$
For each $x>\sw(c),$ by the definition of $\sw(\cdot)$ and \lemref{lem:N}, there is $\os>c$ such that $v(y,s)=v(y,c)$, for all $(y,s)\in [x,+\infty)\times[c,\os],$ i.e. $[x,+\infty)\times[c,\os)\subseteq\SS$, then we have $\sw(s)\leq x$ for all $s\in [c,\os)$. So $\lim\limits_{s\to c+}\sw(s)\leq x$, which implies $\lim\limits_{s\to c+}\sw(s)\leq \sw(c)$.

On the other hand, for each $x>\lim\limits_{s\to c+}\sw(s),$ there is $\os>c$ such that $x>\sw(s)$ for all $s\in (c,\os],$ i.e. $[x,+\infty)\times(c,\os]\subseteq\SS$, so $v(x,s)=v(x,\os)$ for all $ s\in (c,\os]$. By the continuity of $v$, we have $v(x,c)=v(x,\os),$ which implies $\sw(c)\leq x,$ and thus, $\sw(c)\leq\lim\limits_{s\to c+}\sw(s).$
\end{proof}

\begin{lemma}\label{lemma:xcb}
The switching boundary $\sw(\cdot)$ is bounded in $[0,\cc)$.
\end{lemma}
\begin{proof}
Suppose, on the contrary, there is a convergence sequence $c_n\in [0,\cc)$ such that $\sw(c_n)\to+\infty$.
Because $\sw(\cdot)$ is continuous in $[0,\cc)$, we have $c_n\to \cc$.
Let
$$u^n(x)=\frac{v(x,c_n)-v(x,\cc)}{\cc-c_n}\geq 0.$$
Because $u^n(0)=0$ and \eqref{vc}, we conclude that $u^n$ is bounded and independent of $n$. Notice,
\begin{align*}
	-\LL u^n+c_n [b+(1-b) h_n] u^n_x=-\TT v(\cdot,\cc)=-\TT g(\cdot)~\hbox{in}\; [0,\sw(c_n)],
\end{align*}
where $$h_n=H(v_x(x,\cc),v_x(x,c_n))=H(g'(x),v_x(x,c_n))\in[0,1].$$
By the $L_p$ estimation and $\sw(c_n)\to+\infty$, we know $u^n$ is bounded in $W^2_p[0,L]$ for any $L>0$, so it has a subsequence converges to some $u\in W^2_p[0,L]$ weekly in $W^2_p[0,L]$ and uniformly in $C^{1+\al}[0,L]$.
Noting that
$h_n(x)\to h(x):=1_{\{g'(x)\leq 1\}}~\hbox{if}\;g'(x)\neq 1,$
so $u(x)$ satisfies
$-\LL u+\cc [b+(1-b) h] u_x=-\TT g(\cdot)~\hbox{a.e. in}\; \R^+.$
Since
$\lim\limits_{x\to+\infty}g'(x)=0,$ we have $\lim\limits_{x\to+\infty} \TT g(x)=1$. Thus,
there is $x^*>0$ large enough such that $\TT v(x,\cc)=\TT g(x)\geq 2/3$ for $x\geq x^*$. Let
$$\Phi(x)=A e^{-\ep (x-x^*)}-\frac{1}{2r},~~ x\geq x^*,$$
where $A=|u(x^*)|+\frac{1}{2r}$ and $\ep>0$ is sufficiently small such that
$-\ep^2+\big(\mu-\cc [b+(1-b) h]\big) \ep+r\geq 0.$
Then
\begin{align*}
	-\LL \Phi+\cc [b+(1-b) h] \Phi_x
	=&A\Big(-\ep^2+\big(\mu-\cc [b+(1-b) h]\big)\ep+r\Big) e^{-\ep (x-x^*)}-\frac{1}{2}\\[3mm]
	\geq &-\frac{1}{2}> -\TT v(x,\cc),~~ x\geq x^*.
\end{align*}
Since $\Phi(x^*)\geq u(x^*),$ applying the comparison principle yields $\Phi(x)\geq u(x)$ for $x\geq x^*$. It thus follows that $u(+\infty)\leq \Phi(+\infty)=-1/(2r).$ However, since $u^n\geq 0$, $u\geq 0$, leading to a contradiction.
\end{proof}

Similar to the proof of \lemref{lemma:xcc} we can prove $\liminf\limits_{c\to \cc-}\sw(c)=\limsup\limits_{c\to \cc-}\sw(c).$ So we can define
$\sw(\cc)=\lim\limits_{c\to \cc-}\sw(c),$
and consequently, $\sw(\cdot)\in C[0,\cc].$

We now prove \eqref{increasing}. If $\sw(\cdot)$ is strictly increasing, then for any $c_0\in [0,\cc)$, $$[0,\sw(c_0)]\times (c_0,\cc) \subseteq \{(x,c)\in \bQ \;|\; x< \sw(c)\}\subseteq\NS.$$ By the definition of $\NS$, $v(\sw(c_0),\cdot)$ is strictly decreasing in $(c_0,\cc)$, which implies $(\sw(c_0),c_0)\in \NS$. 
Conversely, if $\sw(\cdot)$ is not strictly increasing, since $\sw(\cdot)$ is continuous in $[0,\cc)$, there exists $c_0 \in [0,\cc)$ such that $\sw(c)\leq \sw(c_0)$ for all $c\in (c_0,c_0+\ep)$. Then $$(\sw(c_0),\infty)\times (c_0,c_0+\ep) \subseteq \{(x,c)\in \bQ \;|\; x> \sw(c)\}\subseteq \SS,$$ which implies $v(x,c_0)=v(x,c_0+\ep)$ for all $x>\sw(c_0)$. By the continuity of $v$, it leads to $v(\sw(c_0),c_0)=v(\sw(c_0),c_0+\ep)$, namely $(\sw(c_0),c_0) \in \SS$. So $\big\{(x,c)\in \bQ \;\big|\; x=\sw(c)\big\} \nsubseteq \NS$. Hence, we proved \eqref{increasing}.

We proceed to prove \eqref{decreasing}. If $\sw(\cdot)$ is non-increasing, then for any $c_0\in [0,\cc)$, $$(\sw(c_0),+\infty)\times [c_0,\cc) \subseteq \{(x,c)\in \bQ \;|\; x> \sw(c)\}\subseteq\SS.$$ Then $v(x,c_0)=v(x,\cc)$ for all $x>\sw(c_0)$. By the continuity of $v$, it leads to $v(\sw(c_0),c_0)=v(\sw(c_0),\cc)$, which implies $(\sw(c_0),c_0) \in \SS$. So $$\big\{(x,c)\in \bQ \;\big|\; x=\sw(c)\big\} \subseteq\SS.$$
Conversely, if $$\big\{(x,c)\in \bQ \;\big|\; x=\sw(c)\big\} \subseteq\SS,$$ then for any $c_0\in [0,\cc)$, there exists $s>c_0$ such that $v(\sw(c_0),s)=v(\sw(c_0),c_0)$. By the definition of $\sw(\cdot)$, we have $\sw(c)\leq \sw(c_0)$ for all $c\in (c_0,s)$. Since this $c_0\in [0,\cc)$ is arbitrary, $\sw(\cdot)$ is non-increasing. Therefore, we have proven \eqref{decreasing}. 

Now \propref{prop:characterizationx} is a consequence of Lemmas \ref{lem:N} $-$ \ref{lemma:xcb} and above. 

\section{Proof of \propref{prop:xi+}}\label{sec:proofofprop:xi+}
We now prove \propref{prop:xi+} in this section.

First, the monotonicity of $x\mapsto\swb(x,c)$ and $x\mapsto\minswb(x,c)$ follows from \lemref{lem:N} immediately.

We next prove \eqref{valueatboundary}.
\begin{itemize}
\item
If $\sw\big(\swb(x,c)\big)<x$, then by \propref{prop:characterizationx},
$\big(x,\swb(x,c)\big)\in\SS$. On the other hand,
it follows from definition that $\big(x,\swb(x,c)\big)\in\NS$, leading to a contradiction.
\item
If $\sw\big(\swb(x,c)\big)>x$, then by the continuity of $\sw(\cdot)$, we have $\sw(s)>x$ for any $s\in (\swb(x,c)-\ep, \swb(x,c))$, where $\ep$ is a sufficiently small positive constant.
Consequently, by \propref{prop:characterizationx}, $(x,s)\in\NS$. On the other hand, since $\minswb(x,c)<\swb(x,c)$, we have $v(x,s)=v\big(x,\swb(x,c)\big)$ for any $s\in(\minswb(x,c),\swb(x,c))$, so $(x,s)\in\SS$, leading to a contradiction.
\end{itemize}
Therefore, we have $\sw\big(\swb(x,c)\big)=x$, namely \eqref{valueatboundary} holds.

We now establish \eqref{lvatdecreasingpoint}.
If $\swb(x,c)=\cc$, then $v\big(\cdot,\swb(x,c)\big)=g(\cdot)$, so \eqref{lvatdecreasingpoint} follows from \eqref{Lg}. If $\swb(x,c)<\cc$, then
$\big(x,\swb(x,c)\big)\in \NS$, so the equation \eqref{lvatdecreasingpoint} follows from \eqref{-Lv=0}.

To prove \eqref{lvatincreasingpoint}, we suppose $s:=\minswb(x,c)>0$ and let $a_{n}\to s-$. If $s\leq \swb(x, a_{n})$ for some $n$, then $s\in (a_{n},\swb(x, a_{n})]\subseteq (\minswb(x, a_{n}),\swb(x, a_{n})]$, so that $v(x,s)=v(x, a_{n})$, which implies $\minswb(x,s)\leq a_{n}<s$.
But this contradicts that $s=\minswb(x,c)=\minswb(x,s)$. Hence, $a_{n}\leq \swb(x, a_{n})<s$, which implies $\swb(x, a_{n})\to s$.
Since $\big(x,\swb(x, a_{n})\big)\in\NS$, we obtain from \eqref{-Lv=0} that
\begin{align*}
-\LL v(y, \swb(x, a_{n}))-\swb(x, a_{n})\TT v(y, \swb(x, a_{n}))=0~\hbox{for all}~y\in (0,x).
\end{align*}
Thanks to the continuity of $v_{x}$, sending $n\to\infty$ in above gives the desired equation \eqref{lvatincreasingpoint}.

To show the right-continuity property \eqref{xi+}, since $\swb(x,c)$ is non-decreasing w.r.t. $x$, it suffices to prove $$\swb(x,c)\geq c^*:=\limsup\limits_{y\to x+} \swb(y,c).$$
Let $y_{n}$ be a sequence such that $y_n\to x+$ and $\swb(y_n,c)\to c^*$. By definition, we have $v(y_n,\swb(y_n,c))=v(y_n,c)$, so it follows from the continuity of $v$ that $v(x,c^*)=v(x,c)$ which by definition implies $ \swb(x,c)\geq c^*,$ completing the proof of \eqref{xi+}.
Similarly, one can prove \eqref{xi-}.

We now prove \eqref{v_c=0}.
Since $v(x,s)=v(x,c)$ for all $s\in (\minswb(x,c), \swb(x,c))$, we have $v_{c}(x,s)=0$.
Now suppose $\minswb(x,c)<s:=\swb(x,c)<\cc$.
Thanks to the monotonicity of $v$ and \lemref{lem:ucL}, we have
\begin{align*}
0\leq \frac{v(x,s)-v(x,s+\Ds)}{\Ds}\leq \frac{v(x,s-\Ds)-v(x,s)}{\Ds}+2B\Ds
\end{align*}
for $0<\Ds<\min\{s-\minswb(x,c),\cc-s\}.$ Since $v(x,\xi)=v(x,c)$ for all $\xi\in [\minswb(x,c), \swb(x,c)]$, we see the right hand side in the above estimate is $2B\Ds$.
By sending $\Ds\to 0+$, we get
$v_c(x,s)=0$, completing the proof of \eqref{v_c=0}.
The proof of \propref{prop:xi+} is complete.

\section{Proof of \propref{prop:characterizationy}}\label{profprop:characterizationy} 
We next study the properties of the converting boundary, which is defined as
\begin{align*}
\sd(c)=\inf\Big\{x\in\R^{+}\;\Big|\; v_x(x,c)\leq 1\Big\},~~ c\in[0,\cc].
\end{align*}

From the above definition we know $v_x(x,c)> 1$ when $x<\sd(c)$. By the continuity of $v_{x}$ and \eqref{vxx}, we have $v_x(x,c)\leq 1$ when $x\geq\sd(c)$, so $\sd(c)$ is the free boundary separating the two regions $\{v_x > 1\}$ and $\{v_x \leq 1\}$, i.e.,
$$\{v_x > 1\}=\{x<\sd(c)\},~~\{v_x \leq 1\}=\{x\geq \sd(c)\}.$$

It is not hard to see \propref{prop:characterizationy} will be a consequence of \lemref{lem:yc} and \lemref{lem:y<x}.

\begin{lemma}\label{lem:yc}
The converting boundary $\sd(\cdot)$ is continuous and positive in $[0,\cc]$ with $\sd(\cc)=y_{0}$.
\end{lemma}
\begin{proof} 
For any $c\in [0,\cc]$, let
$y_*=\liminf\limits_{s\to c}\sd(s) $ and $y^* =\limsup\limits_{s\to c}\sd(s)$.
Then by the continuity of $v_x$ we have $v_x(y_*,c)=v_x(y^*,c)=1$.
Note $v_x(y,s)> 1$ when $y<\sd(s)$, so by \eqref{vxx} we have
$y\mapsto v_x(y,s)$ is non-increasing when $y<\sd(s)$. Then the continuity of $v_x$ implies $y\mapsto v_x(y,c)$ is non-increasing in $[0,y^*]$, so $v_x(\cdot,c)= 1$ in $(y_*,y^*).$
However, substituting $v_x(\cdot,c)=1$ into $-\LL v-c \TT v=0$ we conclude $v(\cdot,c)=\mu/r$ and thus $v_x(\cdot,c)=0$ in $(y_*,y^*),$ leading to a contradiction.
Now we established the continuity of $\sd(\cdot)$.

Since $1=v_x(\sd(\cc),\cc)=g'(\sd(\cc))$, we obtain $\sd(\cc)=y_{0}>0$. 
Moreover, since $v_{x}$ is continuous, we have $v_x(\sd(c),c)=1$, so \eqref{vx0>1} implies $\sd(c)>0$ for all $c\in [0,\cc)$.
The proof is complete.
\end{proof}

The definition of $\sd(\cdot)$ and \eqref{vxxc} imply that $\sd(c)\leq \sw(c)$ for any $c\in[0,\cc]$.
The next result \lemref{lem:y<x}, further shows that these two curves do not intersect.

\begin{lemma}\label{lem:y<x}
	We have for all $c\in (0,\cc]$\footnote{In the case of $ c = 0 $, one can still establish the proof using the same approach, by extending the lower bound of $ c $ to a negative value.} 
	\begin{align}\label{y<x1}
		\sd(c)<\sw(c).
	\end{align}
\end{lemma}
\begin{proof}
	We proceed by contradiction. Suppose \eqref{y<x1} fails, i.e., there exists a $c\in (0,\cc]$ such that $\sd(c)=\sw(c)$. By the definition of $\sd(\cdot)$, this implies $v_x(x,c)> 1$ and $\TT v(x,c)<0$ for all $ x\in [0, \sw(c))$. 
	We analyze three cases:

	{\bf Case 1: $\minswb(\sw(c),c)<c$.}
	
	Denote $s=\minswb(\sw(c),c)$ and 
	$$u(x)=\frac{v(x,s)-v(x,c)}{c-s}.$$
	Note that
	$$-\LL v(x,s)-s \TT v(x,s)\geq 0=-\LL v(x,c)-c \TT v(x,c),~~ x\in [0,\sw(c)].$$
	After taking the difference, we obtain
	\begin{align*}
		-\LL u+s [b+(1-b) h] u_x\geq-\TT v(\cdot,c) > 0~\hbox{in}\;[0,\sw(c)],
	\end{align*}
	where $$h=H(v_x(x,s),v_x(x,c))\in[0,1].$$
	Since $u\geq 0$ and $u(\sw(c))=0$, the strong maximum principle and the Hopf lemma imply $u_x (\sw(c))<0$. However, since $u(x)=0$ for all $ x\geq \sw(c)$, so continuity of 
	$u_x $ forces $u_x (\sw(c))=0$, leading to a contradiction.
	
	{\bf Case 2: $c<\swb(\sw(c),c)$.}
	
	Denote $s=\swb(\sw(c),c)$ and
	$$u(x)=\frac{v(x,s)-v(x,c)}{c-s}.$$
	Since $\sw(s)=\sw(c)$, we have
	$$-\LL v(x,s)-s \TT v(x,s)=-\LL v(x,c)-c \TT v(x,c)=0,~~ x\in [0,\sw(c)].$$
	An analogous reasoning to Case 1 leads to a contradiction.

	{\bf Case 3: $\minswb(\sw(c),c)\geq c $ and $c\geq \swb(\sw(c),c)$.} 
	
	In this case $\minswb(\sw(c),c)=c=\swb(\sw(c),c)$.
	Choose $x_n\to \sw(c)+$ and let $\os_n=\swb(x_n,c)$ and $\us_n=\minswb(x_n,c)$.
	If $$\os:=\limsup\limits_{n\to\infty}\os_n>c~~\text{ or }~~\us:=\liminf\limits_{n\to\infty}\us_n<c,$$ then by the continuity of $v$, we have $v(\sw(c),\os)=v(\sw(c),c)$ or $v(\sw(c),\us)=v(\sw(c),c)$, which goes back to Case 1 or 2. 
	So we may assume $\os_n,\;\us_n\to c.$
	Let $$u^n(x)=\frac{v(x,\us_n)-v(x,\os_n)}{\os_n-\us_n}.$$ 
	Since $\us_n>0$ when $n$ is sufficiently large, by \eqref{lvatdecreasingpoint} and \eqref{lvatincreasingpoint} we have
	$$-\LL v(x,\os_n)-\os_n \TT v(x,\os_n)=-\LL v(x,\us_n)-\us_n \TT v(x,\us_n)=0,~~ x\in [0,x_n].$$
	so
	\begin{align*}
		-\LL u^n+\us_n [b+(1-b) h_n] u^n_x=-\TT v(\cdot,\os_n)~\hbox{in}\;[0,x_n],
	\end{align*}
	where $$h_n=H(v_x(x,\os_n),v_x(x,\us_n))\in[0,1].$$
	Thanks to \eqref{vL}, there exist subsequences of $v(\cdot, \os_n)$ and $v(\cdot, \us_n)$ that converge to $v(\cdot, c)$ uniformly in $C^{1+\alpha}[0, \sw(c)]$. 
	Since $v_x(\cdot, c) > 1$ in $[0, \sw(c))$, it follows that the subsequences of $h_n$ converge to 0 in $[0, \sw(c))$. 
	On the other hand, because $\TT v(\cdot, \os_n)$ and $u^n$ are bounded and independently of $n$, the $L_p$ estimate implies that $u^n$ is bounded in $W_p^2[0, \sw(c)]$. Consequently, there exists a subsequence of $u^n$ that converges to some $u$ weakly in $W_p^2[0, \sw(c)]$ and uniformly in $C^{1+\alpha}[0, \sw(c)]$. This limit $u$ satisfies 
	\begin{align*}
		-\LL u + c b u_x = -\TT v(\cdot, c) \quad \text{in } [0, \sw(c)],
	\end{align*} 
	and $u(\sw(c)) = u_x(\sw(c)) = 0$. 
	Given that $-\TT v(\cdot, c) \geq 0$ in $[0, \sw(c)]$, the strong maximum principle and the Hopf lemma imply $u_x(\sw(c)) < 0$, which is a contradiction.
\end{proof}

\section{Proof of \thmref{thm:averi}}\label{profthm:averi}
Suppose $V(x,c)$ is the value function defined in \eqref{value} and $v(x,c)$ is the solution to \eqref{v_pb}. We come to prove $v(x,c)=V(x,c)$.

For any $(x,c)\in \bQ$, any admissible dividend payout strategy $\{\DD_t\}_{t\geq 0}\in \Pi_{[c,\cc]}$
and any constant $T>0$,
through a change of measure argument,
one can show that $X_{t}$ is a Brownian motion under some equivalent probability measure, so $\int_{0}^{T}\P(X_t\in K)\dt=0$ for any subset $K$ of $\R$ with zero Lebesgue measure.

Applying It\^o's formula (see \cite{protter2005stochastic} Theorem 33 on page 81) to $e^{-rt} v(X_t,\RM_t)$ and then taking expectation
lead to
\begin{align*}
v(x,c)
&=\E\Bigg[e^{-r(\tau\wedge T) }v(X_{\tau\wedge T},\RM_{\tau\wedge T})+\int_0^{\tau\wedge T}e^{-rt} \Big(-\LL v(X_t,\RM_t)+v_x(X_t,\RM_t)\DD_t\Big) \dt
\\
&\qquad\quad- \sum_{0\leq t\leq \tau\wedge T} e^{-rt} \Big(v(X_t,\RM_t)-v(X_t,\RM_{t-})\Big)- \int_0^{\tau\wedge T}e^{-rt} v_c (X_t,\RM_t)\d \RM^{c}_t \\
&\qquad\quad
-\int_0^{\tau\wedge T}e^{-rt}\si v_x (X_t,\RM_t) \dw_t \;\Bigg|\; X_0=x,~\RM_0=c\Bigg],
\end{align*}
where $\RM_t^c$ is the continuous part of $\RM_t$.
From the estimate \eqref{vx} we know $v_x(X_t,\RM_t)$ is bounded,
so the mean of the stochastic integral is zero. 
Applying \eqref{v_pb}, we get
\begin{align*}
-\LL v(X_t,\RM_t)+v_x(X_t,\RM_t)\DD_t-\DD_t
\geq& -\LL v(X_t,\RM_t)-\sup\limits_{b\RM_t\leq \DD\leq \RM_t}\DD (1-v_x(X_t,\RM_t))\\[3mm]
=& -\LL v(X_t,\RM_t)-\RM_t \TT v(X_t,\RM_t)\geq 0.
\end{align*}
Together with
$$ -v_c \geq 0, \quad M_{t}\geq M_{t-}, \quad\d \RM^{c}_t\geq 0,$$ we have
\begin{align}\label{vE}
v(x,c)
\geq
\E\Bigg[e^{-r(\tau\wedge T)}v(X_{\tau\wedge T},\RM_{\tau\wedge T})
+\int_0^{\tau\wedge T}e^{-rt} \DD_t \dt\;\Bigg|\; X_0=x,~\RM_0=c\Bigg].
\end{align}
Because $v$, $\DD_t\geq 0$, dropping the first term in the expectation and sending $T\to+\infty$ in above, the monotone convergence theorem gives
\begin{align*}
v(x,c)
\geq \E\Bigg[\int_0^{\tau}e^{-rt} \DD_t \dt\;\Bigg|\; X_0=x,~\RM_0=c\Bigg].
\end{align*}
Since $\{\DD_t\}_{t\geq 0}\in \Pi_{[c,\cc]}$ is arbitrary selected, we obtain $v(x,c)\geq V(x,c)$.

We now prove the opposite inequality $v(x,c)\leq V(x,c),$ and show that $\{\DD^*_{t}\}_{t\geq 0}$ is an optimal strategy.
Thanks to the continuity of $X^{*}$ and $\sd(\cdot)$, the right-continuity of $\RM^{*}$ given by \eqref{xi+}, it is not hard to verify $\{\DD^*_{t}\}_{t\geq 0}\in \Pi_{[c,\cc]}$ is an admissible strategy.

We now assume $t\in[0, \zeta_1)$, where 
$$\zeta_1=\sup\big\{t\geq 0 \;\big|\; \RM^*_t\leq c\big\}.$$
Since $\RM^*_0=c$, we have $\RM^*_t=c$. It follows
$\swb\Big(\max\limits_{s\in[0,t]}X^*_s,c\Big)=\RM^*_t=c,$
and
$ \Big(\max\limits_{s\in[0,t]}X^*_s,c\Big)\in \NS.$
It follows from \eqref{-Lv=0} that
\begin{align*}
-\LL v(y,c)-c \TT v(y,c)=0~~\hbox{if }0<y< \max\limits_{s\in[0,t]}X^*_s,
\end{align*}
combing
\begin{align}\label{P0}
\int_0^T\P\big(X^*_t=\max\limits_{s\in[0,t]}X^*_s\big) \;\dt=0,
\end{align}
yields
\begin{align*}
-\LL v(X^*_t,\RM^*_t)-\RM^*_t \TT v(X^*_t,\RM^*_t)
&=-\LL v(X^*_t, c)-c\TT v(X^*_t,c)=0 ~~\hbox{a.e. in}~~ \Omega\times [0,\zeta_1\wedge \tau\wedge T].
\end{align*}

By the construction of $\DD^*_t$, it gives
\begin{align*}
&~-\LL v(X^*_t,\RM^*_t)+\DD^*_t v_x(X^*_t,\RM^*_t) -\DD^*_t\\
=&~-\LL v(X^*_t,\RM^*_t)-\RM^*_t\TT v(X^*_t,\RM^*_t)=0 ~~\hbox{a.e. in}~~ \Omega\times [0,\zeta_1\wedge \tau\wedge T]. 
\end{align*}
Now applying It\^o's formula to $ e^{-rt}v(X^{*}_{t},c)$ on $[0, \zeta_1\wedge \tau\wedge T)$ gives
\begin{align}
v(x,c)
&=\E\Bigg [e^{-r(\zeta_1\wedge \tau\wedge T) }v(X^*_{\zeta_1\wedge \tau\wedge T},c)
+\int_0^{\zeta_1\wedge \tau\wedge T}e^{-rt} \Big(-\LL v(X^*_t,c)+\DD^*_t v_x(X^*_t,c)\Big) \dt \nn\\
&\qquad\quad- \int_0^{\zeta_1\wedge \tau\wedge T}e^{-rt}\si v_x (X^*_t,c) \dw_t\;\Bigg|\; X_0=x,~\RM_0=c\Bigg]\nn\\
&=\E\Bigg[e^{-r(\zeta_1\wedge \tau\wedge T) }v\big(X^*_{\zeta_1\wedge \tau\wedge T},\RM^*_{ (\zeta_1 \wedge \tau\wedge T) -}\big) \nn\\
&\qquad\quad+\int_0^{\zeta_1\wedge \tau\wedge T}e^{-rt} \Big(-\LL v(X^*_t,\RM^*_t)+\DD^*_t v_x(X^*_t,\RM^*_t)\Big) \dt\;\Bigg|\; X_0=x,~\RM_0=c\Bigg] \nn\smallskip\\
&=\E\Bigg[e^{-r(\zeta_1\wedge \tau\wedge T) }v\big(X^*_{\zeta_1\wedge \tau\wedge T},\RM^*_{ (\zeta_1 \wedge \tau\wedge T) -}\big)+\int_0^{\zeta_1\wedge \tau\wedge T} e^{-rt} \DD^*_t \dt\;\Bigg|\; X_0=x,~\RM_0=c\Bigg].\label{beforezeta}
\end{align}

Now let us consider the case $t\in[\zeta_1,\zeta_{2})$, where 
\[\zeta_2=\inf\big\{t\geq \zeta_1 \;\big|\; \RM^*_t=\cc\big\}.\]
Since $c<\RM^*_t<\cc$, by the construction of $\DD^*_t$, \eqref{P0} and \eqref{lvatdecreasingpoint},
\begin{align*} 
&\quad-\LL v(X^*_t,\RM^*_t)+\DD^*_t v_x(X^*_t,\RM^*_t) -\DD^*_t\nn\smallskip\\
&=-\LL v(X^*_t,\RM^*_t)-\RM^*_t\TT v(X^*_t,\RM^*_t)\nn\\[3pt]
&=-\LL v\Big(X^*_t, \swb\Big(\max\limits_{s\in[0,t]}X^*_s,c\Big)\Big)- \swb\Big(\max\limits_{s\in[0,t]}X^*_s,c\Big)\TT v\Big(X^*_t, \swb\Big(\max\limits_{s\in[0,t]}X^*_s,c\Big)\Big)\nn\\
&=0 ~~\hbox{a.e. in}~~ \Omega\times [\zeta_1\wedge \tau\wedge T,\zeta_2\wedge \tau\wedge T].
\end{align*}

By virtue of \eqref{valueatboundary},
$$\sw(\RM^*_t)=\sw\Big(\swb\Big(\max\limits_{s\in[0,t]}X^*_s,c\Big)\Big)=\max\limits_{s\in[0,t]}X^*_s,$$
implying
$X^*_t\leq \sw(\RM^*_t).$
If $X^*_t=\sw(\RM^*_t)$, then thanks to \eqref{v_c=0}, we have $v_c(X^*_t,\RM^*_t)=0$, so
$$-v_c(X^*_t,\RM^*_t)\d (\RM^*)^{c}_t=0.$$
If $X^*_t< \sw(\RM^*_t)$, then $\d (\RM^*)^{c}_t=0$, so the above equation also holds.
Also, if $\RM^*_t>\RM^*_{t-}$, then $X^*_t=\max\limits_{s\in[0,t]}X^*_s$, and thus
$\RM^*_t=\swb\Big(\max\limits_{s\in[0,t]}X^*_s,c\Big)=\swb(X^*_t,c).$
Because $c\leq \RM^*_{t-}\leq \RM^*_t=\swb(X^*_t,c)$ and $v\big(x,\swb(x,c)\big)=v(x,s)$ for any $s\in[c,\swb(x,c)]$, we get
$v(X^*_t,\RM^*_t)=v(X^*_t,\swb(X^*_t,c))=v(X^*_t,\RM^*_{t-}).$ 
Specially, we have
\begin{align*}
v(X^*_{\zeta_1\wedge \tau\wedge T},\RM^*_{\zeta_1\wedge \tau\wedge T})=
v(X^*_{\zeta_1\wedge \tau\wedge T},\RM^*_{ (\zeta_1 \wedge \tau\wedge T) -}).
\end{align*}

Now applying It\^o's formula to $ e^{-rt}v(X^*_t,\RM^*_t)$ on $( \zeta_1\wedge \tau\wedge T,\zeta_2\wedge \tau\wedge T]$ gives
\begin{align}
&\quad\; \E\Big[e^{-r(\zeta_1\wedge \tau\wedge T) }v\big(X^*_{\zeta_1\wedge \tau\wedge T},\RM^*_{ (\zeta_1 \wedge \tau\wedge T) -}\big)\;\Big|\; X_0=x,~\RM_0=c\Big]\nn \smallskip\\
&=\E\Bigg[ e^{-r(\zeta_2\wedge \tau\wedge T) }v\big(X^*_{\zeta_2\wedge \tau\wedge T},\RM^*_{\zeta_2\wedge \tau\wedge T}\big)
+\int_{\zeta_1\wedge \tau\wedge T}^{\zeta_2\wedge \tau\wedge T}e^{-rt} \Big(-\LL v( X^*_t,\RM^*_t)+v_x( X^*_t,\RM^*_t)\DD^*_t\Big) \dt\nn\\
&\qquad\quad
-\sum_{\zeta_1\wedge \tau\wedge T\leq t\leq \zeta_2\wedge \tau\wedge T} e^{-rt} \Big(v(X^*_t,\RM^*_t)-v(X^*_t,\RM^*_{t-})\Big) - \int_{\zeta_1\wedge \tau\wedge T}^{\zeta_2\wedge \tau\wedge T}e^{-rt} v_c ( X^*_t,\RM^*_t)\d (\RM^*)^{c}_t
\nn \\
&\qquad\quad
- \int_{\zeta_1\wedge \tau\wedge T}^{\zeta_2\wedge \tau\wedge T}e^{-rt}\si v_x ( X^*_t,\RM^*_t) \dw_t\;\Bigg|\; X_0=x,~\RM_0=c\Bigg]\nn\smallskip\\
&=\E\Bigg[e^{-r(\zeta_2\wedge \tau\wedge T) }v\big(X^*_{\zeta_2\wedge \tau\wedge T},\RM^*_{\zeta_2\wedge \tau\wedge T}\big)
+\int_{\zeta_1\wedge \tau\wedge T}^{\zeta_2\wedge \tau\wedge T}e^{-rt} \DD^*_t \dt\;\Bigg|\; X_0=x,~\RM_0=c\Bigg].
\label{middle}
\end{align}

Finally, for $t\in[\zeta_2,\tau)$, we have $\RM^*_t=\cc$. We can deduce similarly to the case
$t\in[0,\zeta_1)$ to get
\begin{align*}
&\quad\;\E\Big[e^{-r(\zeta_2\wedge \tau\wedge T) }v\big(X^*_{\zeta_2\wedge \tau\wedge T},\RM^*_{\zeta_2\wedge \tau\wedge T}\big)\;\Big|\; X_0=x,~\RM_0=c\Big]\smallskip\\
&=\E\Bigg[e^{-r(\tau\wedge T) }v\big(X^*_{\tau\wedge T},\RM^*_{\tau\wedge T}\big)
+\int_{\zeta_2\wedge \tau\wedge T}^{\tau\wedge T}e^{-rt} \DD^*_t \dt\;\Bigg|\; X_0=x,~\RM_0=c \Bigg].
\end{align*}

Combining \eqref{beforezeta}, \eqref{middle} and above, we obtain
\begin{align}\label{wholeeq}
v(x,c)&=\E\Bigg[e^{-r(\tau\wedge T) }v\big(X^*_{\tau\wedge T},\RM^*_{\tau\wedge T}\big)+\int_0^{ \tau\wedge T} e^{-rt} \DD^*_t \dt\;\Bigg|\; X_0=x,~\RM_0=c\Bigg].
\end{align}
Because $v$ is bounded and $r>0$, it yields
\[\lim_{T\to+\infty}e^{-r(\tau\wedge T) }v\big(X^*_{\tau\wedge T},\RM^*_{\tau\wedge T}\big) 1_{\tau=+\infty}=\lim_{T\to+\infty}e^{-r T}v(X^*_{ T},\RM^*_{ T}) 1_{\tau=+\infty}=0;\]
whereas, by virtue of $X^*_{\tau}1_{\tau<+\infty}=0$ and $v(0,c)=0$, it follows
\[\lim_{T\to+\infty}e^{-r(\tau\wedge T) }v\big(X^*_{\tau\wedge T},\RM^*_{\tau\wedge T}\big) 1_{\tau<+\infty}=e^{-r\tau}v(X^*_{\tau},\RM^*_{\tau}) 1_{\tau<+\infty}
=e^{-r\tau}v(0,\RM^*_{\tau}) 1_{\tau<+\infty}=0.\]
Therefore, by sending $T\to+\infty$ in \eqref{wholeeq} and using the dominated and monotone convergence theorems, we obtain
$$
v(x,c)=\E\Bigg[\int_0^{\tau}e^{-rt} \DD^*_t \dt\;\Bigg|\; X_0=x,~\RM_0=c\Bigg].
$$
In view of the definition of $V(x,c)$, we conclude that $v (x,c)\leq V(x,c)$, and consequently, that $\{\DD^*_{t}\}_{t\geq 0}$ is an optimal dividend payout strategy.
This completes the proof of \thmref{thm:averi}.

\section{The no-ceiling model}\label{proof:no ceiling}

In this section, we study the no-ceiling model and simplify the notation $v^{\infty}$ as $v$. 

Since the constant $K$ in \eqref{vc} and \eqref{vx} are independent of $\cc$ (see the proof of \lemref{lem:vix_ub} and \lemref{lem:uix_b}), by taking limits in \eqref{vc}, \eqref{vx} and \eqref{vxx}, we obtain:
\begin{lemma}\label{inftylem:vv_b}
	The following estimates hold
	\begin{gather}
		\label{inftyvvc} 0\leq -v _{c}\leq K ~\mbox{a.e.},\\
		\label{inftyvvx} 0\leq v _x\leq K~\mbox{a.e.},
	\end{gather}
	for some constant $K>0$.
	In addition, 
	\begin{align}\label{inftyvvxx}
		v _x(y,c)\leq \max\{v _x(x,c), 1\},~~ \forall\; y\geq x ~\mbox{a.e.},
	\end{align}
\end{lemma}
Actually, the ``almost everywhere" condition in the estimates \eqref{inftyvvx} and \eqref{inftyvvxx} can be eliminated, because $v_x(\cdot,c)$, that will be proven in \lemref{inftylem:vxc}, is continuous in $\R^+$. 

The following \lemref{inftylem:vv_eq}-\lemref{inftylem:vxc} will confirm that $v$ defined in \eqref{inftyvv_def} is the unique weak solution of \eqref{inftyvv_pb}.

\begin{lemma}\label{inftylem:vv_eq}
	We have $v(0,c)=0$ for all $c>0$ and
	\begin{align}\label{inftyvv_eq0}
		-\LL v(\cdot, c) - c\TT v(\cdot, c) \geq 0 \quad \text{weakly in } \R^+.
	\end{align}
	Moreover, we have for any $(x, c) \in \NS^{\infty}$, $v(\cdot, c) \in C^2(0, x)$ and there holds
	\begin{align}\label{inftyvv_eq}
		-\LL v(\cdot, c) - c\TT v(\cdot, c) = 0 \quad \text{in } (0, x).
	\end{align}
\end{lemma}
\begin{proof}
	By the estimate \eqref{inftyvx} and $v^{\cc}(0,c)=0$ we have $0\leq v^{\cc}(x,c)\leq K x$. Since the constant $K$ is independent of $\cc$, by letting $\cc\to +\infty$ we obtain $0\leq v(x,c)\leq K x$, which implies $v(0,c)=0$.
	
	Let $\cc\to+\infty$ in $-\LL v^{\cc}(\cdot, c) - c\TT v^{\cc}(\cdot, c) \geq 0$ we get \eqref{inftyvv_eq0}. 
	
	Now we prove \eqref{inftyvv_eq}. Suppose $(x, c) \in \NS^{\infty}$, by definition, we have $v(x, c + 1/n) < v(x, c)$ for any $n \in \N^+$. Since $v^{\cc}$ tends to $v$ as $\cc\to+\infty$, there exists $\cc_n > n$ such that $v^{\cc_n}(x, c + 1/n) < v^{\cc_n}(x, c)$. Then there exists $c_n \in [c, c + 1/n)$ such that $v^{\cc_n}(x, s) < v^{\cc_n}(x, c_n)$ for all $s > c_n$, and thus $$-\LL v^{\cc_n}(\cdot, c_n) - c_{n}\TT v^{\cc_n}(\cdot, c_n) = 0$$ in $(0, x)$. 
	Let $n \to +\infty$ we get \eqref{inftyvv_eq}, and this implies $v(\cdot, c) \in C^2(0, x)$.
\end{proof}

\begin{lemma}\label{inftylem:vv_eq1}
	If $\swb^{\infty}(x,c)<+\infty$, then for all $s\in [0,\swb^{\infty}(x,c)]$, $v(\cdot,s)\in W^{2}_p[0,x]$ (for any $p>1$ and $\al=1-1/p$).
\end{lemma}
\begin{proof}
	Let $c^* = \swb^{\infty}(x,c)$. Then we have $v(x, c^* + 1) < v(x, c^*)$. Since $v^{\cc}$ tends to $v$ as $\cc \to +\infty$, we obtain $v^{\cc}(x, c^* + 1) < v^{\cc}(x, c^*)$ for sufficiently large $\cc$.
	Let $v^{\cc,n}_i(x),\;i = 1, 2, \cdots, n$ denote the sequence of discrete approximating functions of $v^{\cc}$ constructed in \secref{sec:approximation}, which solve the variational inequality system \eqref{vi_pb}. For any sufficiently large $n$, there exists $i_0 \in \N^+$ such that $c_{i_0} := \cc - i_0 \cc/n \in (c^*, c^* + 1)$ and $v^{\cc,n}_{i_0}(x) > v^{\cc,n}_{i_0-1}(x)$. This implies $v^{\cc,n}_{i_0}(y) > v^{\cc,n}_{i_0-1}(y)$ for all $y \leq x$.
	Then \eqref{Lvi} becomes
	\begin{align*}
		-\LL v^{\cc,n}_i = \Big( \sum\limits_{j=i_0}^i c_j 1_{\{ v^{\cc,n}_{j-1} < v^{\cc,n}_i \leq v^{\cc,n}_j \}} \Big) \TT v^{\cc,n}_i ~~\text{a.e. in } (0,x),~~\forall~i \geq i_0.
	\end{align*}
	Since $0 \leq c_i \leq c_j \leq c_{i_0} \leq c^* + 1$ for $i \geq j \geq i_0$, together with \eqref{inftyvvx}, the RHS above is bounded and independent of $\cc$ and $n$. We can thus apply the $L_p$ estimation to conclude that for each $i \geq i_0$ (i.e., $c_i \leq c_{i_0}$), $|v^{\cc,n}_i|_{W^{2}_{p}[0,x]}$ is bounded and independent of $\cc$ and $n$. Letting $n \to +\infty$, we obtain that for all $s \in [0, c^*]$, $|v^{\cc}(\cdot,s)|_{W^{2}_{p}[0,x]}$ is bounded and independent of $\cc$. Finally, letting $\cc \to +\infty$, we conclude that $v(\cdot,s) \in W^{2}_{p}[0,x]$ for all $s \in [0, c^*]$.
\end{proof}

\begin{lemma}\label{inftylem:vxc0}
	If $\swb^{\infty}(x,c)<+\infty$, then $v_x(y,c)$ is continuous for every $y\leq x$.
\end{lemma}
\begin{proof}
	If $\swb^{\infty}(x,c)<+\infty$, by the continuity of $v$, we have $\swb^{\infty}(x+\ep,c)<+\infty$ for some $\ep>0$. \lemref{inftylem:vv_eq1} yields $v(\cdot,c)\in W^{2}_{p}[0,x+\ep]$, and the Sobolev embedding theorem gives $v(\cdot,c)\in C^1[0,x+\ep)$. 
\end{proof}

\begin{lemma}\label{inftylem:vr}
	If $\swb^{\infty}(x,c)=+\infty$, then $v(y,c)=y$ for all $y\geq x$.
\end{lemma}
\begin{proof}
	Suppose $\swb^{\infty}(x,c)=+\infty$. From \eqref{inftyvv_limc}, we have $v(x,c)=\lim\limits_{c\to+\infty}v(x,c)=x$. The first inequality in \eqref{inftyvv}, i.e. $v(y,c)\geq y$ for $y\geq 0$, implies that for $y< x$, $(v(y,c)-v(x,c))/(y-x)\leq 1$. By \eqref{inftyvvxx} we get $v_x(y)\leq 1$ a.e. for $y\geq x$. Then we conclude $v(y,c)\leq y$ for $y\geq x$. Together with \eqref{inftyvv} we have $v(y,c)=y$ for all $y\geq x$.
\end{proof}

\begin{lemma}\label{inftylem:vxc}
	For any $c\in [0,+\infty)$, $v_x(\cdot,c)$ is continuous in $[0,+\infty)$.
\end{lemma}
\begin{proof}
	Let $x_0=\inf\{x>0\;|\;\swb^{\infty}(x,c)=+\infty\}$. Then $\swb^{\infty}(x,c)<+\infty$ for all $x<x_0$, thus by \lemref{inftylem:vxc0} we have $v_x(\cdot,c)$ is continuous in $[0,x_0)$. In addition, by \lemref{inftylem:vr} we have $v(x,c)=x$ for all $x\geq x_0$, so $v_x(\cdot,c)$ is continuous in $(x_0,+\infty)$.
	
	To complete the proof, it suffices to show that $v_x(\cdot,c)$ is continuous at $x_0$. 
	First, apply \eqref{inftyvv_eq0},\eqref{inftyvv} and \eqref{inftyvvx}, we conclude that there exists a constant $L>0$ such that $v_{xx}(x,c)\leq L$ weakly in $[0,x_0]$, this implies $$v_x(x,c)\geq v_x(x_0+,c)-L(x_0-x)=1-L(x-x_0)$$ for all $x< x_0$, which gives $\liminf\limits_{x\to x_0-} v_x(x,c)\geq 1$.
	
	Next, we prove $\limsup\limits_{x\to x_0-} v_x(x,c)\leq 1$. Suppose that $\limsup\limits_{x\to x_0-} v_x(x,c)> 1$, then there exists a sequence $x_n\to x_0-$ and $\delta>0$ such that $v_x(x_n,c)\geq 1+\delta$. Combining with \eqref{inftyvvxx} implies $v_x(x,c)\geq 1+\delta$ for all $x\in [x_1,x_0)$. Integrating, we conclude $$v(x,c)\leq v(x_0,c)-(1+\delta)(x_0-x)=x_0-(1+\delta)(x_0-x)< x$$ for $x\in [x_1,x_0)$, which contradicts \eqref{inftyvv}. Thus, $$\liminf\limits_{x\to x_0-} v_x(x,c)=\limsup\limits_{x\to x_0-} v_x(x,c)=1,$$ so $$\lim\limits_{x\to x_0-} v_x(x,c)=1=\lim\limits_{x\to x_0+} v_x(x,c).$$
	And since $v(\cdot,c)$ is continuous at $x_0$, by applying the mean value theorem we have $v_x(x_0,c)=1$.
	Hence, $v_x(\cdot,c)$ is continuous at $x_0$ and therefore continuous in $[0,+\infty)$.
\end{proof}

By combining \lemref{inftylem:vv_eq} and \lemref{inftylem:vxc}, we establish that the function 
$v$ defined in \eqref{inftyvv_def} is a solution to the variational inequality given in \eqref{inftyvv_pb}.
Furthermore, following the line of reasoning presented in \appref{sec:solvability}, we can prove that the solution to \eqref{inftyvv_pb} is unique.

\subsection{The properties of the switching boundary} \label{sec:switching}

By the variational inequality in \eqref{inftyvv_pb} and the continuity of $v_x(\cdot,c)$, following the discussion of \lemref{lem:N}, we can prove that $v_x(x ,c)\leq 1$, $\forall\; (x,c)\in S$ still holds for the case $\cc=+\infty$. Furthermore, by virtue of the uniqueness of the solution in $\A^{\infty}_0$, we can further show that $v(y,s)=v(y, c)$ for all $(y,s)\in [x,+\infty)\times[\minswb^{\infty}(x,c),\swb^{\infty}(x,c)]$. As a result, \lemref{lem:N} still holds for the case $\cc=+\infty$. i.e. we have
\begin{lemma}\label{inftylem:NN}
	If $(x,c)\in\SS^{\infty}$, then
	\begin{align}\label{inftyvvx<=1}
		v_x(x ,c)\leq 1
	\end{align}
	and $v(y,s)=v(y, c)$ for all $(y,s)\in [x,+\infty)\times[\minswb^{\infty}(x,c),\swb^{\infty}(x,c)].$ As a consequence, we have
	\begin{align}\label{inftyvvxxc}
		v_x(\sw^{\infty}(c),c)\leq 1,~c\in [0,+\infty),
	\end{align}
	and
	\begin{gather*}
		\Big\{(x,c)\in \bQ^{\infty} \;\Big|\; x> \sw^{\infty}(c)\Big\}\subseteq\SS^{\infty}
		\subseteq\Big\{(x,c)\in \bQ^{\infty} \;\Big|\; x\geq \sw^{\infty}(c)\Big\},\\
		\Big\{(x,c)\in \bQ^{\infty} \;\Big|\; x< \sw^{\infty}(c)\Big\}\subseteq\NS^{\infty}
		\subseteq\Big\{(x,c)\in \bQ^{\infty} \;\Big|\; x\leq \sw^{\infty}(c)\Big\}. 
	\end{gather*}
\end{lemma}
Therefore, $\sw^{\infty}(c)$ is the switching free boundary.

\begin{lemma}\label{inftylemma:xc}
	We have
	$\sw^{\infty}(c)>0$ for all $c\in[0,+\infty).$
\end{lemma}
\begin{proof}
	Denote $g^{\cc}(x)$ as the terminal function which is defined in \eqref{g_def} for fixed $\cc>0$.
	For $x>0$ and $c>0$, let $\cc=\max\{c,\;2 r/\ga\}$, since $v(0,c)=g^{\cc}(0)=0$, $v$ is non increasing w.r.t $c$ and $v^{\cc}\leq v$,
	it follows that
	\[v(x,c)-v(0,c)=v(x,c)\geq v(x,\cc)\geq v^{\cc}(x,\cc)=g^{\cc}(x)-g^{\cc}(0).\]
	From this, we can derive
	\begin{align}\label{inftyvvx0>1}
		v_x(0,c)\geq (g^{\cc})'(0)>1
	\end{align}
	by recalling that $\cc \ga > r$.
	Comparing to \eqref{inftyvvxxc}, we conclude that $\sw^{\infty}(c)>0$.
\end{proof}

We are going to prove the continuity of $\sw^{\infty}(c)$. 
Firstly, since the constant $K$ in \eqref{vix_ub} and \eqref{ui_b} are independent of $\cc$, similar to \lemref{lem:uci_b}, we can prove 
\begin{lemma}\label{inftylem:uuci_b}
	Suppose $v^{\cc,n}_i$, $i=0,1,\cdots,n$ is the sequence of discrete approximating functions of $v^{\cc}$ constructed in \secref{sec:approximation}, which solves \eqref{vi_pb}, and $x^{\cc,n}_i$ is the sequence of free boundary points. Let $\Dc^{\cc,n}=\cc/n$, $c^{\cc,n}_i=\cc- i \Dc^{\cc,n}$ and
	$$u^{\cc,n}_i:=\frac{v^{\cc,n}_i-v^{\cc,n}_{i-1}}{\Dc^{\cc,n} },~~ i=1,2,\cdots,n.$$
	For any fixed $c^*>0$, if there exists a constant $a>0$ such that $x^{\cc,n}_i\geq a$ for all $c^{\cc,n}_i\leq c^*$,
	we then have
	\begin{align}\label{inftyui_b2}
		u^{\cc,n}_{i-1}\leq u^{\cc,n}_i+B(a,c^*) \Dc^{\cc,n}~~\text{if}~c^{\cc,n}_i\leq c^*,
	\end{align}
	where $B(a,c^*)>0$ is independent of $n$ and $\cc$.
\end{lemma}

Based on \lemref{inftylem:uuci_b}, we have
\begin{lemma}\label{inftylem:uucL}
	If $0\leq c_*<c_*+\Dc_*\leq c^*-\Dc^*<c^*<+\infty$, then there exists a constant $B(c^*)>0$ which only depends on $c^*$ such that
	\begin{align}\label{inftyuucL}
		\frac{v(x,c^*-\Dc^*)-v(x,c^*)}{\Dc^*}\leq \frac{v(x,c_*)-v(x,c_*+\Dc_*)}{\Dc_*}+B(c^*)(c^*-c_*).
	\end{align}
\end{lemma}
\begin{proof}
	Suppose $v^{\cc,n}_i$, $i=0,1,\cdots,n$ is the sequence of discrete approximating functions of $v^{\cc}$ constructed in \secref{sec:approximation}, which solves \eqref{vi_pb}, and $x^{\cc,n}_i$ is the sequence of free boundary points. 
	Let $\Dc^{\cc,n}=\cc/n$, $c^{\cc,n}_i=\cc- i \Dc^{\cc,n}$. 
	And let $\{v^{\cc,n^{\cc}_k}\}$ be the linear interpolation of $v^{\cc,n^{\cc}_k}_i$, $i=0,1,\cdots,n^{\cc}_k$ that converges to $v^{\cc}$. 
	
	We first prove that there exists $a(c^*) > 0$ which is independent of $\cc$ such that
	\begin{align}\label{inftyxxnk}
		x^{\cc,n^{\cc}_k}_i\geq a(c^*)~~\text{if}~c^{\cc,n^{\cc}_k}_i\leq c^*.
	\end{align}
	Due to \eqref{xnk}, for each $\cc>0$, there exists a constant $a_{\cc}>0$ such that
	$x^{\cc,n^{\cc}_k}_i\geq a_{\cc}$ if $c^{\cc,n^{\cc}_k}_i\leq c^*$. It suffices to prove this $a_{\cc}$ is bounded below by a constant $a(c^*) > 0$ for every $\cc>c^*$. 
	Conversely, assume that there exists a subsequence $\{y_m\}$, where $y_m := x^{\cc_m,n^{\cc_m}_{k_m}}_{i_m}$, which converges to $0$ as $m \to \infty$. By virtue of \eqref{vjx<=1}, we have $(v^{\cc_m,n^{\cc_m}_{k_m}}_{i_m})'(x)\leq 1$ for all $x\geq y_m$.
	Denote $s_m = c^{\cc_m,n^{\cc_m}_{k_m}}_{i_m}$. Given that $s_m$ is bounded within the interval $[0,c^*]$, there must exist an $s_0 \in [0,c^*]$ and a subsequence of $\{s_m\}$ (which we still denote by $\{s_m\}$ for simplicity) such that $s_m \to s_0$. 
	Then there exists a subsequence of the functions $v^{\cc_m,n^{\cc_m}_{k_m}} (\cdot,s_m)$ that converges to $v(\cdot,s_0)$. Consequently, we can infer that $v_x(0,s_0)\leq 1$, which stands in contradiction to \eqref{inftyvvx0>1}. So \eqref{inftyxxnk} holds.
	The remaining proof can be completed following the procedure of \lemref{lem:ucL}.
\end{proof}

By adhering strictly to the proof of \lemref{lemma:xcc}, we can prove 
\begin{lemma}\label{inftylemma:xxcc}
	The switching boundary $\sw^{\infty}(\cdot)$ is continuous in $[0,+\infty)$.
\end{lemma}

Now, we are going to prove $ \sw^{\infty}(\cdot) $ is bounded in $ [0, +\infty) $. For this, we first give

\begin{lemma}\label{inftylem:vbc}
	For each $c>0$, there exists
	\begin{align}\label{inftyvbc}
		v(x,c)\leq x+\frac{1}{c}\Big(\frac{\mu^2}{rb}+\frac{\mu}{b} x\Big).
	\end{align}
\end{lemma}
\begin{proof}
	Let $\psi=x+\frac{1}{c}(\frac{\mu^2}{rb}+\frac{\mu}{b} x)$, then we have
	$$
	-\LL \psi - c \TT \psi = -\mu \Big(1+\frac{\mu}{c b}\Big)+r \Big[x+\frac{1}{c}\Big(\frac{\mu^2}{rb}+\frac{\mu}{b} x\Big)\Big] + c b \frac{\mu}{c b} = r\Big(1+\frac{\mu}{c b}\Big) x \geq 0
	$$
	and $-\psi_c\geq 0$. Together with $\psi(0,c)\geq 0$ and $\lim\limits_{c\to +\infty}\psi(x,c)=x$, we see $\psi$ is a super solution of \eqref{inftyvv_pb} and get \eqref{inftyvbc}.
\end{proof}

\begin{lemma}\label{inftylem:xb} The switching boundary
	$ \sw^{\infty}(c) $ is bounded in $ [0, +\infty) $, specifically, \eqref{inftyxl_ub} holds. 
\end{lemma}
\begin{proof}
	Suppose, for contradiction, there exist $c_n\to +\infty$ and $x_0>\frac{\mu}{r b} \left(1 + \sqrt{1 + 2b - 2b^2}\right)$ such that $\sw^{\infty}(c_n) \geq x_0$. 
	Then the equation of $v(\cdot,c_n)$ holds in $(0,x_0)$. So we have 
	\begin{align}\label{inftyinteq}
		-\int_0^{x_0} \LL v(x,c_n) \d x = c_n \int_0^{x_0}\TT v(x,c_n) \d x.
	\end{align}
	Denote $$x_n=\inf\{x\in [0,x_0]\;|\;v_x(x,c_n)\leq 1\}.$$ Then we have $v_x(x,c_n)> 1$ for all $x< x_n$ and by \eqref{inftyvvxx} we have $v_x(x,c_n)\leq 1$ for all $x\geq x_n$. 
	Moreover, noting that $v(0,c_n)=0$, so we have
	$$
	c_n \int_0^{x_0}\TT v(x,c_n) \d x= c_n b \Big(x_0-v(x_0,c_n)\Big) + c_n (1-b) \Big[\Big(x_0-v(x_0,c_n)\Big)-\Big(x_n-v(x_n,c_n)\Big)\Big].
	$$ 
	Due to the first inequality in \eqref{inftyvv} and the estimate \eqref{inftyvbc}, we conclude 
	\begin{align}\label{inftyinteq1}
		c_n\int_0^{x_0}\TT v(x,c_n) \d x \leq (1-b) \Big(\frac{\mu^2}{rb}+\frac{\mu}{b} x_n\Big) \leq (1-b) \Big(\frac{\mu^2}{rb}+\frac{\mu}{b} x_0\Big).
	\end{align}
	On the other hand, by \eqref{inftyvvx0>1} and \eqref{inftyvxx} we have $v_x(x_0,c_n) \leq \max\{v_x(0,c_n),\;1\} = v_x(0,c_n)$. So
	\begin{align}
		-\int_0^{x_0}\LL v(x,c_n) \dx &= - \frac{\si^2}{2} [v_x(x_0,c_n)-v_x(0,c_n)] - \mu v(x_0,c_n) + r \int_0^{x_0} v(x,c_n) \d x\nonumber\\\label{inftyinteq2}
		&\geq - \mu v(x_0,c_n) + r \int_0^{x_0} v(x,c_n) \dx.
	\end{align}
	Together with \eqref{inftyinteq}, \eqref{inftyinteq1} and \eqref{inftyinteq2} we have
	\begin{align*}
		-\mu v(x_0,c_n) + r \int_0^{x_0} v(x,c_n) \d x \leq (1-b) \Big(\frac{\mu^2}{rb}+\frac{\mu}{b} x_0\Big).
	\end{align*}
	Since $v(x,c_n) \to x$, letting $n\to +\infty$ yields 
	\begin{align*}
		-\mu x_0 + r \frac{x_0^2}{2} \leq (1-b) \Big(\frac{\mu^2}{rb}+\frac{\mu}{b} x_0\Big).
	\end{align*}
	Solving this gives
	\[
	x_0 \leq \frac{\mu}{r b} \left(1 + \sqrt{1 + 2b - 2b^2}\right),
	\]
	contradicting to the assumption. So \eqref{inftyxl_ub} holds.
\end{proof}

\subsection{The properties of the converting boundary} \label{sec:converting}

\begin{lemma}\label{inftylem:yyc}
	The converting boundary $\sd^{\infty}(\cdot)$ is continuous in $[0,+\infty)$ and satisfies
	$$
	0< \sd^{\infty}(c)< \sw^{\infty}(c),~~\forall ~ c>0.
	$$ 
\end{lemma}
\begin{proof}
	Using the method in \lemref{lem:yc}, we can prove $\sd^{\infty}(\cdot)$ is continuous in $[0,+\infty)$ and satisfies
	$$0< \sd^{\infty}(c)\leq \sw^{\infty}(c),~~\forall ~ c>0.$$ 
	
	We now prove the equality cannot hold. 
	If $\swb^{\infty}(\sw^{\infty}(c),c)<+\infty$, the conclusion follows from a discussion analogous to that in \lemref{lem:y<x}. 
	Now, consider the case $\swb^{\infty}(\sw^{\infty}(c),c)=+\infty$. 
	For contradiction, suppose $\sd^{\infty}(c)=\sw^{\infty}(c)$. By the definition of $\sd^{\infty}(c)$, $v_x(x,c)>1$ for all $x\in(0, \sd^{\infty}(c)).$ Then by integration, we derive $v(\sw^{\infty}(c),c)> \sw^{\infty}(c)$. 
	However, since $\swb^{\infty}(\sw^{\infty}(c),c)=+\infty$, we know $v(\sw^{\infty}(c),c)=\lim\limits_{s\to+\infty} v(\sw^{\infty}(c),s)=\sw^{\infty}(c)$. This forms a contradiction.
\end{proof}

\begin{lemma}\label{inftylem:yl}
	The converting boundary $\sd^{\infty}(\cdot)$ satisfies \eqref{inftyyl}. 
\end{lemma}
\begin{proof}
	If $\limsup\limits_{c\to +\infty}\sd^{\infty}(c)>\frac{\mu}{r}$, there exist $c_n\to +\infty$ and $x_0>\frac{\mu}{r}$ such that $\sd^{\infty}(c_n) \geq x_0$. It follows that $v(x,c_n)\geq 1$ for $x\leq x_0$. Then by the equation of $v(x,c_n)$ we have
	$$
	-\frac{\si^2}{2} v_{xx}(x,c_n) -\mu v_x(x,c_n) + r v(x,c_n) = c_n b (1-v_x(x,c_n)) \leq 0,~~x\in (0,x_0).
	$$
	Letting $n\to +\infty$, since $v(x,c_n) \to x$ we have
	$-\mu + r x \leq 0$ for all $x\in (0,x_0).$
	This contradicts $x_0>\frac{\mu}{r}$. So we have $\limsup\limits_{c\to +\infty}\sd^{\infty}(c)\leq \frac{\mu}{r}$. Similarly, we can prove $\liminf\limits_{c\to +\infty}\sd^{\infty}(c)\geq \frac{\mu}{r}$. Therefore, \eqref{inftyyl} holds.
\end{proof}

\end{appendices}

\end{document}